\documentclass[AMA,STIX1COL]{WileyNJD-v2}
\usepackage{amsmath,amsfonts,amssymb}
\usepackage{graphicx}
\usepackage{todonotes}
\usepackage{hyperref}

\usepackage{colortbl}

\usepackage{tikz} 
\usetikzlibrary{shapes,arrows,calendar,matrix,backgrounds,folding,calc,positioning,spy,pgfplots.groupplots}
\usepackage{pgfplots} 
\usepackage{pgfplotstable} 

\usepackage{siunitx}
\DeclareSIUnit[number-unit-product = {}]\Q{~}
\DeclareSIPrefix\kilo{K}{3}
\DeclareSIPrefix\mega{M}{6}
\DeclareSIPrefix\giga{}{9}
\DeclareSIPrefix\terra{}{12}

\newcommand{\numQ}[1]{\SI[zero-decimal-to-integer,scientific-notation = engineering,round-precision=1,exponent-to-prefix = true]{#1}{\Q}\!\!\!}
\newcommand{\numQb}[1]{{\color{green!50!black}\SI[zero-decimal-to-integer,scientific-notation = engineering,round-precision=1,exponent-to-prefix = true]{#1}{\Q}\!\!\!}}
\newcommand{\numQQ}[1]{\SI[zero-decimal-to-integer,scientific-notation = engineering,round-precision=1,exponent-to-prefix = true]{#1}{\Q}\!\!\!}
\newcommand{\numQQb}[1]{{\color{green!50!black}\SI[zero-decimal-to-integer,scientific-notation = engineering,round-precision=1,exponent-to-prefix = true]{#1}{\Q}\!\!\!}}

\usepackage{multirow}

\usepackage{bm}

\newcommand{\rr}{{\mathbb{R}}}
\renewcommand{\t}{\boldsymbol{\tau}}
\newcommand{\n}{\boldsymbol{n}}
\newcommand{\nG}{\boldsymbol{\mu}}

\newcommand{\tr}{\operatorname{tr}}
\newcommand{\sgn}{\operatorname{sgn}}
\newcommand{\eps}{\varepsilon}

\newcommand{\dx}{\mathop{\mathrm{d}x}}
\newcommand{\ds}{\mathop{\mathrm{d}s}}
\newcommand{\dshat}{\mathop{\mathrm{d}\hat s}}
\newcommand{\divergence}{\operatorname{div}}
\newcommand{\curl}{\operatorname{curl}}
\newcommand{\mesh}{\mathcal{T}_h}
\newcommand{\facets}{\mathcal{F}_h}

\articletype{Article Type}%

\received{26 April 2016}
\revised{6 June 2016}
\accepted{6 June 2016}

\begin{document}
\title{Divergence-free tangential finite element methods for incompressible flows on surfaces
}

\author[1]{Philip L. Lederer}

\author[2]{Christoph Lehrenfeld}

\author[1]{Joachim Sch\"oberl}


\address[1]{\orgdiv{Institute for Analysis and Scientific Computing}, \orgname{TU Wien}, \orgaddress{\country{Austria}},
  Email:  joachim.schoeberl@tuwien.ac.at (\url{https://orcid.org/0000-0002-1250-5087
}), 
  philip.lederer@tuwien.ac.at (\url{https://orcid.org/0000-0003-1875-7442})
}

\address[2]{\orgdiv{Institute for Numerical and Applied Mathematics}, \orgname{University of G\"ottingen}, \orgaddress{\country{Germany}}, 
	Email:
		lehrenfeld@math.uni-goettingen.de (\url{https://orcid.org/0000-0003-0170-8468})
}

\corres{
	Christoph Lehrenfeld, Institute for Numerical and Applied Mathematics, Georg-August-Universit{\"a}t G{\"o}ttingen,
	37083 G{\"o}ttingen, Germany.
        Email: lehrenfeld@math.uni-goettingen.de
	}


\abstract[Summary]{
  In this work we consider the numerical solution of incompressible flows on two-dimensional manifolds. Whereas the compatibility demands of the velocity and the pressure spaces are known from the flat case one further has to deal with the approximation of a velocity field that lies only in the tangential space of the given geometry.
  Abandoning $H^1$-conformity allows us to construct finite elements which are -- due to an application of the Piola transformation -- exactly tangential. To reintroduce continuity (in a weak sense) we make use of (hybrid) discontinuous Galerkin techniques. To further improve this approach, $H(\divergence_{\Gamma})$-conforming finite elements can be used to obtain exactly divergence-free velocity solutions. 
  We present several new finite element discretizations.
On a number of numerical examples we examine and compare their qualitative properties and accuracy.
%
%
  
%
}

\keywords{Surface PDEs, tangential vector field, incompressible Navier--Stokes equations, Piola transformation, divergence-conforming finite elements}

\maketitle

\section{Introduction}
Partial Differential Equations (PDEs) that are posed on curved surfaces play an important role in several applications in engineering, physics and mathematics. Surface PDEs describing flows on surfaces appear for instance in the modeling of emulsions, foams and biological membranes, cf. Slattery et al.\cite{slattery2007interfacial} or Brenner\cite{brenner2013interfacial}, or liquid crystals, cf. de Gennes and Prost\cite{de1993physics} or Napoli and Vergori\cite{napoli2016hydrodynamic}. The numerical treatment of these PDEs gained an increasing amount of attention in the field of numerical simulations and numerical analysis in the last two decades.
In this work we consider vector valued PDEs for viscous incompressible flows on surfaces that are immersed in the three dimensional space. 

A main source of difficulty for vector valued PDEs on surfaces is the fact that the unknown vector field is typically \emph{tangential}, i.e. for a two dimensional manifold $\Gamma$ embedded in $\mathbb{R}^3$ the unknown field is only two-dimensional.
In recent works almost exclusively $[H^1(\Gamma)]^3$-conforming finite elements have been used to approximate the
unknown tangential vector field. Tangential vector fields are only weakly imposed through the variational formulation. In this work we follow a different approach: we abandon continuity of the finite elements. This loss of conformity however allows us to construct exactly tangential vector fields.
This is achieved by mapping finite element functions from the two-dimensional reference element by a straight-forward generalization of the well-known \emph{Piola transformation}. This guarantees that the resulting (possibly higher order) vectorial basis functions are \emph{exactly tangential} to the surface. This specifically means that no additional enforcement of the tangentiality condition is needed to be enforced through the variational formulation. One could say that we trade one structure property (continuity) for the other (tangential vector fields). To deal with the missing continuity we apply well established techniques from the flat case: discontinuous Galerkin (DG) methods and variants such as the hybrid DG (HDG) methods.

It turns out that we do not have to abandon continuity completely, but can preserve continuity at least for the co-normal component, i.e. the in-plane normal component across element interfaces, resulting in $H(\divergence_{\Gamma})$-conforming finite elements. These finite elements in conjuction with suitable (hybrid) DG techniques have been proven to be excellent discretizations for incompressible fluid flows in the flat case due to benefitial properties such as exactly divergence-free solutions, pressure robustness and energy stability. These properties can easily be transfered to the case of surface Stokes and surface Navier--Stokes equations as we will explain in the sequel of this paper. 

\subsection{State of the art}
Let us briefly give an overview on the state of the art in the literature. 
Initially surface finite element methods (SFEM) for \emph{scalar} PDEs have been introduced in the seminal work by Dziuk\cite{dziuk1988finite}, we also refer to the survey papers by Dziuk and Elliott\cite{dziuk_elliott_2013} and Bonito et al. \cite{2019arXiv190602786B}. The use and analysis of $H^1(\Gamma)$-conforming surface finite elements has been extended to higher order discretizations and adaptive schemes by Demlow et al. \cite{demlow2007adaptive,demlow2009higher}, including the analysis of geometry errors\cite{holst2012geometric}. In the last decade the extension to non-($H^1(\Gamma)$-)conforming surface finite elements has been done in several works, cf. for instance Refs. \cite{dedner2013analysis,antonietti2014high} and eventually also HDG formulations have been considered by Cockburn and Demlow\cite{cockburn2016hybridizable}.
In all these works an explicit surface mesh is used. A different approach is based on a mesh of the surrounding 3D space and a level set description of the surface. In the ``TraceFEM'' the trace on the level set surface of finite element functions of this background mesh are used for the approximation of the solution. The method was introduced by Olshanskii et al. \cite{olshanskii2009finite}. Sometimes the method is also known as ``CutFEM'', cf. e.g. Ref.\cite{burman2015stabilized}. We also refer to the overview paper on TraceFEM by Olshanskii and Reusken\cite{OlRe_2017}.

Mixed formulations of the surface Poisson problem can be seen as an intermediate step towards vector valued PDEs as they involve the vector valued surface flux that is approximated seperately from the primal unknown, resulting in a system of first order surface PDEs. These mixed formulations have been considered for instance by Rognes et al.\cite{rognes2013automating} and -- as part of their mixed formulations of DG and HDG methods -- by Antonietti et al. \cite{antonietti2014high} and Cockburn and Demlow\cite{cockburn2016hybridizable}. Here, tangential surface finite elements are constructed for the flux. Primal DG formulations in the context of TraceFEM/CutFEM have also been considered in \cite{burman2017cut}.
In the works by Rognes et al. \cite{rognes2013automating} and Cockburn et al. \cite{cockburn2016hybridizable} the construction of the spaces is based on the Piola transformation, resulting in (broken) surface Raviart-Thomas spaces as considered already by Nedelec \cite{nedelec1978computation} and Bendali \cite{bendali1984numerical}.
Let us note that the analysis of mixed Poisson formulations including variational crimes in a general framework has been considered in Holst and Stern\cite{holst2012geometric}. In all these works, where the primal unknown is scalar, an isoparametric geometry approximation, i.e. using order $k$ for approximating the scalar unknown and order $k$ for the geometry approximation is sufficient to obtain order $k+1$ error estimates in the $L^2$-norm. Only for the superconvergence property in the HDG method by Cockburn and Demlow\cite{cockburn2016hybridizable} an increased geometry order is necessary.

While scalar problems on surfaces and their numerical treatment seem to be well understood vector valued problems have drawn an increased interest recently.
Relevant models of viscous fluidic surfaces based on 3D Cartesian differential operators are described in Refs. \cite{MR3875687,koba2017energetic}. 
In the context of finite element methods these PDEs require vector valued finite element spaces on surfaces that are tangential.
Starting with the Vector Laplacian surface finite elements that are  $[H^1(\Gamma)]^3$-conforming, i.e. three dimensional, and impose the tangential condition through the variational formulation have been proposed by Hansbo et al.\cite{hansbo2016analysis}. In that paper a penalty formulation and a Lagrange multiplier based formulation are considered to drive the normal component of the discrete approximation to zero. Further, it is already observed that -- in contrast to the scalar problem -- an isoparametric discretization is not sufficient to preserve optimal order $L^2$-errors. In Refs. \cite{MR3840893,jankuhn2019trace} similar approaches are considered and analysed in the context of TraceFEM discretizations. 
Hansbo  et al. \cite{hansbo2017stabilized} extended their approach to Darcy problems on surfaces using $[H^1(\Gamma)]^3$-conforming (low order) surface FEM. 
The surface Stokes problem based on a velocity-pressure formulation has been discretized using TraceFEM based on a stabilized P1-P1 discretization in Ref. \cite{Olshanskii2018AFE}.
In the very recent paper by Olshanskii et al.\cite{2019arXiv190902990O} the TraceFEM with a P2-P1 discretization has been considered. Also only very recently Bonito et al.\cite{bonito2019} presented a low order $H(\divergence_{\Gamma})$-conforming discretization for the surface Stokes problem on $C^4$ smooth closed surfaces with a focus on the numerical treatment of \emph{killing fields}. Their approach to discretize the surface Stokes problem is similar to the discretizations that we treat for the surface Navier--Stokes equations in this paper.
Similar methods to Hansbo et al. \cite{hansbo2016analysis} have been proposed in Nestler et al.\cite{NNPV18,NNV19}, where vector- and tensor-valued surface PDE models are considered.

A vorticity formulation has been considered for the surface Stokes problem in 
Refs. \cite{RV15,NRV17,reusken2018stream} where the explicit construction of tangential vector fields is circumvented.
Vorticity formulations have also been used to solve the surface Navier--Stokes equations in Refs. \cite{Azencot2014,nitschke2012finite}. 
Velocity-pressure formulations for the surface Navier--Stokes problem have been recently considered using higher order $[H^1(\Gamma)]^3$-conforming surface FEM by Reuther and Voigt \cite{voigtreuther18} (low order) and by Fries\cite{MR3846120} (higher order) and based on a low order TraceFEM discretization with penalty by Olshanskii and Yushutin\cite{Olshanskii2018APF}.
We also mention the numerical approaches based on discrete exterior calculus in \cite{NRV17} and spectral methods in \cite{GA18}.

In all the previous works either closed smooth surfaces or at least smooth surfaces with boundaries, e.g. in Fries\cite{MR3846120}, have been considered. The case of only piecewise smooth geometries has not been addressed in the finite element literature so far to the best of our knowledge.

As we will base our discretization for the surface Navier--Stokes equations on $H(\divergence_{\Gamma})$-conforming elements, let us also briefly mention previous works on $H(\divergence_{\Gamma})$-conforming methods in the plane. In the context of DG discretizations Cockburn et al.\cite{cockburn2007note} were the first to realize that energy stability and local mass conservation for DG methods is only achieved for $H(\divergence_{\Gamma})$-conforming finite elements which result in pointwise divergence-free solutions. We extended this idea to HDG methods and considered several extensions and improvements and evaluated the computational efficiency of the resulting methods in  Refs.\cite{lehrenfeld2010hybrid,LS_CMAME_2016,LedererSchoeberl2017,LLS_SIAM_2017,LLS_ESAIM_2019,SLLL_SEMA_2018,SJLLLS_CAMWA_2019}, cf. also the discussion in Section~\ref{sec::HdivBenefits} below.

\subsection{Content and structure}
In this paper we introduce non-($H^1(\Gamma)$-)conforming finite elements for incompressible surface flows, starting from the Vector Laplacian to the unsteady surface Navier--Stokes equations. We present different DG and HDG discretizations -- most of which are new -- and compare them to $[H^1(\Gamma)]^3$-conforming methods. The use of $H(\divergence_{\Gamma})$-conforming finite element methods results in pointwise divergence-free and exactly tangential solutions. Our methods are high order accurate but allow for surfaces which are only piecewise smooth. Moreover, arbitrary surface meshes can be dealed with, i.e. an exact geometry description does not need to be known.

Model problems are presented in Section \ref{sec::modelprob} before several numerical schemes with tangential finite elements are introduced in Section \ref{sec::FEMconstruction}. Several numerical examples for the different problems and discretizations are presented and discussed in Section \ref{sec::numex} before we conclude the manuscript.

  
  

  

\section{Model problems on the surface} \label{sec::modelprob}
\subsection{Notation and surface differential operators}
Let $\Gamma$ be a sufficiently smooth connected two-dimensional stationary and oriented surface embedded in $\mathbb{R}^3$. At every point $x \in \Gamma$ we denote by $\n(x)$ a uniquely oriented unit normal vector, and by $P(x) = I - \n(x) \n(x)^T$, with the identity matrix $I$, the corresponding orthogonal projection onto the tangential plane of $\Gamma$ at $x$. In this work we assume that there exists a sufficiently smooth extension into the neighborhood $\mathcal{O}(\Gamma)$ of $\Gamma$ which induces a projection $p: \mathcal{O}(\Gamma) \rightarrow \Gamma$.
Assume a given scalar-valued, sufficiently smooth function $\phi:\Gamma \rightarrow \rr$, and let $\phi^p = \phi \circ p: \mathcal{O}(\Gamma) \rightarrow \rr$ denote its extensions to the neighborhood $\mathcal{O}(\Gamma)$ of $\Gamma$. Then, we define the \textit{scalar surface gradient} through the gradient of $\phi^p$ in the embedding space and the projection onto the tangential plane. Hence, for any $x \in \Gamma$ we have
\begin{align*}
  \nabla_\Gamma \phi(x) :=  P(x) \nabla \phi^p(x),
\end{align*}
where $\nabla \phi^p$ is the column vector consisting of all partial derivatives. As a direct consequence we realize that the scalar surface gradient lies in the tangential plane of $\Gamma$. Further, let us note that the surface gradient is independent of the concrete choice of the projection $p$. In the same manner let $u=(u_1,u_2,u_3)^T: \Gamma \rightarrow \rr^3$ be a given vector-valued and sufficiently smooth function and denote by $u^p: \mathcal{O}(\Gamma) \rightarrow \rr^3$ its extension to the neighborhood. According to the above definition we can define the component-wise surface gradient through $\nabla u^p P$, where $\nabla u^p = (\nabla u^p_1, \nabla u^p_2, \nabla u^p_3)^T$ is the standard Jacobian matrix of $u^p$. Hence, $\nabla u^p P$ is the matrix where each row gives the scalar surface gradient of the components of $u^p$. We define another operator called the \textit{tangential surface gradient} by applying an additional projection from the left
\begin{align*}
  \nabla_\Gamma u(x) := P(x) \nabla u^p(x) P(x) \quad \forall x \in \Gamma.
\end{align*}
Note, that in the literature this operator is usually known as the covariant derivative on $\Gamma$.
We are now able to state the symmetric surface strain tensor which is -- following Gurtin and Murdoch\cite{Gurtin1975} -- defined as
\begin{equation}\label{def:surfacestrain}
  \varepsilon_\Gamma(u) := \frac{1}{2} (\nabla_\Gamma u + \nabla_\Gamma u ^T), 
\end{equation}
and the surface divergence operator
\begin{align*}
    \divergence_\Gamma u &:= \tr(\nabla_\Gamma u). 
\end{align*}
As usual the divergence operator $\divergence_\Gamma$ applied on a matrix valued function $\sigma$ reads as the row-wise divergence.

So far we assumed that $\phi$ and $u$ are sufficiently smooth so that the above differential operators exist in a point-wise sense on $\Gamma$. We can generalize to the notion of weak derivatives in the usual sense. For instance we define the weak gradient $g^\phi \in [L^2(\Gamma)]^3$ (if it exists) of a given function $\phi \in L^2(\Gamma)$,  as the function that fulfills
$
\int_\Gamma g^\phi \cdot v = - \int_\Gamma \phi \divergence_\Gamma v
$
for all $v \in [C^{\infty}_0(\Gamma)]^3$.

In the next three subsections we introduce the surface PDE problems considered in this work. To this end let $f \in [L^2(\Gamma)]^3$ such that $f \cdot \n = 0$ on $\Gamma$ be a given force vector.

\subsection{Vector-valued elliptic problem on $\Gamma$:} 
We seek for a solution $u: \Gamma \rightarrow \rr^3$ with $u \cdot \n = 0$ on $\Gamma$ of the second order partial differential equation given by
\begin{subequations} \label{prob:veclap}
\begin{align} \label{prob:veclap-a}
  - P \divergence_\Gamma( \eps_\Gamma(u)) + u &= f \quad \textrm{in } \Gamma, \\
\label{prob:veclap-b}  
  u &= 0 \quad \textrm{on } \partial \Gamma.
\end{align}
\end{subequations}
For the ease of presentation we only consider homogeneous Dirichlet boundary conditions in this part but the introduced methods can all be extended to more general boundary conditions, as demonstrated in the numerical examples. Further note that in the case of a closed surface ($\partial \Gamma = \emptyset$) no boundary conditions are prescribed. Let $V:=\{ u \in [H^1(\Gamma)]^3: u = 0 \textrm{ on } \partial \Gamma \}$ with the norm $\| u \|_1^2:= \| \nabla_\Gamma u \|^2_{L^2(\Gamma)} + \| u \|^2_{L^2(\Gamma)}$, and let
\begin{align*}
  V_{\t}(\Gamma):= \{ u \in V: u \cdot \n = 0 \textrm{ on } \Gamma \}.
\end{align*}
Following Ref.~\cite{MR3840893}, the variational formulation of \eqref{prob:veclap} is given by: Find $u \in V_{\t}$ such that
\begin{align} \label{eq:weak:veclap}
a(u,v) + m(u,v) = f(v)   \quad \forall v \in V_{\t},
\end{align}
where
\begin{align*}
a(u,v) := \int_\Gamma  \varepsilon_\Gamma(u): \varepsilon_\Gamma(v) \dx, \quad m(u,v) := \int_\Gamma u \cdot v \dx, \quad f(v) :=\int_\Gamma f \cdot v \dx.
\end{align*}
While finite element discretization of such variational problems are well understood on a flat surface, the tangential constraint $u \cdot \n =0$ on $\Gamma$ makes the construction of an appropriate numerical scheme on surfaces much more difficult. In Section~\ref{sec::FEMconstruction} we discuss a natural approach how to deal with this challenge.

\begin{remark}
  Note that the above variational formulation can also be defined on piecewise smooth connected surfaces. Then the strong form of the partial differential equation is given by equation \eqref{prob:veclap} defined on each (smooth) sub domain and continuity conditions of the trace and the normal fluxes at the common interfaces. In Section \ref{sec::houseofcards} we demonstrate that our methods are applicable for such problems.
\end{remark}
 \subsection{Stationary Stokes equations on $\Gamma$:} We consider a Newtonian fluid on $\Gamma$, see Refs. \cite{MR3875687,MR3614501}, and assume for the ease of representation that $\partial \Gamma \neq \emptyset$, see Remark~\ref{stokesbnd}. Adding the incompressibility constraint and the pressure as Lagrange multiplier we now seek for a solution $u: \Gamma \rightarrow \rr^3$ with $u \cdot \n = 0$ on $\Gamma$ and $p: \Gamma \rightarrow \rr$ such that
 \begin{subequations}\label{prob:stokes}
\begin{align} \label{prob:stokes:momentum}
  - 2 \nu P \divergence_\Gamma( \eps_\Gamma(u)) + \nabla_\Gamma p &= f\hphantom{0} \quad \text{in } \Gamma, \\
  \operatorname{div}_\Gamma (u) & = 0\hphantom{f}  \quad \text{in } \Gamma, \\
  u &= 0\hphantom{f} \quad \text{on } \partial \Gamma.
\end{align}
\end{subequations}
Here, \eqref{prob:stokes:momentum} can also be read as $- P \divergence_\Gamma \sigma(u) = f$ with the stress tensor
$\sigma(u,p) = -p P + 2 \nu \eps_\Gamma(u)$ where $\nu$ is the kinematic viscosity.
Following the derivation in Ref.~\cite{MR3846120} the weak formulation is given by
\begin{alignat*}{2}
2 \nu a(u,v) + b(v,p) &= f(v)   \quad &\forall v \in V_{\t}, \\
b(u,q) &=0    \quad &\forall q \in Q,
\end{alignat*}
with
$  b(u,q) = \int_\Gamma \divergence_\Gamma(u) q \dx$,
and the pressure space $ Q := \{ q \in L^2(\Gamma): \int_\Gamma q = 0 \}$. Beside the crucial constraint that $u$ is in the tangential space of $\Gamma$ we now further have to deal with the divergence constraint. In particular this plays a key role in the discretization as one has to find compatible function spaces for the discrete velocity and the pressure spaces.

\begin{remark}\label{stokesbnd}
  As the variational formulation of the standard stationary (flat) Stokes equations is defined without a mass bilinear form $m(u,v)$ it demands further constraints to filter out the possibly non-trivial kernel of the composite bilinear form involving $a(\cdot,\cdot)$ and $b(\cdot,\cdot)$. On surfaces with a boundary sufficient constraints can be obtained from suitable boundary conditions. However, for the surface Stokes equations on closed surfaces the non-trivial kernel, the so-called \emph{killing fields}, need to be taken care of to guarantee uniqueness. For a more detailed discussion we refer to Refs.\cite{bonito2019,voigtreuther18}.
\end{remark}

\subsection{Unsteady Navier--Stokes equations on $\Gamma$:} We extend the Stokes model to the fully unsteady Navier--Stokes case, i.e. we seek for a solution $u: \Gamma \times (0,T] \rightarrow \rr^3$ with $u \cdot \n = 0$ on $\Gamma$ and $p: \Gamma \times (0,T]  \rightarrow \rr$ such that
\begin{subequations} \label{prob:navstokes}
\begin{align} 
\partial_t u - 2 \nu P \operatorname{div}_\Gamma (\varepsilon_\Gamma(u)) + (u \cdot \nabla_\Gamma)u + \nabla_\Gamma p & = f \hphantom{0 u_0 g}\text{ on } \hphantom{\partial}\Gamma, ~ t \in (0,T], \\
  \divergence_\Gamma(u) & = 0 \hphantom{f u_0 g}\text{ on } \hphantom{\partial}\Gamma, ~ t \in (0,T], \\
  \mathcal{B}(u) &= g \hphantom{0 f u_0}\text{ on } \partial \Gamma, ~ t \in (0,T], \\
  u &= u_0 \hphantom{0 f g}\text{ on } \hphantom{\partial}\Gamma \times \{ 0 \},
\end{align}
\end{subequations}
with $T>0$, a given boundary differential operator $\mathcal{B}$, boundary values $g$ and initial values $u_0$. Beside the difficulties already discussed in the Stokes setting, the challenging part now is to deal with the nonlinear convection and the time derivative.

\section{Construction of tangential finite element methods} \label{sec::FEMconstruction}
In this section we focus on the derivation of new tangential DG schemes for the problems introduced in Section~\ref{sec::modelprob}. After introducing some basic notation in Section~\ref{sec::basicingandmot}, we concentrate on rather standard DG versions first in Sections~\ref{sec::discveclap}--\ref{sec::discnavstokes}. Benefits of choosing a $H(\divergence_{\Gamma})$-conforming formulations are discussed in Section~\ref{sec::HdivBenefits}. In Section~\ref{sec::DGvsHDG} we introduce similar HDG formulations and explain their computational advantage over the standard DG formulations.
\subsection{Basic ingredients and motivation}\label{sec::basicingandmot}
In this section we discuss a natural approach for the construction of tangential vector fields. To this end we summarize the findings of the works\cite{rognes2013automating,CASTRO2016241}. Let $\mathcal{T}_h$ be a consistent triangulation of the smooth manifold $\Gamma$ into curved triangles such that for every element $T \in \mesh$ there exists a not degenerated polynomial mapping $\Phi_T$ of order $k_g$ from the reference element
\begin{align*}
  \hat T := \{(x,y) \in \rr^2: 0 \le x+y \le 1\}.
\end{align*}
to the physical element $T$, i.e. $  \Phi_T: \hat T \rightarrow T $. With respect to this triangulation we write $\Gamma_h := \cup_{T \in \mesh} T$ for the corresponding locally smooth discrete manifold. In the following we will use the notation $\eps_\Gamma$, $P$ and $\nabla_\Gamma$ also for operations w.r.t. to $\Gamma_h$ (instead of $\Gamma$).
Further, we define the set $\mathcal{F}_h$ as the union of all element interfaces. We assume that $\mesh$ is shape regular and quasi uniform, thus there exists a mesh size $h$ with $h \simeq \operatorname{diam}(T)$ for all $T \in \mesh$. We denote by $F_T = \Phi_T' \in \rr^{3 \times 2}$ the Jacobian of the finite element mapping and remind the reader that the columns of $F_T$ span the tangential space for each point in $T$. In the following, we will drop the subscript $T$ in the transformation $\Phi_T$ and the derived quantities such as the Jacobian $F_T$ unless the association to the element $T$ shall be highlighted. Next, we write
\begin{align*}
  F^{-1} = (F^T F)^{-1} F^T
\end{align*}
for the Moore-Penrose pseudo inverse of the Jacobian $F$, and set $J := \sqrt{\det(F^T F)}$ as the functional determinant. Now let $\hat \phi$ be a differentiable function defined on the reference element $\hat T$ and let $ x = \Phi(\hat x) \in T$ for all points $\hat x \in \hat T$. Using the classical pull back we define a function $\phi$ on $T$ by
\begin{align*}
  \phi(x) = \hat \phi(\hat x).
  \end{align*}
  Following the ideas of Rognes et al.\cite{rognes2013automating}, this pullback allows us to calculate the surface gradient of the function $\phi$ by
  \begin{align*}
    \nabla_\Gamma \phi(x) = F^{-T} \nabla \hat\phi(\hat x).
  \end{align*}
  As $F^{-T} = F (F^T F)^{-T}$, the above equation shows that the gradient $\nabla \phi$ is a linear combination of the two tangent vectors $\t_1 = F e_1$, $\t_2 = F e_2$ and hence lies in the tangent space as expected from differential geometry, see for example \cite{dziuk_elliott_2013}. The above construction allows us to define an $H^1$-conforming finite element space of order $k$ on the surface $\Gamma_h$ by
  \begin{align*}
  S_h^k :=  \{ w \in C^0(\Gamma_h): \forall T \in \mathcal{T}_h ~ \exists \hat w \in \mathbb{P}^k(\hat T) \textrm{ s.t. } w|_T \circ \Phi_T = \hat w \}.
  \end{align*}
where $\mathbb{P}^k(\hat T)$ is the space of polynomials up to degree $k$ on $\hat T$.
Whereas the above technique allows an easy construction of a scalar approximation space, the problems \eqref{prob:veclap}, \eqref{prob:stokes} and \eqref{prob:navstokes} demand for vector valued approximation spaces. In particular we aim for a space that can handle the constraint $u \cdot \n = 0$ in a proper way. Obviously, the simple product space $S_h^k \times S_h^k \times S_h^k$ is not convenient and we need a different construction. The solution to this is given by using the Piola transformation instead of the classical pull back. Originally, thus on flat surfaces, this mapping is used for the construction of $H(\divergence)$-conforming finite element spaces as it preserves the normal components on element interfaces. To this end let $\hat u$ be a vector valued function on $\hat T$, then we define on $T$ the function
  \begin{align} \label{piola}
    u(x) = \mathcal{P}_T(\hat u)(x) := \frac{1}{J} F \hat u(\hat x).
  \end{align}
  Due to the multiplication with $F$, we again see that the constructed vector field $u$ lies in the tangential plane of $T$. This finding is the key argument and motivation for the construction of the numerical schemes in this work.
  As we will explain below, cf. Lemma \ref{lemma:exdivfree}, the factor $1/J$ is important for the construction of $H(\divergence_{\Gamma})$-conforming finite element spaces, i.e. finite element spaces with continuous in-plane normal components. We want to mention that in Ref. \cite{rognes2013automating} and Ref. \cite{CASTRO2016241} the authors already realized that the Piola mapping can be used on surfaces to construct  appropriate  finite element spaces for the approximation of the spaces $H(\divergence_{\Gamma})$ and $H(\curl_\Gamma)$ on $\Gamma_h$. There, the according finite elements on the surface triangulation $\mesh$ are based on the Raviart-Thomas and Brezzi-Douglas-Marini finite element families on the reference element $\hat T$, cf. Refs. \cite{girault2012finite, brezzi1985two,brezzi2012mixed}.

  \subsection{A discontinuous Galerkin discretization for vector valued elliptic problems} \label{sec::discveclap}
Based on the findings from the last section we now motivate a new discontinuous Galerkin (DG) method for the approximation of problem \eqref{prob:veclap}. However, we want to mention that the discretization can also be adjusted to other elliptic problems like the Vector Laplacian without a mass term. To this end we need several operators motivated by their definitions on a flat surface, see Arnold et al.\cite{arnold2002unified}. Let $T_1,T_2 \in \mesh$ be two arbitrary elements with common edge $E = \overline{T_1} \cap \overline{T_2}$. Let $\n_h$ be the oriented unit normal vector on $\Gamma_h$, and let $\t$ be a uniquely oriented tangential vector on $E$ such that $T_1,T_2$ are located on the left and on the right side respectively with respect to the direction of $\t$, see Figure \ref{fig:sketchofedgepatch}. Using these vectors we define the \emph{in-plane}  unit (outer) normal vectors
\begin{align*}
  \nG_1 := \n_h|_{T_1} \times \t \quad \textrm{and} \quad  \nG_2 := - \n_h|_{T_2} \times \t.
\end{align*}
Let us stress here that $\n_h$ will in general be discontinuous across element interfaces so that $\nG_1$ and $\nG_2$ will not be parallel.
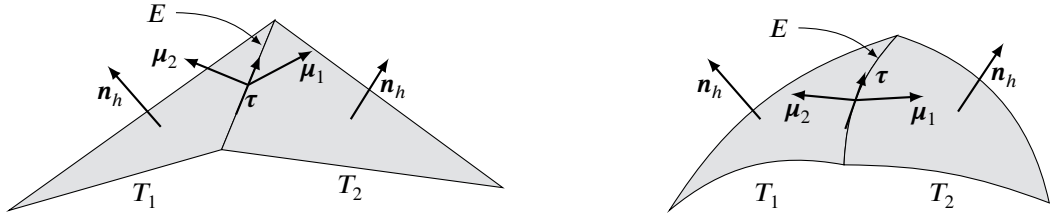
\begin{figure}
  \begin{center}
        \begin{tikzpicture}
  \draw[fill=lightgray!60] (0.5,0) -- (3.3,0.8) -- (4,2.5) -- cycle;
  \draw[fill=lightgray!60] (3.3,0.8) -- (7,0.3) -- (4,2.5) -- cycle;  
  \draw[->, shorten >=0.5cm,shorten <=0.5cm, thick,>=latex] (3.3,0.8) --  (4,2.5) node[below, midway] {~$\t$};
  \draw[->, thick,>=latex] (5, 1.2) -- (5.5,2) node[right, midway] {$\n_{h}$};
  \draw[->, thick,>=latex] (2.5, 1.1) -- (1.8,1.9) node[left, midway] {$\n_{h}$};
  \draw[->, thick,>=latex] (3.65, 1.65) -- (2.8,2);
  \draw[->, thick,>=latex] (3.65, 1.65) -- (4.5,2.1); 
  \node at (2.55,2) {$\nG_{2}$};
  \node at (4.5,1.8) {$\nG_{1}$};
  \node at (2.3,0.2) {$T_1$};
  \node at (5,0.3) {$T_2$};
  \draw[->, >=latex]  (2.7,2.6) node[left]{$E$} to [bend left = 20] (3.83,2.17);
\end{tikzpicture}
\hspace{2cm}
    \begin{tikzpicture}
        \draw[fill=lightgray!60] (1,0.2) to [bend left = 25] (3.3,0.8) to [bend left = 20] (4,2.5) to [bend right = 20] cycle;
  \draw[fill=lightgray!60] (3.3,0.8) to [bend left = 10] (6,0.3) to [bend right = 30] (4,2.5) to [bend right = 20] cycle;

  \draw[->, shorten >=0.5cm,shorten <=0.5cm, thick,>=latex] (3.15,0.8) to (3.75,2.5); 
  
  \draw[->, thick,>=latex] (4.8, 1.5) -- (5.4,2.4) node[right, midway] {$\n_{h}$};
  \draw[->, thick,>=latex] (2.2, 1.4) -- (1.5,2.2) node[left, midway] {$\n_{h}$};
  \draw[->, thick,>=latex] (3.45, 1.65) -- (2.6,1.73) node[below] {$~~\nG_{2}$};
  \draw[->, thick,>=latex] (3.45, 1.65) -- (4.35,1.7) node[below] {$\nG_{1}$};

   \node at (2.3,0.4) {$T_1$};
  \node at (4.6,0.4) {$T_2$};
  \node at (3.8,1.95) {$\t$};
  \draw[->, >=latex]  (2.7,2.6) node[left]{$E$} to [bend left = 20] (3.75,2.2);
\end{tikzpicture}
  \end{center}
  \caption{Two elements $T_1,T_2 \in \mesh$ with a common edge $E = \overline{T_1} \cap \overline{T_2}$ and the corresponding in-plane normal vectors $\nG_1,\nG_2$, the common tangential vector $\t$ and the oriented normal vector $\n_h$. On the left with a linear geometry approximation, on the right side with a high order approximation.}
  \label{fig:sketchofedgepatch} 
\end{figure}

Let $u$ and $\sigma$ be vector and matrix valued functions respectively. We define the vector valued average and jump on $E$ by
\begin{align*}
  \{\!\!\{ \sigma \nG \}\!\!\} := \frac{\sigma|_{T_1} \nG_1 - \sigma|_{T_2} \nG_2}{2} \quad \textrm{and} \quad [\![ u ]\!] = [\![ u ]\!]_{\t} \t + [\![ u ]\!]_{\nG} \overline{\nG}
\end{align*}
with the scalar valued \emph{normal} and \emph{tangential jump} and the averaged in-plane normal
\begin{align} 
  [\![ u ]\!]_{\nG} := u|_{T_1} \cdot \nG_1  + u|_{T_2} \cdot \nG_2, \qquad
  [\![ u ]\!]_{\t} := (u|_{T_1} - u|_{T_2}) \cdot \t, \qquad
  \overline{\nG} = \tfrac12 \nG_1 + \tfrac12 \nG_2.
\end{align}
In the case of $E \subset \partial \Gamma \neq \emptyset$
we set $[\![ u ]\!]_{\nG} = u \cdot \nG$ and the operators $  \{\!\!\{ \cdot \}\!\!\}$ and $[\![ \cdot ]\!] $ are replaced  by the identity. Based on the Piola mapping, see equation \eqref{piola}, we define a finite element space of order $k$ by
\begin{align*}
  W^k_h := \{ v \in [L^2(\Gamma_h)]^3 : \forall T \in \mesh ~ \exists \hat u \in [\mathbb{P}^k(\hat T)]^2 \text{ s.t. } v|_T = \mathcal{P}_T(\hat u)  \}.
\end{align*}
Here, $\mathbb{P}^k(\hat T)$ is the scalar-valued polynomial space of order $k$ on $\hat T$. By construction, we have that
\begin{align} \label{eq::disctangentplane}
  u_h(x) \cdot \n_h(x) =0 \quad \forall x \in \Gamma_h, \quad \forall u_h \in  W^k_h,
\end{align}
thus discrete functions in $W^k_h$ are exactly in the tangent plane of $\Gamma_h$. Assuming that $f$, originally defined on $\Gamma$, has a smooth extension on $\Gamma_h$, we follow Arnold et al.\cite{arnold2002unified} to define the DG method (based on an symmetric interior penalty formulation): Find $u_h \in W_h^k$ such that
\begin{subequations} \label{eq:DGmethod}
\begin{align} 
a_h(u_h,v_h) + m_h(u_h,v_h) = \int_{\Gamma_h} f \cdot v_h \ds  =: f_h(v_h) \quad \forall v_h \in W^k_h,
\end{align}
where
\begin{equation} \label{eq:DGblf}
  a_h(u_h,v_h)\!:=\!\sum\limits_{T \in \mesh} \int_T \eps_\Gamma(u_h)\!\colon\!\eps_\Gamma(v_h) \dx + \sum\limits_{E \in \mathcal{F}_h} \int_E \{\!\!\{ -\eps_\Gamma(u_h) \}\!\!\} \!\colon\! [\![ v_h ]\!] \ds 
                  + \int_E \{\!\!\{ -\eps_\Gamma(v_h) \}\!\!\} \!\colon\! [\![ u_h ]\!] \ds +  \frac{\alpha k^2}{h} \int_E [\![ u_h ]\!] \!\colon\! [\![ v_h ]\!] \ds\!.\!\!
\end{equation}
\end{subequations}
with the constant $\alpha$ chosen big enough, see Ref.\cite{arnold2002unified}, and $m_h(u_h,v_h) = \int_{\Gamma_h} P u_h v_h \ds$.
Note that the projection $P$ in the definition of the bilinear form $m_h(\cdot,\cdot)$ is redundant for functions in $W_h^k$. However, we will use $m_h(\cdot,\cdot)$ for non-tangential functions as well later on. 
%

\begin{remark} \label{rem::opensurface}
  In the case of nonhomogeneous Dirichlet or other types of boundary conditions one uses the standard techniques employed for DG methods on flat surfaces, see Ref. \cite{arnold2002unified}.
\end{remark}

\begin{remark} \label{rem::isometries}
  The application of the Piola mapping in the definition of functions in $W_h^k$ results in exactly tangential fields. Furthermore, it results in a discretization that is -- as the continuous problem -- invariant under isometries, i.e. if $u_h$ solves \eqref{eq:DGblf} on $\Gamma_h$, $\mathcal{P}(u_h)$ solves the corresponding problem on $\Gamma_h'=\Phi(\Gamma_h)$ if $\Phi$ is an isometry and $\mathcal{P}$ the corresponding Piola mapping. We discuss this in more detail below in the numerical examples, cf.~Sections~\ref{sec::houseofcards}and~\ref{sec:schaeferturek:isometric} below.
\end{remark}

\subsection{Discretization of the surface Stokes equations} \label{sec::discstokes}

In this section we focus on the construction of a numerical scheme to solve the Stokes problem on $\Gamma$, see equation \eqref{prob:stokes}. Note that we assumed $\partial \Gamma \neq \emptyset$ and thus also $\partial \Gamma_h \neq \emptyset$. As known from the literature, see for example Refs.\cite{brezzi2012mixed, braess}, discrete stability demands to find a compatible pair of the discrete velocity and pressure space. We aim to construct a method based on the works \cite{cockburn2005locally,cockburn2007note,LS_CMAME_2016} for the flat case. In the latter works, the discrete velocity space is chosen to be the $H(\divergence)$-conforming BDM space of order $k_u$, and the pressure is approximated by discontinuous polynomials of order $k_u-1$. These two spaces fulfill the property that the divergence of the velocity space is a subspace of the discrete pressure space. This ensures stability in the sense of Babu\v{s}ka-Brezzi, see for example Ref.\cite{LedererSchoeberl2017}, and leads to major benefits such as exactly divergence-free discrete velocities, see also Section \ref{sec::HdivBenefits}. Following the ideas of Rognes et al.\cite{rognes2013automating} and Castro et al.\cite{CASTRO2016241}, we now map this method to the surface using the previously introduced Piola mapping, see equation \eqref{piola}, to define an $H(\divergence_{\Gamma})$-conforming velocity space on $\Gamma_h$. 
Then for an arbitrary order $k$ we set
\begin{align} \label{eq:Hdivspace}
V_h^{k} := \{u_h \in W_h^{k}: [\![ u_h ]\!]_{\nG} = 0 ~ \forall E \in \facets  \} = W_h^k \cap H_0(\divergence_\Gamma),
\end{align}
where $H_0(\divergence_\Gamma)$ is the subspace of $H(\divergence_\Gamma)$ with zero normal trace at the boundary $\partial\Gamma$. Regarding the implementation of $V_h^{k}$ note, that the Piola mapping not only helps to incorporate property \eqref{eq::disctangentplane}, but can also be used to incorporate the zero normal jump. This is done by choosing the BDM basis on the reference element $\hat T$ and map it with the Piola transformation to the physical elements $T \in \mesh$. As the BDM basis functions are associated to the scalar normal moments on the edges of the reference element, the mapped functions are then associated to the according (in-plane) normal moments on the edges of the mapped element $T$. This automatically results in a normal continuous function, i.e. $  [\![ u ]\!]_{\nG} = 0$  on all edges. A detailed discussion on this topic is given in \cite{CASTRO2016241}. Next we define the discrete pressure space as
\begin{align*}
  Q_h^k := \{ q \in L^2(\Gamma_h): \forall T\in\mesh ~ \exists \hat q \in \mathbb{P}^k(\hat T) \text{ s.t. } q|_T = \hat q \circ \Phi_T^{-1} \} \cap L^2_0(\Gamma_h).
\end{align*}
For a fixed velocity approximation order $k_u$ the $H(\divergence_{\Gamma})$-conforming DG method then read as: Find $(u_h,p_h) \in V_h^{k_u} \times Q_h^{{k_u}-1}$ such that
\begin{align} \label{discstokesproblem}
  \begin{array}{rll}
  a_h(u_h, v_h) + b_h(v_h, p_h) &= f_h(v_h) &\quad \forall v_h \in V_h^{k_u}, \\
  b_h(u_h, q_h) &= 0 &\quad \forall q_h \in Q_h^{{k_u}-1}, 
  \end{array}
\end{align}
with
\begin{align*}
  b_h(u_h, q_h) = -\int_{\Gamma_h} \divergence_\Gamma(u_h) q_h \dx.
\end{align*}
Note that the appearing jumps in $a_h(u_h, v_h)$ now only read as the tangential jump since functions in $V_h^{k_u}$ are normal continuous. As a consequence we have the following lemma.
\begin{lemma} \label{lemma:exdivfree}
Let $u_h \in V_h^{k_u}$ such that $b_h(u_h, q_h) = 0 ~ \forall q_h \in Q_h^{{k_u}-1}$. Then $u_h$ is exactly divergence-free, i.e. $\divergence_\Gamma(u_h) = 0$ on $\Gamma_h$.
\end{lemma}
\begin{proof}
For each $T \in \mesh$ let $\hat u_{h,T} \in \mathbb{P}^{k_u}(\hat T, \rr^2)$ be such that $u_h = \mathcal{P}_T(\hat u_{h,T})$. Applying the chain rule shows
\begin{align} \label{divpiola}
  \divergence_\Gamma( u_h) = \frac{1}{J} \divergence( \hat u_{h,T})
\end{align}
where $\divergence( \hat u_{h,T})$ reads as the divergence on the flat reference element $\hat T$. As $\hat u_{h,T} \in \mathbb{P}^{k_u}(\hat T, \rr^2)$ we have $\divergence( \hat u_{h,T}) \in \mathbb{P}^{k_u-1}(\hat T)$. Choosing $\hat q_{h,T} =\sgn(J) \divergence( \hat u_{h,T})$ and $q_{h,T} = \hat q_{h,T} \circ \Phi^{-1}$, we define the global function $q_h$ as $q_h = q_h^\ast + c$ with $q_h^\ast|_T = q_{h,T}$ and the global constant $c \in \rr$ such that $ \int_{\Gamma_h} q_h = 0$, thus $q_h \in Q_h^{{k_u}-1}$. Then we have with
$
b_h(u_h, c) = -c \int_{\Gamma_h}  \divergence_\Gamma(u_h) \dx = -c \int_{\partial \Gamma_h} u_h \cdot \nG \ds = 0
$
that there holds
\begin{align*}
 0 = -b_h(u_h, q_h) = \int_{\Gamma_h} \divergence_\Gamma(u_h) q_h^\ast \dx = \sum\limits_{T \in \mesh} \int_{T} \divergence_\Gamma(u_h) q_{h,T} \dx = \sum\limits_{T \in \mesh} \int_{\hat T} \frac{1}{J} \sgn(J)\divergence( \hat u_{h,T})^2 |J| \dx = \sum\limits_{T \in \mesh} \int_{\hat T}\divergence( \hat u_{h,T})^2\dx 
\end{align*}
which shows that $\divergence( \hat u_{h,T}) = 0$, thus we conclude by equation \ref{divpiola}.
\end{proof}
From Lemma \ref{lemma:exdivfree} we conclude that the solution $u_h$ of \eqref{discstokesproblem} is exactly divergence-free. Note that this statement is completely independent of the geometry approximation. This further leads to pressure robustness (see Sections \ref{sec::HdivBenefits} and \ref{sec:stokesnumex}) and shows that the function spaces $V_h^{k_u}$ and $Q_h^{{k_u}-1}$ are compatible, thus the stability proof of \eqref{discstokesproblem} follows the same steps as in the flat case, see Refs.\cite{lehrenfeld2010hybrid,LS_CMAME_2016}.

\subsection{Discretization of the surface Navier--Stokes equations} \label{sec::discnavstokes}
The discretization of the Navier--Stokes problem \eqref{prob:navstokes} follows the ideas of Refs.\cite{LS_CMAME_2016, lehrenfeld2010hybrid} which we briefly summarize in the following. We aim to use a semi-discrete method to decouple the discretization in space and time. This is then further combined with an (high order) operator splitting method to efficiently deal with the nonlinear convection term. For the latter one we use a standard upwinding scheme in space which guarantees energy stability. To this end let $w_h, u_h, v_h \in V_h^{{k_u}}$, then we define the form
\begin{align} \label{convectionblf}
  c_h(w_h)(u_h,v_h) = \sum\limits_{T \in \mesh} -\int_T u_h \otimes w_h : \nabla v_h \dx + \int_{\partial T} w_h \cdot \nG ~u_h^{\operatorname{up}} \cdot v_h \ds,
\end{align}
where the upwinding value is chosen in upstream direction with respect to the convection velocity $w_h$: Let $x \in \partial T$ and $T'$ be the neighboring element, s.t. $x \in T, T'$, then we define
\begin{equation*}
  u_h^{\operatorname{up}} := (u_h \cdot \nG_T) \cdot \nG_T +
  (u_h^{\operatorname{up},\t} \cdot \t) \cdot \t, \quad \text{ where } u_h^{\operatorname{up},\t}(x) = u_h^{\operatorname{up},\t}|_T(x) \text{ if } u_h \cdot \nG_T \geq 0 \text{ and }
u_h^{\operatorname{up},\t}(x) = u_h^{\operatorname{up},\t}|_{T'}(x) \text{ otherwise.}
\end{equation*}
We notice that the upwinding only affects the discontinuous tangential component. 
This type of formulation is typically derived by applying partial integration in an element-by-element fashion and introducing \emph{numerical fluxes}, cf. e.g. Ref.\cite{hesthaven2007nodal}. Let us further note that the integration by parts formula for vector fields that are not exactly tangential involves an additional term including the mean curvature, cf. Ref.\cite{MR3846120} Eq. (3). 

In combination with the bilinear forms defined in the previous section and again assuming that the initial velocity $u_0$ has a smooth extension onto $\Gamma_h$, we derive the following spatially discrete DAE problem: find $u_h(t) \in V_h^{{k_u}}, p_h(t) \in Q_h^{{k_u}-1}$ such that
\begin{subequations}
\begin{align}
    m_h(\frac{\partial}{\partial t} u_h , v_h) + 2 \nu a_h (u_h,v_h) + b_h(v_h,p_h) + c_h(u_h)(u_h,v_h) &= f_h(v_h) && \forall v_h \in V_h^{{k_u}}, t \in (0,T], \\
    b_h(u_h,q_h) &= 0 && \forall q_h \in Q_h^{{k_u}-1}, t \in (0,T], \\
    m_h(u_h,v_h) &= m_h(u_0,v_h) && \forall v_h \in V_h^{{k_u}}, t =0, \label{eq:nse:initial}
\end{align}
\end{subequations}
where \eqref{eq:nse:initial} resembles the $L^2$-projection of the initial velocity onto the discrete velocity space $V_h^{{k_u}}$. Now let $\underline{u}_h(t)$ and $\underline{p}_h(t)$ be the finite element coefficient vectors of the functions $u_h,p_h$ respectively, and let $t_i$ with $i = i,\ldots,N$ where $t_0 = 0, t_N = T$ be an equidistant mesh for $[0,T]$ with time step size $\Delta t$. To get the solution at time $t_{i+1}$ we solve a step of the fully discrete lowest order implicit explicit (IMEX) splitting scheme
\begin{align} \label{loIMEX}
  \begin{array}{rl}
    M \frac{\underline{u}_h(t_{i+1}) -\underline{u}_h(t_{i})  }{\Delta t} (\frac{\partial}{\partial t} \underline{u}_h) + 2 \nu  A \underline{u}_h(t_{i+1})  + B^T\underline{p}_h(t_{i+1}) &= f_h(v_h) - C( \underline{u}_h(t_{i}))  \underline{u}_h(t_{i}),\\
    B  \underline{u}_h(t_{i+1}) &= 0.
      \end{array}
\end{align}
Here the matrices $M,A,B$ are the corresponding matrices of the bilinear forms $m_h,a_h$ and $b_h$ respectively, and $C( \underline{u}_h(t_{i}))  \underline{u}_h(t_{i})$ reads as the evaluation of the convection trilinear form $c_h$ with the wind $w_h = u_h(t_i)$. Above system shows that we treat the stiffness $A$ and divergence constraint $B$ implicitly which guarantees exactly divergence-free velocity solutions at each point in time. The convection $C$ however is treated explicitly. For more details regarding the efficiency and accuracy of above splitting methods for the Navier--Stokes equations we refer to Refs.\cite{LS_CMAME_2016, lehrenfeld2010hybrid} for the flat case. The simplest variant to generalize \eqref{loIMEX} to higher order order in time are IMEX schemes, cf. the works by Ascher et al. \cite{ascher1995implicit,ascher1997implicit}. Below in the numerical examples we use a second order IMEX based on two compatible explicit and implicit Runge-Kutta schemes of second order.


\subsection{On the benefits of $H(\divergence_{\Gamma})$-conformity} \label{sec::HdivBenefits}

Additionally to the fact that the $H(\divergence_{\Gamma})$-conforming space $V_h^{k_u}$ is tangential there are several advantages over other Stokes discretizations (see also the HDG Stokes discretization introduced  in Section \ref{sec:stokesnumex}). 
Below, we focus on two important ones. For other aspects -- which transfer directly from the flat case -- such as the ability to reduce the set of basis functions, convection stability (beyond energy stability) and good dissipation properties we refer to the literature, e.g. Refs.\cite{lehrenfeld2010hybrid,LS_CMAME_2016,fehn2019}.
\paragraph{Pressure robustness}
As discussed in Section \ref{sec::discstokes}, Lemma \ref{lemma:exdivfree} shows that solutions $u_h$ of \eqref{discstokesproblem} are exactly divergence-free independently of the geometry approximation. Further, the combination of the velocity space $V_h^{{k_u}}$ and the pressure space $Q_h^{{k_u}-1}$ allows to test equation \eqref{discstokesproblem} with exactly divergence-free velocity test functions. This is a crucial advantage as it enables us to derive velocity error estimates that are independent of the pressure approximation (and the viscosity $\nu$). This property is known in the literature as \emph{pressure robustness} and was first introduced by Linke\cite{linke2014role}. In the following we briefly sketch the main idea (hence the occurring problem): let the r.h.s. $f$ of the Stokes problem be decomposed as
\begin{align*}
f = \nabla_\Gamma \theta + \xi, 
\end{align*}
with $\theta \in H^1(\Gamma) / \rr$ and $\xi \in [L^2(\Gamma)]^3$ such that $\xi \cdot \n = 0$. Testing the continuous problem \eqref{prob:stokes} with a test function $v \in V_0:=\{ v \in H(\divergence_{\Gamma}, \Gamma): \divergence_\Gamma(v) = 0 \}$ we see that $2 \nu a(u,v) = \int_\Gamma \xi \cdot v \dx$, hence the velocity is not steered by the gradient $\nabla_\Gamma \theta$. If the same observation can be made in the discrete setting the method is called pressure robust and one can derive velocity error estimates that are independent of the pressure approximation (see Section \ref{sec:stokesnumex} for a numerical observation of this phenomenon). A rigorous analysis of such (Helmholtz) decompositions on smooth manifolds is given in Ref.\cite{reusken2018stream}. For the above findings in the continuous setting it was crucial that one tests with exactly divergence-free test functions $v \in V_0$. Lemma \ref{lemma:exdivfree} shows that the same can be done in the discrete setting which leads to pressure robustness of the discretization given by equation \eqref{discstokesproblem}. Note, that also standard methods can be made pressure robust, see Ref. \cite{blms:2015, lmt:2016, 2016arXiv160903701L, MR3683678}. Further, it also plays an important role in the Navier--Stokes case \cite{MR3564690, MR3824769, MR3875918, gauger2019}.

\paragraph{Energy stability}
The standard DG upwind formulation of the convection bilinear form \eqref{convectionblf} allows to show the following stability result (up to geometry errors). Let $\partial \Gamma_{in}:=\{ x \in \partial \Gamma: w_h \cdot \nG <0\}$ then there holds 
\begin{align*}
  \frac{1}{2} \int _{\partial \Gamma_{in}} | w_h \cdot \nG | u_h^2 \ds + c_h(w_h)(u_h,u_h) \ge 0 \quad \forall u_h,w_h \in V_h^{{k_u}} \textrm{ with } \divergence_\Gamma(w_h)=0.
\end{align*}
Note that the crucial assumption here is that the transport velocity $w_h$ is exactly (surface) divergence-free. Considering the time discretization \eqref{loIMEX} of the Navier--Stokes equations, the transport velocity was chosen as the discrete velocity of the old time step, i.e. $w_h = u_h(t_i)$. Thus, as $\divergence_\Gamma(u_h(t_i)) = 0$ (exactly) we can apply above stability estimate showing energy stability of the Navier--Stokes discretization.


\subsection{Hybrid discontinuous Galerkin formulations} \label{sec::DGvsHDG}
Although the formulation of Section \ref{sec::discveclap} fulfills property \eqref{eq::disctangentplane}, the computational overhead introduced by the DG formulation is a major drawback. A well known technique to overcome this circumstance is to use hybrid discontinuous Galerkin (HDG) formulations and a static condensation strategy. The idea is to introduce additional unknowns on the skeleton $\mathcal{F}_h$ which circumvents the direct coupling of element (interior) unknowns with neighboring elements. We give two examples how this can be achieved and comment on further efficiency improvements. Below, we will use HDG discretizations for the numerical examples.
\subsubsection{An HDG method for the Vector Laplacian} \label{sec::HDG}
Let $\Psi_E$ be the unique polynomial map from the reference edge $\hat E = [0,1]$ to an edge $E \in \facets$, with $E = \overline{T_1} \cap \overline{T_2},~T_1,T_2 \in \mesh$ as before. We define the space of (discontinuous) piecewise polynomials on the skeleton as
  \begin{align*}
    \Lambda_h^k := \{ v \in  L^2(\facets): \forall E \in \facets ~\exists  \hat v \in \mathbb{P}^k(\hat E) \text{ s.t. } v|_E = \hat v \circ \psi_E^{-1} \}.
  \end{align*}
Here, $\mathbb{P}^k(\hat E)$ is the scalar-valued polynomial space of order $k$ on $\hat E$.  Then we define the following two vector valued functions
\begin{align*}
  \lambda_{T_1} := \lambda_a \t + \lambda_b \nG_1 \quad \textrm{and} \quad   \lambda_{T_2} := \lambda_a \t + \lambda_b \nG_2  \quad \textrm{for } \lambda_a,\lambda_b \in     \Lambda_h^k, 
\end{align*}
where $\t$ and $\nG_1, \nG_2$ are defined as above. Here $\lambda_a$ is introduced to decouple the weak (in the DG sense) tangential continuity whereas $\lambda_b$ is introduced for the weak normal continuity. 
Note, that $\lambda_{T_1}$ might not be equal to $\lambda_{T_2}$ as the in-plane normal vector may be different. The HDG method then reads as: Find $u_h,\lambda_T \in W_h^k \times (\Lambda_h^k \times \Lambda_h^k)$ such that
\begin{align*}
  a_h^{\operatorname{HDG}}((u_h,\lambda_T) , (v_h,\theta_T)) + m_h(u_h,v_h) =  f_h(v_h) \quad \forall (v_h,\theta_T) \in W^k_h \times (\Lambda_h^k \times \Lambda_h^k) ,
\end{align*}
where
\begin{align*}
  a_h^{\operatorname{HDG}}((u_h,\lambda_T) , (v_h,\theta_T)) :=
  \sum\limits_{T \in \mesh} &\int_T \eps_\Gamma(u_h):\eps_\Gamma(v_h) \dx + \int_{\partial T} (\eps_\Gamma(u_h)~ \nG ) \cdot (\theta_T - v_h) \ds \\
    &+ \int_{\partial T} (\eps_\Gamma(v_h)~ \nG ) \cdot (\lambda_T - u_h) \ds +  \frac{\alpha k^2}{h} \int_{\partial T} (\lambda_T - u_h) \cdot (\theta_T - v_h)  \ds.
\end{align*}
The HDG method has more unknowns compared to the DG version, but the element unknowns $u_h \in W_h^k$ depend only on the facet function $\lambda_T$. Hence, in a linear system, we can form the Schur complement with respect to those unknowns which is known as \emph{static condensation}. This can be done in an element-by-element fashion already during the setup of the linear systems. The Schur complement system is typically much smaller, especially for higher order elements (note that the facet unknowns scale with $k$ whereas the element unknowns scale with $k^2$).
An example for the impact of static condensation with HDG methods in comparison to DG methods will be given in Section~\ref{sec::sphere}.

\subsubsection{An $H(\divergence_{\Gamma})$-conforming HDG method for the Stokes and the Navier--Stokes equations} \label{sec::HdivDG}

Similarly to the standard DG method of Section \ref{sec::discveclap}, we can also use a hybridization technique for the $H(\divergence_{\Gamma})$-conforming DG methods of Sections \ref{sec::discstokes} and \ref{sec::discnavstokes}. As normal continuity is already incorporated in the space $V_h^{{k_u}}$, we do not need the facet variable $\lambda_b$. Hence the facet variable on the skeleton $\facets$ reduces to the single valued quantity
\begin{align*}
  \lambda_{T_1} = \lambda_{T_2} = \lambda_a \t \quad \textrm{with} \quad \lambda_a\in \Lambda_h^{k_u},
\end{align*}
as the tangential vector $\t$ is the same on both sides. 
With the definitions $\lambda_T(\lambda_a,u_h) := \lambda_a \t  + (u_h \cdot \nG) \nG$ and
$\theta_T(\theta_a,u_h) := \theta_a \t  + (u_h \cdot \nG) \nG$, we can use the same HDG bilinear form $a_h^{\operatorname{HDG}}$ for the treatment of viscosity.
For the convection operator we have to replace the upwind flux $u_h^{\text{up}}$, cf. \eqref{convectionblf}, to avoid communication between interior element unknowns in $V_h^{{k_u}}$. Hence, we make the following choice: 
\begin{align*}
  u_h^{\operatorname{HDG},\operatorname{up}} := (u_h \cdot \nG) \nG + 
  \begin{cases}
    (u_h \cdot \t) \t & \text{if } w_h \cdot \nG > 0,\\
    \lambda_a \t & \textrm{else}.
  \end{cases}
\end{align*}
Further, to make sure that $\lambda_a$ takes the value of the upwind neighbor also in the convection dominated regime we add a consistent stabilization (the ``downwind glue'') as introduced in Ref. \cite{egger2009hybrid} and define
\begin{align*}
c^{\operatorname{HDG}}_h(w_h)((u_h,\lambda_T), (v_h,\theta_T)) := \sum\limits_{T \in \mesh} -\int_T u_h \otimes w_h : \nabla v_h \dx + \int_{\partial T} w_h \cdot \nG ~u_h^{\operatorname{HDGup}} \cdot v_h \ds + \int_{\partial T_{\text{out}}}  w_h \cdot \nG ~(\lambda_T -(u_h \cdot \t) \t) \cdot  \theta_T \ds,
\end{align*}
where $\partial T_{\text{out}}$ denotes the outflow boundary of an element, i.e. $\partial T_{\text{out}} := \partial T \cap \{ w_h \cdot \nG > 0 \}$. Let us stress that in the convective limit $\nu \to 0$, the HDG solution converges to the DG solution discussed above. 

The semi-discretization of the new hybrid $H(\divergence_{\Gamma})$-conforming DG method for the Navier--Stokes equations then reads: Find $(u_h,\lambda_a, p_h) \in V_h^{{k_u}} \times \Lambda_h^{k_u} \times Q_h^{{k_u}-1}$ such that
\begin{subequations}
\begin{align}
  m_h( \frac{\partial}{\partial t} u_h,v_h) +  2 \nu a_h^{\operatorname{HDG}}((u_h,\lambda_T) , (v_h,\theta_T)) \qquad\qquad& \nonumber \\
  + c_h(u_h)((u_h,\lambda_T), (v_h,\theta_T)) + b_h(v_h, p_h) &= f_h(v_h) &&~ \forall v_h, \theta_a \in V_h^{{k_u}} \times \Lambda_h^{k_u}, t \in (0,T],  \\
    b_h(u_h,q_h) &= 0 &&~ \forall q_h \in Q_h^{{k_u}-1}, t \in (0,T],  \\ 
    m_h(u_h,v_h) - m_h(u_0,v_h)&= 0 &&~ \forall v_h \in V_h^{{k_u}}, t =0, \label{eq:nse:initial2}
\end{align}
\end{subequations}
where we implicitly made use of $\lambda_T=\lambda_T(\lambda_a,u_h)$ and $\theta_T=\theta_T(\theta_a,v_h)$.
Using this spatial discretization we can then again use the implicit-explicit splitting scheme for the discretization of the Navier--Stokes equations on $\Gamma_h$, see equation \eqref{loIMEX}. 

\paragraph{Efficiency aspects and superconvergence}
One advantage of HDG methods that we did not address so far is the ability to achieve superconvergence in diffusion dominated problems. If the (trace) unknowns on the skeleton are approximated with polynomials of order $k$, one can (for example) apply a local postprocessing strategy (see Ref. \cite{cockburn2009unified}) to define a local element-wise approximation of order $k+1$ which has order $k+2$ accuracy in the $L^2$ norm. Alternatively, one can consider the previous HDG formulation and reduce the facet degree by one order while preserving the order of accuracy by introducing a slight modification in the formulation, cf. Ref. \cite[Section 2.2]{LS_CMAME_2016}. In the context of $H(\operatorname{div})$-conforming methods, the reduction of the polynomial degree for the normal component requires a bit more care in order to preserve the benefitial properties of $H(\operatorname{div})$-conforming methods, cf. Ref.\cite{LLS_SIAM_2017,LLS_ESAIM_2019}.

\section{Numerical examples} \label{sec::numex}

In this section we consider several numerical examples. We start with a comparison of $[H^1(\Gamma)]^3$-conforming and the previously introduced non-conforming finite element methods for the Vector Laplacian for smooth and piecewise smooth manifolds in Subsection~\ref{sec::numex::veclap}. In the subsequent Subsection~\ref{sec:stokesnumex} we consider and compare the non-conforming methods for a surface Stokes problem. In the remainder we fix the method to the $H(\divergence_{\Gamma})$-conforming HDG method and consider several surface versions of a well-known benchmark problem in 2D in Subsection~\ref{sec:schaeferturek} and the Kelvin-Helmholtz instability problem on different rotationally symmetric surfaces in Subsection~\ref{sec::numex::KH}. Finally, in Subsection~\ref{sec::bunny} we consider a self-organisation process on the Stanford bunny geometry.
All numerical examples were implemented within the finite element library Netgen/NGSolve, see Refs. \cite{netgen,schoeberl2014cpp11} and \url{www.ngsolve.org}.
The data, i.e. time series, mesh refinement series, etc. as well as scripts for reproduction for all numerical examples can be found at \href{https://doi.org/10.5281/zenodo.3406173}{DOI: \texttt{10.5281/zenodo.3406173}}.

\subsection{Vector Laplacian} \label{sec::numex::veclap}
First, we consider two Vector Laplace problems with different discretizations, compare and discuss them. Additionally to those discretizations introduced before in Section \ref{sec::discveclap}, we consider two $[H^1(\Gamma_h)]^3$-conforming discretizations from the literature that we briefly summarize in Section \ref{sec::h1conf} below and which we denote as \texttt{H1-L} (weak enforcement of tangential condition through Lagrangian multipliers) and \texttt{H1-P} (weak enforcement of tangential condition through penalties).
We introduce labels for the aforementioned methods. The method in \eqref{eq:DGmethod} will be denoted as \texttt{DG} whereas the hybrid version of it is denoted as \texttt{HDG}. For the $H(\operatorname{div},\Gamma_h)$-conforming discretization, i.e. using \eqref{eq:DGmethod} with $W_h^k$ replaced by $V_h^k$, cf. \eqref{eq:Hdivspace}, we use the label \texttt{Hdiv-DG} and denote the HDG version as \texttt{Hdiv-HDG}.
For the numerical computations we only consider the HDG versions. Note however, that the differences between DG and HDG are primarily in a computational aspect, see also the explanations in Section \ref{sec::DGvsHDG}. Hence, the \texttt{DG} versions are only considered for a comparison of this observation.

\subsubsection{$H^1$-conforming discretizations for the Vector Laplacian}\label{sec::h1conf}
In the literature mostly $[H^1(\Gamma_h)]^3$-conforming surface FEM discretizations for vector-valued problems are considered by either applying Lagrangian multipliers for the tangential constraint or using a penalty formulation, see e.g. Refs.\cite{MR3846120,MR3840893,jankuhn2019trace,hansbo2017stabilized,hansbo2016analysis,voigtreuther18}.
As we will use two instances of these methods for comparison to the methods discussed in Section \ref{sec::FEMconstruction}, we briefly state these two formulations and discuss the choices involved. The methods are defined as follows: \\[1ex]
\begin{subequations}
\begin{minipage}{0.535\textwidth}
  Find $(u_h,\lambda_h) \in [S_h^k]^3\times S_h^{k_l}$ s.t. $\forall (v_h,\mu_h) \in [S_h^k]^3 \times S_h^{k_l}$: 
  \begin{align} \label{eq:H1-L}
    \!\!\!\!\!(a_h\!+\!m_h)(u_h,v_h) +\!\! \int_{\Gamma_h}\! (v_h \!\cdot\! \n_h) \lambda_h \!\ds + \!\!\int_{\Gamma_h}\! (u_h \!\cdot\! \n_h) \mu_h \!\ds &= f_h(v_h).\!
  \end{align}
\end{minipage}
\begin{minipage}{0.465\textwidth}
  Find $u_h \in [S_h^k]^3$ s.t. $\forall v_h \in [S_h^k]^3$:
  \begin{align} \label{eq:H1-P}
    \!\!\!(a_h\!+\!m_h)(u_h,v_h) + \!\int_{\Gamma_h}\!\! \rho (u_h \!\cdot\! \n_h) (v_h \!\cdot\! \n_h) \!\ds = f_h(v_h).\!
  \end{align}
\end{minipage}
\end{subequations}
\\
The $[H^1(\Gamma_h)]^3$-conforming method \texttt{H1-L} in \eqref{eq:H1-L} uses Langrange multipliers to enforce the tangential constraint whereas \eqref{eq:H1-P} enforces the tangential constraint through penalization with a penalty parameter $\rho$ (\texttt{H1-P}).
For \texttt{H1-L} the choice of the Lagrange multiplier space is $k_l=k$ as in Ref.\cite{hansbo2016analysis} or Ref.\cite{jankuhn2019trace} (in a TraceFEM context). We note that in Ref.\cite{MR3846120} an $h^{-6}$ scaling of the condition number for Navier--Stokes problems has been observed for $k_l=k$ and $k_l=k-1$ has been used instead.
To drive the normal component sufficiently small in the \texttt{H1-P} method, we make the simple choice $\rho = 10 \cdot h^{-(k+1)}$ and accept the severe ill-conditioning of arising linear systems here. In Refs.\cite{hansbo2016analysis,jankuhn2019trace} choosing such a large penalty is circumvented
by replacing $\varepsilon_\Gamma(u_h)$ (and $\varepsilon_\Gamma(v_h)$) by $\varepsilon_\Gamma(P u_h)$ (and $\varepsilon_\Gamma(P v_h)$) which is a consistent manipulation. Thereby the normal and tangential parts of the discrete solution decouple and a \emph{soft penalty} with a penalty parameter $\rho \sim h^{-2}$ suffices. 

Both choices taken here, $k_l=k$ for \texttt{H1-L} and $\rho = 10 \cdot h^{-(k+1)}$ for \texttt{H1-P} will eventually result in severe ill-conditioning of arising linear systems. However, for the test cases under consideration with \texttt{H1-L} and \texttt{H1-P} we used sparse direct solvers and obtained results of sufficient meaningfulness to serve as candidates for comparison.

\begin{figure}
  \begin{center}
    \includegraphics[width=.84\textwidth]{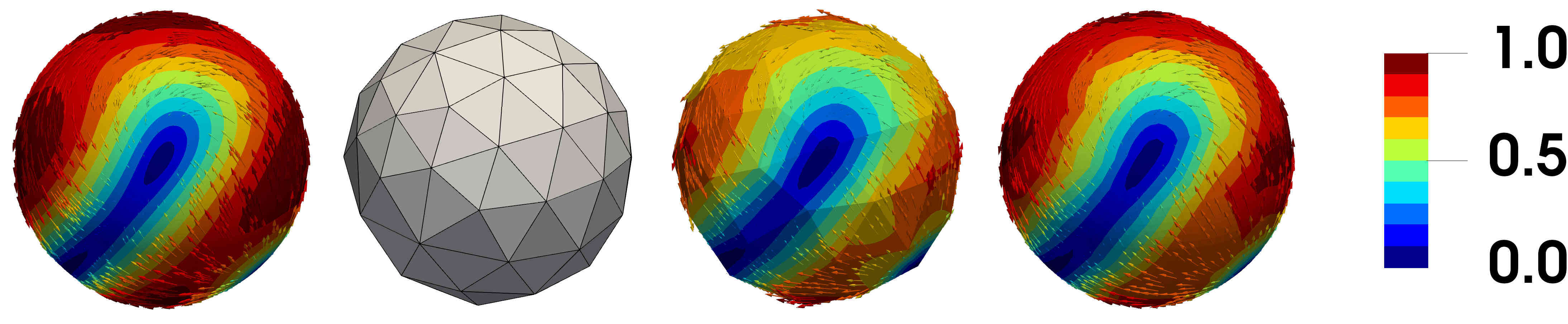}
  \end{center}
  \vspace*{-0.5cm}
  \hspace*{0.145\textwidth}
  $ \displaystyle \Vert u \Vert $
  \hspace*{0.285\textwidth}
  $ \Vert u_{h,(1,1)}^{\texttt{Hdiv-HDG}} \Vert $
  \hspace*{0.075\textwidth}
  $ \Vert u_{h,(1,2)}^{\texttt{Hdiv-HDG}} \Vert $
  \vspace*{-0.25cm}
  \caption{Exact and discrete geometry and solution for the Vector Laplacian in Section \ref{sec::sphere}. Coloring corresponds to velocity magnitude. From left to right: exact geometry and solution, (uncurved) computational mesh on coarsest level $L=0$, mesh and discrete solution for \texttt{Hdiv-HDG} with $k=1$, $k_g=1$, and discrete solution for \texttt{Hdiv-HDG} with $k=1$, $k_g=2$.}
  \label{fig:sphere}
\vspace*{-0.3cm}
\end{figure}

\subsubsection{Vector Laplacian on the sphere}\label{sec::sphere}
We consider the Vector Laplace problem with the (extended) exact solution $u^e = (-z^2, y, x)^T$ on the unit sphere $\Gamma = S_1(0)$, i.e. we choose the r.h.s. $f$ so that $u=u^e|_\Gamma$ solves \eqref{prob:veclap}.
For the geometry approximation we considered two choices, $k_g=k$ and $k_g=k+1$ and for $\alpha$ in the SIP formulations we take $\alpha = 10$. On five successively and uniformly refined meshes we evaluate $L^2$- and $H^1$-errors. The initial mesh (mesh level $L=0$) consists of 124 triangles, see Fig. \ref{fig:sphere}.
%
%
\paragraph{Comparison of methods}%
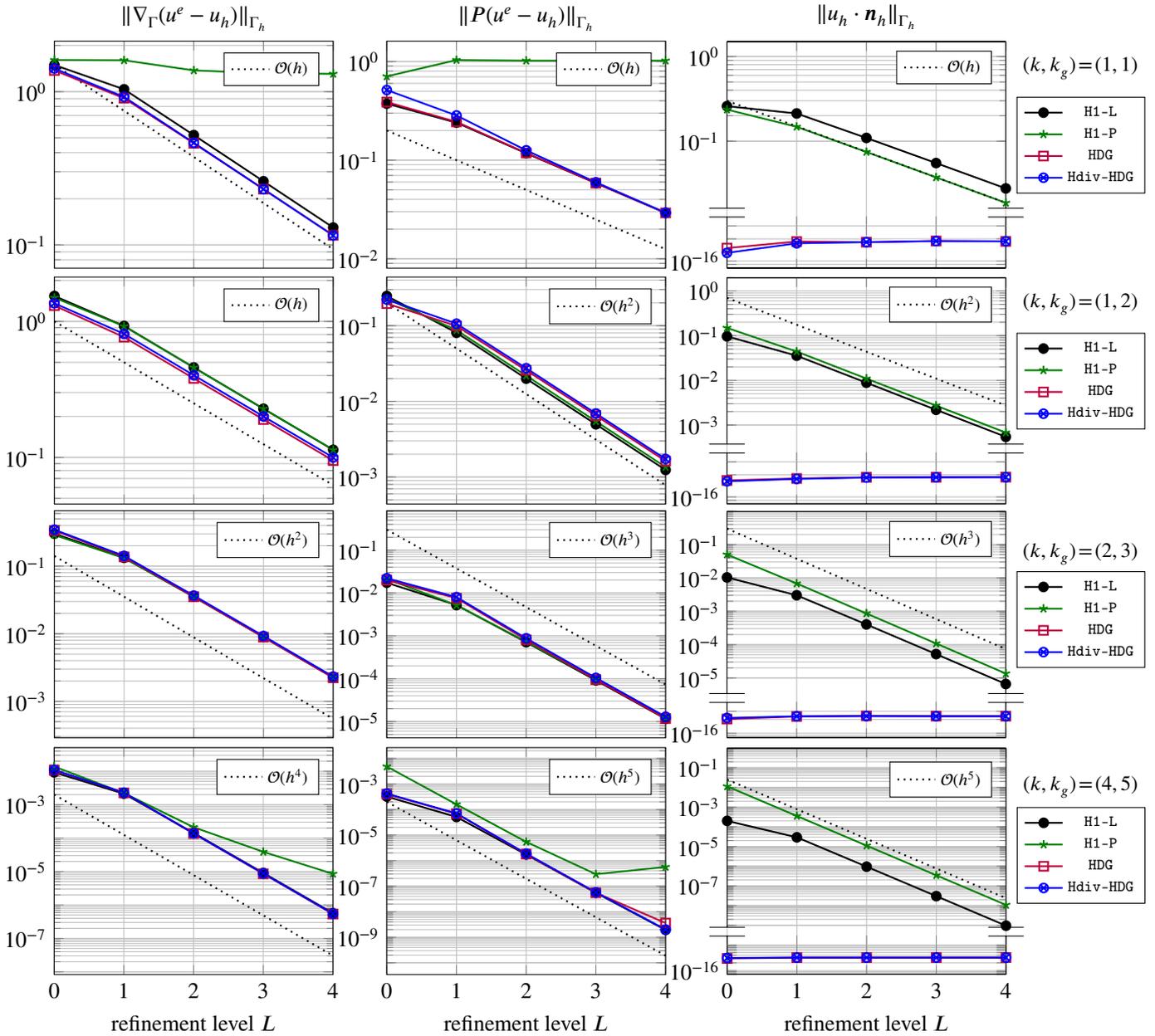
\begin{figure}[b!]
\vspace*{-0.6cm}
\begin{center}
\hspace*{-0.8cm}
\begin{tikzpicture}
  \begin{semilogyaxis}[
    grid=both,
    axis x line=box,
    axis line style={-},
    width=6cm,
    xmin=0, xmax=4,
    xticklabels={},
    title={$\Vert \nabla_\Gamma (u^e - u_h) \Vert_{\Gamma_h}$},
    title style={at={(axis description cs:0.5,0.95)},anchor=south},
    legend style={at={(0.95,0.95)}, fill=white, draw=black},
    ]

    \addplot[dotted,thick] table[x index =0,
    y expr={1.5*2^(-1*\thisrowno{0})}]{data/4_1_2_VecLap_sphere/comparison_H1Error_Pk1_Pg1.dat};
    \addplot[draw=black, thick, mark=*] table[x index=0, y index=5]
    {data/4_1_2_VecLap_sphere/comparison_H1Error_Pk1_Pg1.dat}; 
    \addplot[draw=green!50!black, thick, mark=star] table[x index=0, y index=6]
    {data/4_1_2_VecLap_sphere/comparison_H1Error_Pk1_Pg1.dat}; 
    \addplot[draw=purple,thick, mark=square] table[x index=0, y index=3]
    {data/4_1_2_VecLap_sphere/comparison_H1Error_Pk1_Pg1.dat}; 
    \addplot[draw=blue, thick, mark=otimes] table[x index=0, y index=4]
    {data/4_1_2_VecLap_sphere/comparison_H1Error_Pk1_Pg1.dat};
    \legend{
      \footnotesize $\mathcal{O}(h)$ \\
    }
\end{semilogyaxis}
\end{tikzpicture}
\hspace*{-0.35cm}
\begin{tikzpicture}
  \begin{semilogyaxis}[
    grid=both,
    axis x line=box,
    axis line style={-},
    width=6cm,
    xmin=0, xmax=4,
    xticklabels={},
    title={$\Vert P (u^e - u_h) \Vert_{\Gamma_h}$},
    title style={at={(axis description cs:0.5,0.95)},anchor=south},
    legend style={at={(0.95,0.95)}, fill=white, draw=black},
    ]
    \addplot[dotted,thick] table[x index =0,
    y expr={0.2*2^(-1*\thisrowno{0})}]{data/4_1_2_VecLap_sphere/comparison_L2Error_tang_Pk1_Pg1.dat};
    \addplot[draw=black, thick, mark=*] table[x index=0, y index=5]
    {data/4_1_2_VecLap_sphere/comparison_L2Error_tang_Pk1_Pg1.dat}; 
    \addplot[draw=green!50!black, thick, mark=star] table[x index=0, y index=6]
    {data/4_1_2_VecLap_sphere/comparison_L2Error_tang_Pk1_Pg1.dat}; 
    \addplot[draw=purple,thick, mark=square] table[x index=0, y index=3]
    {data/4_1_2_VecLap_sphere/comparison_L2Error_tang_Pk1_Pg1.dat}; 
    \addplot[draw=blue, thick, mark=otimes] table[x index=0, y index=4]
    {data/4_1_2_VecLap_sphere/comparison_L2Error_tang_Pk1_Pg1.dat};
    \legend{
      \footnotesize $\mathcal{O}(h)$ \\
    }
\end{semilogyaxis}
\end{tikzpicture}
\hspace*{-0.35cm}
\begin{tikzpicture}
\begin{groupplot}[
    group style={
        group name=my fancy plots,
        group size=1 by 2,
        xticklabels at=edge bottom,
        vertical sep=0pt
    },
    width=6cm,
    xmin=0, xmax=4,
]
\nextgroupplot[ymode=log, ymin=1e-2,ymax=1.5e0,
      grid=both,
      axis x line=box,
      separate axis lines,
      every outer x axis line/.append style={-,draw=none},
      every outer y axis line/.append style={-},
      ytick={1e-1,1e0},
      xtick={},
      axis y discontinuity=parallel, 
      height=4.5cm,
      legend style={at={(1.5,0.7)}},
      title={$\Vert u_h \cdot \n_h \Vert_{\Gamma_h}$},
      title style={at={(axis description cs:0.5,0.95)},anchor=south},
      ]
      \addplot[draw=black, thick, mark=*] table[x index=0, y index=5]
      {data/4_1_2_VecLap_sphere/comparison_L2Error_normal_Pk1_Pg1.dat}; 
      \addplot[draw=green!50!black, thick, mark=star] table[x index=0, y index=6]
      {data/4_1_2_VecLap_sphere/comparison_L2Error_normal_Pk1_Pg1.dat}; 
      \addplot[draw=purple,thick, mark=square] table[x index=0, y index=3]
      {data/4_1_2_VecLap_sphere/comparison_L2Error_normal_Pk1_Pg1.dat}; 
      \addplot[draw=blue, thick, mark=otimes] table[x index=0, y index=4]
      {data/4_1_2_VecLap_sphere/comparison_L2Error_normal_Pk1_Pg1.dat};
      \addplot[dotted,thick] table[x index =0,
      y expr={0.3*2^(-1*\thisrowno{0})}]{data/4_1_2_VecLap_sphere/comparison_L2Error_normal_Pk1_Pg1.dat};
      \legend{
        \footnotesize \texttt{H1-L} \\
        \footnotesize \texttt{H1-P} \\
        \footnotesize \texttt{HDG}  \\
        \footnotesize \texttt{Hdiv-HDG} \\
      }
      \draw[thick] (axis description cs:0,1) -- (axis description cs:1,1); 

\nextgroupplot[ymode=log, ymin=8e-17,ymax=3e-16, 
      grid=both,
      axis x line=box,
      separate axis lines,
      every outer x axis line/.append style={-,draw=none},
      every outer y axis line/.append style={-},
      ytick={9e-17,1e-16,2e-16,3e-16},
      yticklabels={,$10^{-16}$,,},
      xticklabels={},
      height=2.25cm,
      legend style={at={(0.95,5.15)}, fill=white, draw=black},
      ]
      \addplot[dotted,thick] table[x index =0,
      y expr={0.3*2^(-1*\thisrowno{0})}]{data/4_1_2_VecLap_sphere/comparison_L2Error_normal_Pk1_Pg1.dat};
      \addplot[draw=purple,thick, mark=square] table[x index=0, y index=3]
      {data/4_1_2_VecLap_sphere/comparison_L2Error_normal_Pk1_Pg1.dat}; 
      \addplot[draw=blue, thick, mark=otimes] table[x index=0, y index=4]
      {data/4_1_2_VecLap_sphere/comparison_L2Error_normal_Pk1_Pg1.dat};
      \legend{
        \footnotesize $\mathcal{O}(h)$ \\
      }
      \draw[thick] (axis description cs:0,0) -- (axis description cs:1,0); 

\end{groupplot}
\node[scale=0.925] at (5.6,2.5) { $(k,k_g)\!=\!(1,1)$};
\end{tikzpicture}\hspace*{-0.5cm}
\end{center}
\vspace*{-0.85cm}
\begin{center}
\hspace*{-0.8cm}
\begin{tikzpicture}
  \begin{semilogyaxis}[
    grid=both,
    axis x line=box,
    axis line style={-},
    width=6cm,
    xmin=0, xmax=4,
    xticklabels={},
    legend style={at={(0.95,0.95)}, fill=white, draw=black},
    ]

    \addplot[dotted,thick] table[x index =0,
    y expr={1*2^(-1*\thisrowno{0})}]{data/4_1_2_VecLap_sphere/comparison_H1Error_Pk1_Pg2.dat};
    \addplot[draw=black, thick, mark=*] table[x index=0, y index=5]
    {data/4_1_2_VecLap_sphere/comparison_H1Error_Pk1_Pg2.dat}; 
    \addplot[draw=green!50!black, thick, mark=star] table[x index=0, y index=6]
    {data/4_1_2_VecLap_sphere/comparison_H1Error_Pk1_Pg2.dat}; 
    \addplot[draw=purple,thick, mark=square] table[x index=0, y index=3]
    {data/4_1_2_VecLap_sphere/comparison_H1Error_Pk1_Pg2.dat}; 
    \addplot[draw=blue, thick, mark=otimes] table[x index=0, y index=4]
    {data/4_1_2_VecLap_sphere/comparison_H1Error_Pk1_Pg2.dat};
    \legend{
      \footnotesize $\mathcal{O}(h)$ \\
    }
\end{semilogyaxis}
\end{tikzpicture}
\hspace*{-0.35cm}
\begin{tikzpicture}
  \begin{semilogyaxis}[
    grid=both,
    axis x line=box,
    axis line style={-},
    width=6cm,
    xmin=0, xmax=4,
    xticklabels={},
    legend style={at={(0.95,0.95)}, fill=white, draw=black},
    ]
    \addplot[dotted,thick] table[x index =0,
    y expr={0.2*2^(-2*\thisrowno{0})}]{data/4_1_2_VecLap_sphere/comparison_L2Error_tang_Pk1_Pg2.dat};
    \addplot[draw=black, thick, mark=*] table[x index=0, y index=5]
    {data/4_1_2_VecLap_sphere/comparison_L2Error_tang_Pk1_Pg2.dat}; 
    \addplot[draw=green!50!black, thick, mark=star] table[x index=0, y index=6]
    {data/4_1_2_VecLap_sphere/comparison_L2Error_tang_Pk1_Pg2.dat}; 
    \addplot[draw=purple,thick, mark=square] table[x index=0, y index=3]
    {data/4_1_2_VecLap_sphere/comparison_L2Error_tang_Pk1_Pg2.dat}; 
    \addplot[draw=blue, thick, mark=otimes] table[x index=0, y index=4]
    {data/4_1_2_VecLap_sphere/comparison_L2Error_tang_Pk1_Pg2.dat};
    \legend{
      \footnotesize $\mathcal{O}(h^2)$ \\
    }
\end{semilogyaxis}
\end{tikzpicture}
\hspace*{-0.35cm}
\begin{tikzpicture}
\begin{groupplot}[
    group style={
        group name=my fancy plots,
        group size=1 by 2,
        xticklabels at=edge bottom,
        vertical sep=0pt
    },
    width=6cm,
    xmin=0, xmax=4,
]
\nextgroupplot[ymode=log, ymin=1.5e-4,ymax=2e0,
      grid=both,
      axis x line=box,
      separate axis lines,
      every outer x axis line/.append style={-,draw=none},
      every outer y axis line/.append style={-},
      ytick={1e-3,1e-2,1e-1,1e0},
      xtick={},
      axis y discontinuity=parallel, 
      height=4.5cm,
      legend style={at={(1.5,0.7)}},
      ]
      \addplot[draw=black, thick, mark=*] table[x index=0, y index=5]
      {data/4_1_2_VecLap_sphere/comparison_L2Error_normal_Pk1_Pg2.dat}; 
      \addplot[draw=green!50!black, thick, mark=star] table[x index=0, y index=6]
      {data/4_1_2_VecLap_sphere/comparison_L2Error_normal_Pk1_Pg2.dat}; 
      \addplot[draw=purple,thick, mark=square] table[x index=0, y index=3]
      {data/4_1_2_VecLap_sphere/comparison_L2Error_normal_Pk1_Pg2.dat}; 
      \addplot[draw=blue, thick, mark=otimes] table[x index=0, y index=4]
      {data/4_1_2_VecLap_sphere/comparison_L2Error_normal_Pk1_Pg2.dat};
      \addplot[dotted,thick] table[x index =0,
      y expr={0.7*2^(-2*\thisrowno{0})}]{data/4_1_2_VecLap_sphere/comparison_L2Error_normal_Pk1_Pg2.dat};
      \legend{
        \footnotesize \texttt{H1-L} \\
        \footnotesize \texttt{H1-P} \\
        \footnotesize \texttt{HDG}  \\
        \footnotesize \texttt{Hdiv-HDG} \\
      }
      \draw[thick] (axis description cs:0,1) -- (axis description cs:1,1); 

\nextgroupplot[ymode=log, ymin=8e-17,ymax=3e-16, 
      grid=both,
      axis x line=box,
      separate axis lines,
      every outer x axis line/.append style={-,draw=none},
      every outer y axis line/.append style={-},
      ytick={9e-17,1e-16,2e-16,3e-16},
      yticklabels={,$10^{-16}$,,},
      xticklabels={},
      height=2.25cm,
      legend style={at={(0.95,5.15)}, fill=white, draw=black},
      ]
      \addplot[dotted,thick] table[x index =0,
      y expr={0.3*2^(-2*\thisrowno{0})}]{data/4_1_2_VecLap_sphere/comparison_L2Error_normal_Pk1_Pg2.dat};
      \addplot[draw=purple,thick, mark=square] table[x index=0, y index=3]
      {data/4_1_2_VecLap_sphere/comparison_L2Error_normal_Pk1_Pg2.dat}; 
      \addplot[draw=blue, thick, mark=otimes] table[x index=0, y index=4]
      {data/4_1_2_VecLap_sphere/comparison_L2Error_normal_Pk1_Pg2.dat};
      \legend{
        \footnotesize $\mathcal{O}(h^2)$ \\
      }
      \draw[thick] (axis description cs:0,0) -- (axis description cs:1,0); 

\end{groupplot}
\node[scale=0.925] at (5.6,2.5) { $(k,k_g)\!=\!(1,2)$};
\end{tikzpicture}\hspace*{-0.5cm}
\end{center}
\vspace*{-0.95cm}
\begin{center}
\hspace*{-0.8cm}
\begin{tikzpicture}
  \begin{semilogyaxis}[
    grid=both,
    axis x line=box,
    axis line style={-},
    width=6cm,
    xmin=0, xmax=4,
    xticklabels={},
    legend style={at={(0.95,0.95)}, fill=white, draw=black},
    ]

    \addplot[dotted,thick] table[x index =0,
    y expr={0.141*2^(-2*\thisrowno{0})}]{data/4_1_2_VecLap_sphere/comparison_H1Error_Pk2_Pg3.dat};
    \addplot[draw=black, thick, mark=*] table[x index=0, y index=5]
    {data/4_1_2_VecLap_sphere/comparison_H1Error_Pk2_Pg3.dat}; 
    \addplot[draw=green!50!black, thick, mark=star] table[x index=0, y index=6]
    {data/4_1_2_VecLap_sphere/comparison_H1Error_Pk2_Pg3.dat}; 
    \addplot[draw=purple,thick, mark=square] table[x index=0, y index=3]
    {data/4_1_2_VecLap_sphere/comparison_H1Error_Pk2_Pg3.dat}; 
    \addplot[draw=blue, thick, mark=otimes] table[x index=0, y index=4]
    {data/4_1_2_VecLap_sphere/comparison_H1Error_Pk2_Pg3.dat};
    \legend{
      \footnotesize $\mathcal{O}(h^2)$ \\
    }
\end{semilogyaxis}
\end{tikzpicture}
\hspace*{-0.35cm}
\begin{tikzpicture}
  \begin{semilogyaxis}[
    grid=both,
    axis x line=box,
    axis line style={-},
    width=6cm,
    xmin=0, xmax=4,
    ytick={1e-6,1e-5,1e-4,1e-3,1e-2,1e-1,1e0},
    xticklabels={},
    legend style={at={(0.95,0.95)}, fill=white, draw=black},
    ]
    \addplot[dotted,thick] table[x index =0,
    y expr={0.3*2^(-3*\thisrowno{0})}]{data/4_1_2_VecLap_sphere/comparison_L2Error_tang_Pk2_Pg3.dat};
    \addplot[draw=black, thick, mark=*] table[x index=0, y index=5]
    {data/4_1_2_VecLap_sphere/comparison_L2Error_tang_Pk2_Pg3.dat}; 
    \addplot[draw=green!50!black, thick, mark=star] table[x index=0, y index=6]
    {data/4_1_2_VecLap_sphere/comparison_L2Error_tang_Pk2_Pg3.dat}; 
    \addplot[draw=purple,thick, mark=square] table[x index=0, y index=3]
    {data/4_1_2_VecLap_sphere/comparison_L2Error_tang_Pk2_Pg3.dat}; 
    \addplot[draw=blue, thick, mark=otimes] table[x index=0, y index=4]
    {data/4_1_2_VecLap_sphere/comparison_L2Error_tang_Pk2_Pg3.dat};
    \legend{
      \footnotesize $\mathcal{O}(h^3)$ \\
    }
\end{semilogyaxis}
\end{tikzpicture}
\hspace*{-0.35cm}
\begin{tikzpicture}
\begin{groupplot}[
    group style={
        group name=my fancy plots,
        group size=1 by 2,
        xticklabels at=edge bottom,
        vertical sep=0pt
    },
    width=6cm,
    xmin=0, xmax=4,
]
\nextgroupplot[ymode=log, ymin=1e-6,ymax=1e0,
      grid=both,
      axis x line=box,
      separate axis lines,
      every outer x axis line/.append style={-,draw=none},
      every outer y axis line/.append style={-},
      ytick={1e-5,1e-4,1e-3,1e-2,1e-1,1e0},
      yticklabels={$10^{-5}$,$10^{-4}$,$10^{-3}$,$10^{-2}$,$10^{-1}$},
      xtick={},
      axis y discontinuity=parallel, 
      height=4.75cm,
      legend style={at={(1.5,0.7)}},
      ]
      \addplot[draw=black, thick, mark=*] table[x index=0, y index=5]
      {data/4_1_2_VecLap_sphere/comparison_L2Error_normal_Pk2_Pg3.dat}; 
      \addplot[draw=green!50!black, thick, mark=star] table[x index=0, y index=6]
      {data/4_1_2_VecLap_sphere/comparison_L2Error_normal_Pk2_Pg3.dat}; 
      \addplot[draw=purple,thick, mark=square] table[x index=0, y index=3]
      {data/4_1_2_VecLap_sphere/comparison_L2Error_normal_Pk2_Pg3.dat}; 
      \addplot[draw=blue, thick, mark=otimes] table[x index=0, y index=4]
      {data/4_1_2_VecLap_sphere/comparison_L2Error_normal_Pk2_Pg3.dat};
      \addplot[dotted,thick] table[x index =0,
      y expr={0.3*2^(-3*\thisrowno{0})}]{data/4_1_2_VecLap_sphere/comparison_L2Error_normal_Pk2_Pg3.dat};
      \legend{
        \footnotesize \texttt{H1-L} \\
        \footnotesize \texttt{H1-P} \\
        \footnotesize \texttt{HDG}  \\
        \footnotesize \texttt{Hdiv-HDG} \\
      }
      \draw[thick] (axis description cs:0,1) -- (axis description cs:1,1); 

\nextgroupplot[ymode=log, ymin=8e-17,ymax=3e-16, 
      grid=both,
      axis x line=box,
      separate axis lines,
      every outer x axis line/.append style={-,draw=none},
      every outer y axis line/.append style={-},
      ytick={9e-17,1e-16,2e-16,3e-16},
      yticklabels={,$10^{-16}$,,},
      xticklabels={},
      height=2.cm,
      legend style={at={(0.95,8.2)}, fill=white, draw=black},
      ]
      \addplot[dotted,thick] table[x index =0,
      y expr={0.3*2^(-3*\thisrowno{0})}]{data/4_1_2_VecLap_sphere/comparison_L2Error_normal_Pk2_Pg3.dat};
      \addplot[draw=purple,thick, mark=square] table[x index=0, y index=3]
      {data/4_1_2_VecLap_sphere/comparison_L2Error_normal_Pk2_Pg3.dat}; 
      \addplot[draw=blue, thick, mark=otimes] table[x index=0, y index=4]
      {data/4_1_2_VecLap_sphere/comparison_L2Error_normal_Pk2_Pg3.dat};
      \legend{
        \footnotesize $\mathcal{O}(h^3)$ \\
      }
      \draw[thick] (axis description cs:0,0) -- (axis description cs:1,0); 

\end{groupplot}
\node[scale=0.925] at (5.6,2.5) { $(k,k_g)\!=\!(2,3)$};
\end{tikzpicture}\hspace*{-0.5cm}
\end{center}
\vspace*{-0.9cm}
\begin{center}
\hspace*{-0.8cm}
\begin{tikzpicture}
  \begin{semilogyaxis}[
    grid=both,
    axis x line=box,
    axis line style={-},
    width=6cm,
    xmin=0, xmax=4,
    ytick={1e-8,1e-7,1e-6,1e-5,1e-4,1e-3,1e-2,1e-1,1e0},
    yticklabels={,$10^{-7}$,,$10^{-5}$,,$10^{-3}$,,$10^{-1}$,},
    xlabel=refinement level $L$,
    xtick={0,1,2,3,4},
    legend style={at={(0.95,0.95)}, fill=white, draw=black},
    ]

    \addplot[dotted,thick] table[x index =0,
    y expr={2e-3*2^(-4*\thisrowno{0})}]{data/4_1_2_VecLap_sphere/comparison_H1Error_Pk4_Pg5.dat};
    \addplot[draw=black, thick, mark=*] table[x index=0, y index=5]
    {data/4_1_2_VecLap_sphere/comparison_H1Error_Pk4_Pg5.dat}; 
    \addplot[draw=green!50!black, thick, mark=star] table[x index=0, y index=6]
    {data/4_1_2_VecLap_sphere/comparison_H1Error_Pk4_Pg5.dat}; 
    \addplot[draw=purple,thick, mark=square] table[x index=0, y index=3]
    {data/4_1_2_VecLap_sphere/comparison_H1Error_Pk4_Pg5.dat}; 
    \addplot[draw=blue, thick, mark=otimes] table[x index=0, y index=4]
    {data/4_1_2_VecLap_sphere/comparison_H1Error_Pk4_Pg5.dat};
    \legend{
      \footnotesize $\mathcal{O}(h^4)$ \\
    }
\end{semilogyaxis}
\end{tikzpicture}
\hspace*{-0.435cm}
\begin{tikzpicture}
  \begin{semilogyaxis}[
    grid=both,
    axis x line=box,
    axis line style={-},
    width=6cm,
    xmin=0, xmax=4,
    xlabel=refinement level $L$,
    ytick={1e-10,1e-9,1e-8,1e-7,1e-6,1e-5,1e-4,1e-3,1e-2,1e-1},
    yticklabels={,$10^{-9}$,,$10^{-7}$,,$10^{-5}$,,$10^{-3}$,,$10^{-1}$},
    xtick={0,1,2,3,4},
    legend style={at={(0.95,0.95)}, fill=white, draw=black},
    ]
    \addplot[dotted,thick] table[x index =0,
    y expr={2e-4*2^(-5*\thisrowno{0})}]{data/4_1_2_VecLap_sphere/comparison_L2Error_tang_Pk4_Pg5.dat};
    \addplot[draw=black, thick, mark=*] table[x index=0, y index=5]
    {data/4_1_2_VecLap_sphere/comparison_L2Error_tang_Pk4_Pg5.dat}; 
    \addplot[draw=green!50!black, thick, mark=star] table[x index=0, y index=6]
    {data/4_1_2_VecLap_sphere/comparison_L2Error_tang_Pk4_Pg5.dat}; 
    \addplot[draw=purple,thick, mark=square] table[x index=0, y index=3]
    {data/4_1_2_VecLap_sphere/comparison_L2Error_tang_Pk4_Pg5.dat}; 
    \addplot[draw=blue, thick, mark=otimes] table[x index=0, y index=4]
    {data/4_1_2_VecLap_sphere/comparison_L2Error_tang_Pk4_Pg5.dat};
    \legend{
      \footnotesize $\mathcal{O}(h^5)$ \\
    }
\end{semilogyaxis}
\end{tikzpicture}
\hspace*{-0.435cm}
\begin{tikzpicture}
\begin{groupplot}[
    group style={
        group name=my fancy plots,
        group size=1 by 2,
        xticklabels at=edge bottom,
        vertical sep=0pt
    },
    width=6cm,
    xmin=0, xmax=4,
]
\nextgroupplot[ymode=log, ymin=1e-10,ymax=1e0,
      grid=both,
      axis x line=box,
      separate axis lines,
      every outer x axis line/.append style={-,draw=none},
      every outer y axis line/.append style={-},
      ytick={1e-9,1e-8,1e-7,1e-6,1e-5,1e-4,1e-3,1e-2,1e-1,1e0},
      yticklabels={,,$10^{-7}$,,$10^{-5}$,,$10^{-3}$,,$10^{-1}$,},
      xtick={},
      axis y discontinuity=parallel, 
      height=4.7cm,
      legend style={at={(1.5,0.7)}},
      ]
      \addplot[draw=black, thick, mark=*] table[x index=0, y index=5]
      {data/4_1_2_VecLap_sphere/comparison_L2Error_normal_Pk4_Pg5.dat}; 
      \addplot[draw=green!50!black, thick, mark=star] table[x index=0, y index=6]
      {data/4_1_2_VecLap_sphere/comparison_L2Error_normal_Pk4_Pg5.dat}; 
      \addplot[draw=purple,thick, mark=square] table[x index=0, y index=3]
      {data/4_1_2_VecLap_sphere/comparison_L2Error_normal_Pk4_Pg5.dat}; 
      \addplot[draw=blue, thick, mark=otimes] table[x index=0, y index=4]
      {data/4_1_2_VecLap_sphere/comparison_L2Error_normal_Pk4_Pg5.dat};
      \addplot[dotted,thick] table[x index =0,
      y expr={0.025*2^(-5*\thisrowno{0})}]{data/4_1_2_VecLap_sphere/comparison_L2Error_normal_Pk4_Pg5.dat};
      \legend{
        \footnotesize \texttt{H1-L} \\
        \footnotesize \texttt{H1-P} \\
        \footnotesize \texttt{HDG}  \\
        \footnotesize \texttt{Hdiv-HDG} \\
      }
      \draw[thick] (axis description cs:0,1) -- (axis description cs:1,1); 

\nextgroupplot[ymode=log, ymin=7e-17,ymax=1e-15, 
      grid=both,
      axis x line=box,
      separate axis lines,
      every outer x axis line/.append style={-,draw=none},
      every outer y axis line/.append style={-},
      ytick={7e-17,8e-17,9e-17,1e-16,2e-16,3e-16,4e-16,5e-16,6e-16,7e-16,8e-16,9e-16,1e-15},
      yticklabels={,,,$10^{-16}$,,,,,,,,,},
      xlabel=refinement level $L$,
      xtick={0,1,2,3,4},
      height=2.05cm,
      legend style={at={(0.95,7.2)}, fill=white, draw=black},
      ]
      \addplot[dotted,thick] table[x index =0,
      y expr={0.3*2^(-5*\thisrowno{0})}]{data/4_1_2_VecLap_sphere/comparison_L2Error_normal_Pk4_Pg5.dat};
      \addplot[draw=purple,thick, mark=square] table[x index=0, y index=3]
      {data/4_1_2_VecLap_sphere/comparison_L2Error_normal_Pk4_Pg5.dat}; 
      \addplot[draw=blue, thick, mark=otimes] table[x index=0, y index=4]
      {data/4_1_2_VecLap_sphere/comparison_L2Error_normal_Pk4_Pg5.dat};
      \legend{
        \footnotesize $\mathcal{O}(h^5)$ \\
      }
      \draw[thick] (axis description cs:0,0) -- (axis description cs:1,0); 

\end{groupplot}
\node[scale=0.925] at (5.6,2.5) { $(k,k_g)\!=\!(4,5)$};
\end{tikzpicture}\hspace*{-0.5cm}
\end{center}
\vspace*{-0.5cm}
\caption{Errors for Vector Laplacian on the sphere for four different discretization methods on 5 successively refined meshes (uniform refinements) for different discretization and geometry orders.}
\vspace*{-0.3cm}
\label{fig:numex:VecLap1}
\end{figure}
In Fig. \ref{fig:numex:VecLap1} we display the $H^1$-semi-norm, the $L^2$-norm of the tangential and the $L^2$-norm of the normal component of the error $u^e - u_h$ on $\Gamma_h$. Different configurations are considered.

In the first row we set geometry and FE approximation order to one, $k_g=k=1$.
First of all, we observe that the normal component is exactly zero only for the (H)DG methods. This is in agreement with our expectations as the $H^1$-conforming methods only implement the tangential constraint weakly whereas the (H)DG methods fulfill the constraint by construction. 
Further, we notice that the convergence rates of all methods are sub-optimal in all norms except for the $H^1$-norm. For the $H^1$-norm all methods are converging optimally with order one except for \texttt{H1-P} which is not converging at all.
This can be explained by the fact that the geometry approximation is only of first order so that the approximation of the normal is only piecewise constant. Elements surrounding the same vertex will have different normals and the penalty diminishes all corresponding normal components at the same time which results in a severe case of locking.
In the Lagrange multiplier formulation \texttt{H1-L} the constraint involves an averaging over a vertex patch which circumvents this locking effect. However, for all methods we observe that the low order approximation of the geometry results in a loss of accuracy by one order in the $L^2$ norms. This need for a higher order accurate geometry approximation is not surprising and has already been discussed in e.g. Refs. \cite{hansbo2016analysis,jankuhn2019trace}. In Fig. \ref{fig:sphere} a comparison of the pictures on the right illustrates the difference between $k_g=1$ and $k_g=2$ for $k=1$. 
In the rows two to four of Fig. \ref{fig:numex:VecLap1} we fix $k_g=k+1$, i.e. a superparametric geometry approximation, and choose $k\in \{1,2,4\}$. For all methods we now observe optimal order of convergence, i.e. $\mathcal{O}(h^k)$ in the $H^1$-semi-norms and $\mathcal{O}(h^{k+1})$ in the $L^2$-norms. For $k=4$, the results of \texttt{H1-P} are not behaving optimally on the finest meshes. This is probably due to the ill-conditioning stemming from the penalty parameter of magnitude $10 h^{-5}$. Let us further mention that we tried (for all methods) $k_g=k$ and obtained essentially the same results, i.e. in this example it seems that the superparametric approximation is not necessary as long as $k > 1$. 

Overall, we can conclude that all methods perform similarly well for this example. The (H)DG methods are obviously -- by construction -- perfect in approximating the tangential constraint, but for the other error measures there is no  significant difference.

\paragraph{Computational aspects}
In this paragraph we want to discuss and compare computational aspects of the different discretizaton methods under consideration for the Vector Laplacian.
Starting with the $H^1$-conforming methods, one immediately notices that the advantage of these methods is their simplicity. Lagrangian finite elements are available in most finite element packages and hence, a realization of these methods is comparably simple. Then again, the two-dimensional vector field is approximated with a three-dimensional vector field which can be seen as undesirable in terms of the computational overhead.
To overcome this issue we introduced tangential finite elements. These however came at the price of abandoning $H^1$-conformity which requires the implementation of weak continuity (at least in tangential direction) through the discrete variational formulation. This approach results in DG methods which come at the disadvantage of introducing more unknowns and more couplings. To alleviate these costs we also discussed the use of hybrid versions of DG methods.
\begin{table}[b]
  \vspace*{-0.5cm}
  \begin{center}
    \scalebox{0.85}{
    \begin{tabular}{r @{~~~~} r@{~~~~} r@{~~~}r@{~~~}r@{~~~}r@{~~~}r @{~~~~}r @{~~~~} r@{~~~}r@{~~~}r@{~~~}r@{~~~}r @{~~~~}r @{~~~~} r@{~~~}r@{~~~}r@{~~~}r@{~~~}r@{\hspace*{0.2cm}}}
      \toprule
      && \multicolumn{5}{c}{\texttt{dof}}
      && \multicolumn{5}{c}{\texttt{gdof}}
      && \multicolumn{5}{c}{\texttt{nze}} 
      \\
      $k$
      && 1 & 2 & 3 & 4 & 5
      && 1 & 2 & 3 & 4 & 5
      && 1 & 2 & 3 & 4 & 5
      \\
      \midrule 
      \texttt{DG}
      && \numQ{193536}   & \numQ{387072} & \numQ{645120} & \numQ{967680} & \numQ{1354752} 
      && \numQ{193536}   & \numQ{387072} & \numQ{645120} & \numQ{967680} & \numQ{1354752} 
      && \numQQ{4644864} & \numQQ{18579456}& \numQQ{51609600}& \numQQ{116121600}& \numQQ{227598336} 
      \\
      \rowcolor[gray]{0.9}
      \texttt{HDG}
      && \numQ{387072}   & \numQ{677376} & \numQ{1032192} & \numQ{1451520} & \numQ{1935360} 
      && \numQ{193536}   & \numQ{290304} & \numQ{387072} & \numQb{483840} & \numQb{580608} 
      && \numQQ{3870720} & \numQQ{8709120}& \numQQb{15482880}& \numQQb{24192000}& \numQQb{34836480}
      \\
      \texttt{Hdiv-DG}
      && \numQ{96768}   & \numQ{241920} & \numQ{451584} & \numQb{725760} & \numQb{1064448} 
      && \numQ{96768}   & \numQ{241920} & \numQ{451584} & \numQ{725760} & \numQ{1064448} 
      && \numQQ{2515968} & \numQQ{12047616}& \numQQ{36900864}& \numQQ{88300800}& \numQQ{180569088}
      \\
      \rowcolor[gray]{0.9}
      \hspace*{-0.2cm}\texttt{Hdiv-HDG}
      && \numQ{193536}   & \numQ{387072} & \numQ{645120} & \numQ{967680} & \numQ{1354752} 
      && \numQ{193536}   & \numQ{290304} & \numQ{387072} & \numQb{483840} & \numQb{580608} 
      && \numQQ{3870720} & \numQQ{8709120}& \numQQb{15482880}& \numQQb{24192000}& \numQQb{34836480}
      \\
      \texttt{H1-L}
      && \numQ{64520}    & \numQ{258056} & \numQ{580616} & \numQ{1032200} & \numQ{1612808} 
      && \numQ{64520}    & \numQ{258056} & \numQ{451592} & \numQ{645128} & \numQ{838664} 
      && \numQ{1693470}  & \numQQ{11128350}& \numQQ{29675552}& \numQQ{48109856}& \numQQ{88510496}
      \\
      \rowcolor[gray]{0.9}
      \texttt{H1-P} 
      && \numQb{48390}    & \numQb{193542} & \numQb{435462} & \numQ{774150} & \numQ{1209606} 
      && \numQb{48390}    & \numQb{193542} & \numQb{338694} & \numQb{483846} & \numQ{628998} 
      && \numQQb{1016082} & \numQQb{6677010}& \numQQ{16692498}& \numQQ{31062546}& \numQQ{49787154}
      \\
      \bottomrule
    \end{tabular}
  }
\end{center}\vspace*{-0.5cm}
\caption{Comparison of different computational quantities (\texttt{dof}: degrees of freedom, \texttt{gdof}: global degrees of freedom that remain after static condensation (if applicable), \texttt{nze}: non-zero entries in system (Schur complement) matrix) for the different schemes on finest level. Numbers in green indicate that the corresponding method yields the best values in the current column.}
\label{tab:compasp}
\end{table}
In Table \ref{tab:compasp} we compare the six different methods introduced before for polynomial orders $1$ to $5$ on the finest mesh $L=4$ of the previous example (Vector Laplacian on the sphere). While the $H^1$-conforming and the HDG methods are exactly those investigated in the previous paragraph, the methods \texttt{DG} and \texttt{Hdiv-DG} are only considered with respect to their computational costs, here. For the HDG methods and the $H^1$-conforming methods we apply static condensation in an element-by-element fashion. 
As measures for the computational costs we take the number of degrees of freedom (\texttt{dof}), the number of global \texttt{dof} that remain in the Schur complement after static condensation (\texttt{gdof}) and the number of non-zero entries (\texttt{nze}) in the Schur complement. 

Let us first take a look at the \texttt{dof} measure. Here, the \texttt{H1-P} method has the smallest number of \texttt{dof} in the low order case $k \leq 3$. The additional \texttt{dof} for approximating a three-dimensional vector field are still less than those obtained from approximating a two-dimensional vector field with discontinuous piecewise polynomials. Only for $k \geq 4$ the \texttt{Hdiv-DG} method with normal-continuity results in less \texttt{dof}. When normal-continuity is also broken up, i.e. when going to \texttt{DG} it requires at least order $k \geq 6$ to beat \texttt{H1-P} in terms of unknowns. When further going to HDG methods there is obviously no advantage over DG methods in terms of \texttt{dof} as only additional unknowns are introduced. As we can apply static condensation with HDG methods it is worth taking a look at those \texttt{dof} that remain in the Schur complement.
The number of \texttt{gdof} can be reduced for \texttt{DG} for all polynomial degrees whereas the step from \texttt{Hdiv-DG} to \texttt{Hdiv-HDG} pays off in terms of \texttt{gdof} for $k \geq 3$.
We observe that \texttt{gdof} is the same for \texttt{HDG} and \texttt{Hdiv-HDG}.
When considering \texttt{nze} the picture shifts even further towards \texttt{HDG} and \texttt{Hdiv-HDG}. For $k=2$ already the methods \texttt{HDG} and \texttt{Hdiv-HDG} outperform the other two DG methods and for $k \geq 3$ the method generates even less unknowns than the $H^1$-conforming methods. We conclude that the hybrid formulations with tangential fields are benefitial not only because of their additional structure properties, but also computationally advantageous when going for higher order discretizations.

\begin{remark}[HDG superconvergence]
  Let us note that we did not consider further tweaks of the HDG methods related to the aspect of superconvergence,
cf. the paragraph on efficiency aspects and superconvergence in Section~\ref{sec::DGvsHDG}. With a corresponding modification the polynomial degree on the facets can be reduced by one order while keeping the same order of accuracy (for diffusion dominated problems), i.e. the costs in terms of \texttt{gdof} and \texttt{nze} for a method with accuracy $k$ (in the $H^1$-norm) are the same of those HDG methods without this modification with order $k - 1$. In the lowest order case we obtatin \texttt{gdof}(\texttt{HDG})=\texttt{gdof}(\texttt{Hdiv-HDG})=\numQ{96768}~
and 
\texttt{nze}(\texttt{HDG})=\texttt{nze}(\texttt{Hdiv-HDG})=\numQQ{967680}~ and hence, the HDG schemes are level with \texttt{H1-P} in terms of \texttt{gdof}, but already more efficient in terms of \texttt{nze}.
\end{remark}


\begin{figure}[h]
  \vspace*{-0.5cm}
  \begin{center}
    \includegraphics[width=.6\textwidth,trim= 2.75cm 8cm 3cm 3.25cm, clip=True]{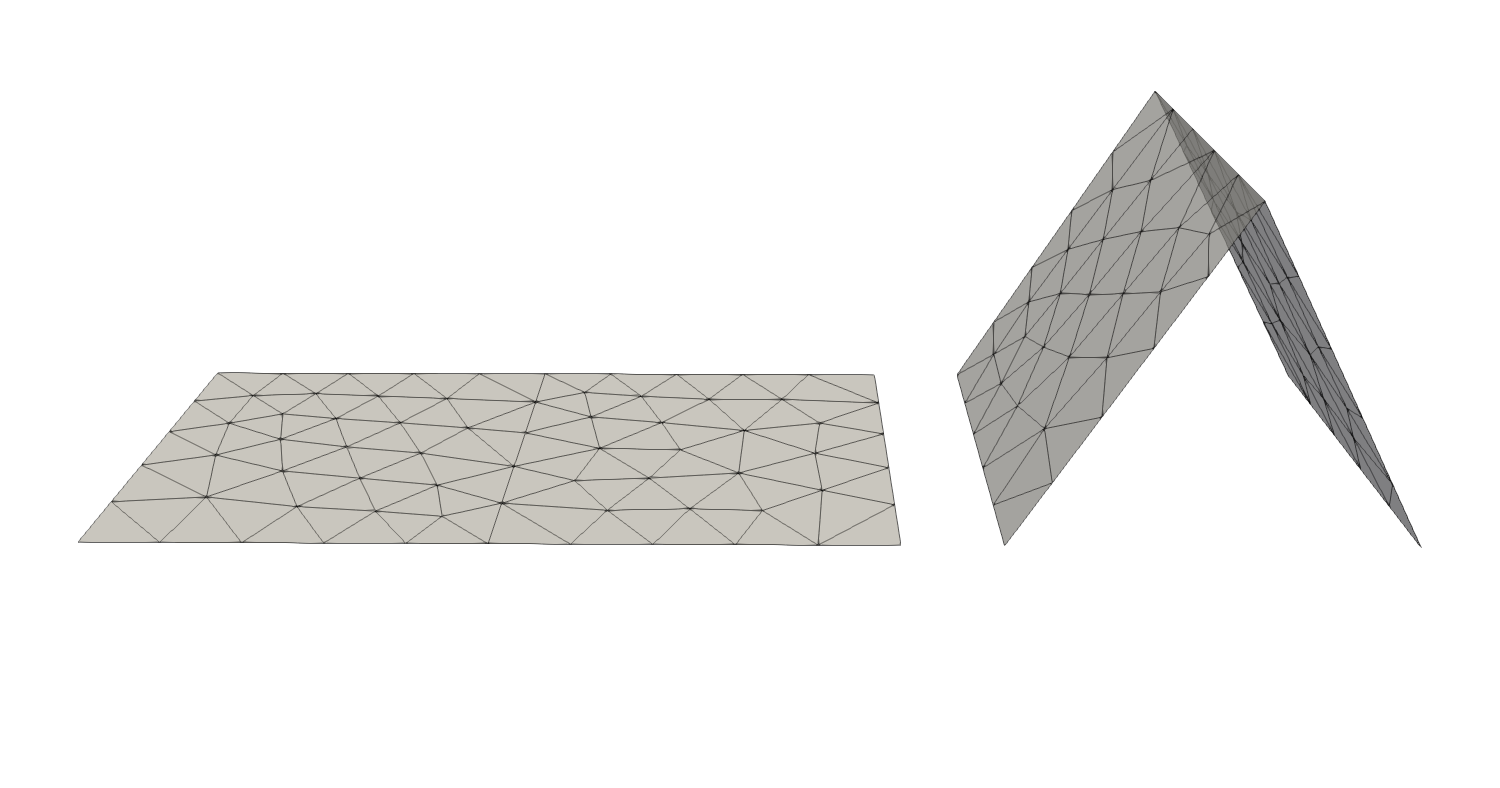}
  \end{center}
  \vspace*{-0.5cm}
  \caption{Geometries $\hat \Gamma$ (left) and $\Gamma$ (right) and coarsest mesh ($L=0$) for the examples in Section \ref{sec::houseofcards}.}
  \vspace*{-1cm}
  \label{fig:houseofcards}
\end{figure}

\subsubsection{Vector Laplacian in the plane and on a house of cards} \label{sec::houseofcards}
In this example we consider two configurations. First, we consider a flat domain case with
$\hat{\Gamma} := (0,2) \times (0,1)$  where we pose the Vector-Laplacian problem
\eqref{prob:veclap} and replace \eqref{prob:veclap-b} with inhomogeneous Dirichlet boundary conditions $u = g$ on $\partial \Gamma$.
As a second example we fold the geometry in the middle and lift it up to a house of two cards:
$$
\Gamma = \Phi_H(\hat \Gamma) \quad \text{with}\quad \Phi_H(\hat x,\hat y) = \left\{
  \begin{array}{ll}
    (\hat x \cdot W,\hat y, H\cdot \hat x) & \text{ if } \hat x \leq 1, \\
    (\hat x \cdot W,\hat y, H\cdot (2-\hat x)) & \text{ if } \hat x > 1,
  \end{array}
\right.
$$
where $H$ is the height and $W$ is the width of one card with $H^2 + W^2 = 1$ so that $F_H := \Phi_H' \in \mathbb{R}^{3\times 2}$ is a length-preserving map with $F_H^T F_H = I \in \rr^{2\times 2}$, cf. Figure \ref{fig:houseofcards} for a sketch. We define with $\hat \t_1 = e_1 = (1,0)$ and $\hat \t_2 = e_2 = (0,1)$ the tangent unit vectors of $\hat \Gamma$. Then, $\t_i = F_H ~ \hat \t_i,~i=1,2$ are the tangent unit vectors for $\Gamma$ a.e., i.e. we can interprete $F_H$ --- and due to $J=1$ also the Piola transformation --- as the operator that realizes a basis transformation from one tangent space to another. 

On $\hat \Gamma$ we pose the Vector Laplacian with r.h.s. $\hat f$ and boundary data $\hat g$ such that the exact solution to \eqref{prob:veclap} on $\hat \Gamma$ is
$$
\hat u(\hat x,\hat y) = (\sin(\pi \hat x) + \cos(\pi \hat y)) ~ \hat\t_1 + (\cos(\pi \hat x) + \sin(\pi  \hat y)) ~ \hat\t_2.
$$
For the same problem on $\Gamma$ (with properly transformed data $f = F_H ~ \hat f \circ \Phi_H^{-1}$, $g = F_H ~ \hat g \circ \Phi_H^{-1}$) the (extended) solution is simply
\begin{equation} \label{eq:veclap2:solrep}
  u(x,y,z) = F_H ~ \hat u \circ \Phi_H^{-1} = (\sin(\pi \hat x) + \cos(\pi \hat y)) ~ \t_1 + (\cos(\pi \hat x) + \sin(\pi  \hat y)) ~ \t_2 \quad \text{ with } \quad(\hat{x},\hat{y}) = \Phi_H^{-1}(x,y,z) = (x/W,y).
\end{equation}
Next, we discuss how this characterization of the solution translates to the discrete level. We note that $\Phi_H$ is piecewise affine, so that the Jacobian is piecewise constant. Let $\hat u_h$ be the discrete solution to the Vector Laplace problem on $\hat \Gamma$ and $u_h = \mathcal{P} \hat u_h$. Then, we can expand \\[-3.5ex]
$$
\hat u_h = \sum_{i=1}^2 \hat u_{h,i} \hat \t_i, \qquad 
u_h := \overbrace{J^{-1}}^{=1} F_H \hat u_h \circ \Phi_H^{-1} = \sum_i \overbrace{\hat u_{h,i} \circ \Phi_H^{-1}}^{=:u_{h,i}} \overbrace{ F_H \hat{\t}_i}^{=\t_i} = \sum_i u_{h,i} \t_i.
$$
and obtain the following relation for the surface gradients: \\[-3.5ex]
$$
\nabla_{\hat \Gamma} \hat u_h = \nabla \hat u_h = \sum_{i,j=1}^2 \frac{\partial \hat u_{h,i}}{\partial \hat \t_j} \hat \t_i \otimes \hat \t_j, \qquad
\nabla_\Gamma u_h = \sum_{i,j} \frac{\partial u_{h,i}}{\partial \t_j} \t_i \otimes \t_j =  \sum_{i,j} \frac{\partial \hat u_{h,i}}{\partial \hat \t_j} \circ \Phi_H^{-1} ~ \t_i \otimes \t_j = F_H \overbrace{ \nabla \hat{u}_h}^{\in \mathbb{R}^{2 \times 2}} F_H^T \in \mathbb{R}^{3 \times 3}.
$$
One may ask, if after applying the Piola transformation also on the test function, $u_h$ solves the discrete Vector Laplace problem on $\Gamma$ for the DG methods. One can easily check that this is true for the DG methods discussed here as there holds for instance 
$$
\nabla_\Gamma u_h \colon \nabla_\Gamma v_h = (F_H \nabla \hat u_h  F_H^T) \colon (F_H  \nabla \hat v_h  F_H^T)
= (\underbrace{F_H^T  F_H}_{=I \in \mathbb{R}^2}  \nabla \hat u_h  \underbrace{F_H^T  F_H}_{=I \in \mathbb{R}^2}) \colon \nabla \hat v_h = \nabla \hat u_h \colon \nabla \hat v_h. 
$$
Hence, $u_h$ solves the discrete Vector Laplace problem on $\Gamma$ for the DG discretizations if $\hat u_h$ solves the discrete problem on $\hat \Gamma$. 
For the $H^1$-conforming methods this does not apply as --- for $\hat w_h$ being the discrete $H^1$-conforming solution to the Vector Laplacian on $\hat \Gamma$ --- the mapped function $w_h = \mathcal{P} \hat w_h$ will not be in $[S_h^k]^3$ for $H>0$. By construction $\hat w_h$ is a continuous (2D) vector field s.t. after the transformation with $\Phi_H$, which has a discontinuous Jacobian, $w_h$ is discontinuous and hence not included in $[S_h^k]^3$.

To illustrate this effect we consider the aforementioned methods with $k=3$ and $H \in \{0, \sqrt{3/4} \}$.
In Figure \ref{fig:numex:VecLap2} we display the error behavior for the two situations on 5 consecutive meshes, starting from the mesh displayed in Figure \ref{fig:houseofcards}.
Note that the geometry is piecewise planar, so that $\Gamma_h = \Gamma$.
For $\hat \Gamma$ we observe that all considered methods behave essentially the same. The deviation between the errors is only marginal. All methods yield an exactly tangential field $u_h \cdot \n_h = 0$ as the equations for the $z$-component of the $H^1$-conforming method completely seperates from those in $x$ and $y$-direction.
When going to $\Gamma$ we observe that the $H^1$-conforming methods fail to converge. This is easily explained by the fact that the solution is discontinuous when represented in the embedding space $\mathbb{R}^3$. Considering the representation of the solution \eqref{eq:veclap2:solrep}, we notice that
the Piola mapping incorporates the tangential field automatically and we observe that the convergence behaviour of the DG methods stay the same when going from $\hat{\Gamma}$ to $\Gamma$. Actually, the numbers used in the plots are identical (up to round-off errors).

\begin{figure}
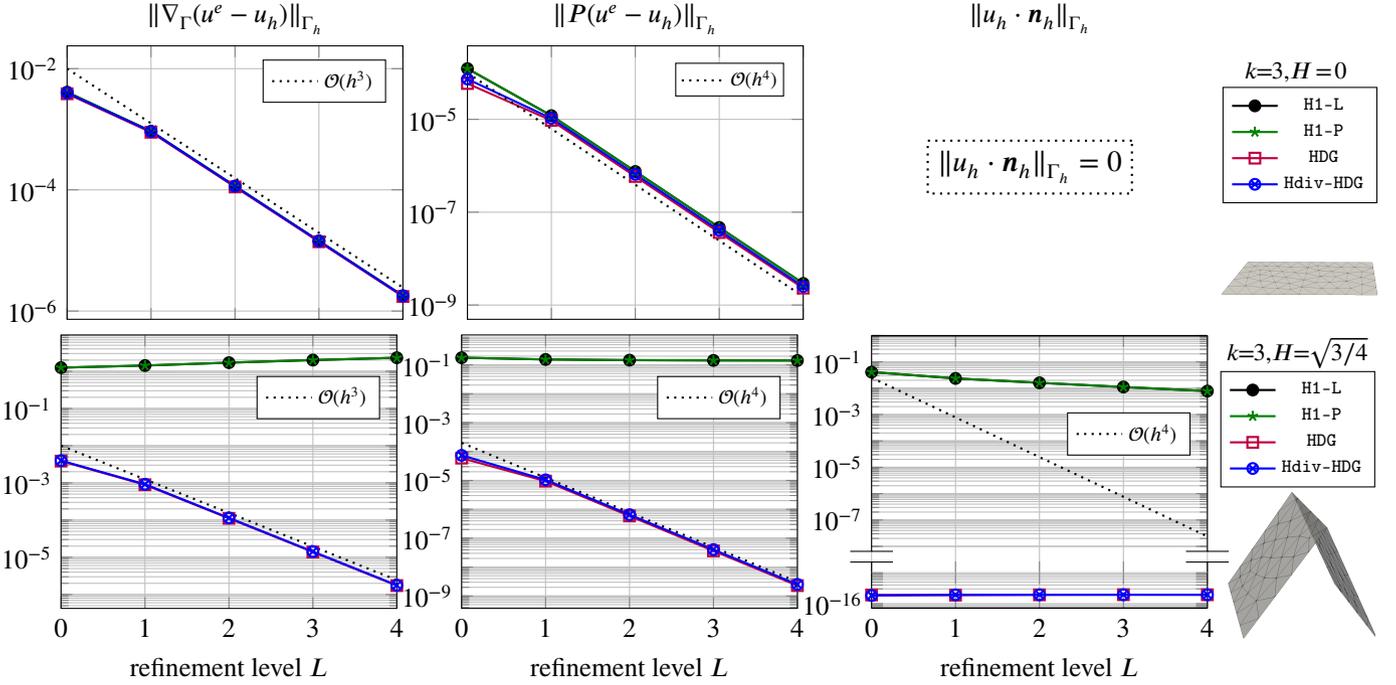

\begin{center}
\hspace*{-0.8cm}
\begin{tikzpicture}
  \begin{semilogyaxis}[
    grid=both,
    axis x line=box,
    axis line style={-},
    width=6cm,
    xmin=0, xmax=4,
    xticklabels={},
    title={$\Vert \nabla_\Gamma (u^e - u_h) \Vert_{\Gamma_h}$},
    title style={at={(axis description cs:0.5,0.95)},anchor=south},
    legend style={at={(0.95,0.95)}, fill=white, draw=black},
    ]

    \addplot[dotted,thick] table[x index =0,
    y expr={0.01*2^(-3*\thisrowno{0})}]
    {data/4_1_3_VecLap_house_of_cards/comparison_H1Error_P3_H0.dat};
    \addplot[draw=black, thick, mark=*] table[x index=0, y index=3]
    {data/4_1_3_VecLap_house_of_cards/comparison_H1Error_P3_H0.dat}; 
    \addplot[draw=green!50!black, thick, mark=star] table[x index=0, y index=4]
    {data/4_1_3_VecLap_house_of_cards/comparison_H1Error_P3_H0.dat}; 
    \addplot[draw=purple,thick, mark=square] table[x index=0, y index=1]
    {data/4_1_3_VecLap_house_of_cards/comparison_H1Error_P3_H0.dat}; 
    \addplot[draw=blue, thick, mark=otimes] table[x index=0, y index=2]
    {data/4_1_3_VecLap_house_of_cards/comparison_H1Error_P3_H0.dat};
    \legend{
      \footnotesize $\mathcal{O}(h^3)$ \\
    }
\end{semilogyaxis}
\end{tikzpicture}
\hspace*{-0.35cm}
\begin{tikzpicture}
  \begin{semilogyaxis}[
    grid=both,
    axis x line=box,
    axis line style={-},
    width=6cm,
    xmin=0, xmax=4,
    xticklabels={},
    title={$\Vert P (u^e - u_h) \Vert_{\Gamma_h}$},
    title style={at={(axis description cs:0.5,0.95)},anchor=south},
    legend style={at={(2.71,0.85)}},
    ]
    \addplot[draw=black, thick, mark=*] table[x index=0, y index=3]
    {data/4_1_3_VecLap_house_of_cards/comparison_L2Error_tang_P3_H0.dat}; 
    \addplot[draw=green!50!black, thick, mark=star] table[x index=0, y index=4]
    {data/4_1_3_VecLap_house_of_cards/comparison_L2Error_tang_P3_H0.dat}; 
    \addplot[draw=purple,thick, mark=square] table[x index=0, y index=1]
    {data/4_1_3_VecLap_house_of_cards/comparison_L2Error_tang_P3_H0.dat}; 
    \addplot[draw=blue, thick, mark=otimes] table[x index=0, y index=2]
    {data/4_1_3_VecLap_house_of_cards/comparison_L2Error_tang_P3_H0.dat};
    \addplot[dotted,thick] table[x index =0,
    y expr={1e-4*2^(-4*\thisrowno{0})}]{data/4_1_3_VecLap_house_of_cards/comparison_L2Error_tang_P3_H0.dat};
      \legend{
        \footnotesize \texttt{H1-L} \\
        \footnotesize \texttt{H1-P} \\
        \footnotesize \texttt{HDG}  \\
        \footnotesize \texttt{Hdiv-HDG} \\
      }
\end{semilogyaxis}
\node[scale=0.725,draw=black,solid,fill=white] at (3.45,3.15) { $\quad~~\quad\mathcal{O}(h^4)$};
\draw[scale=0.8,dotted,thick] (3.525,3.925) -- (4.225,3.925);

\node[scale=0.9] at (10.9,3.25) { $k\!\!=\!\!3$,$H\!=\!0$};
\node[scale=1.1,draw=black,dotted,thick] at (7.4,2) { $\Vert u_h \cdot \n_h \Vert_{\Gamma_h} = 0$};
\node[scale=1] at (7.4,3.95) { $\Vert u_h \cdot \n_h \Vert_{\Gamma_h}$};

\node[] at (11.0,0.7) {
  \includegraphics[width=.12\textwidth,trim= 2.75cm 5cm 19.8cm 5cm, clip=True]
  {graphics/house_of_cards.png}
};
\end{tikzpicture}
\hspace*{-0.8cm}
\end{center}
\vspace*{-0.85cm}
\begin{center}
\hspace*{-0.8cm}
\begin{tikzpicture}
  \begin{semilogyaxis}[
    grid=both,
    axis x line=box,
    axis line style={-},
    width=6cm,
    xmin=0, xmax=4,
    ytick={1e-8,1e-7,1e-6,1e-5,1e-4,1e-3,1e-2,1e-1,1e0,1e1},
    yticklabels={,$10^{-7}$,,$10^{-5}$,,$10^{-3}$,,$10^{-1}$,,},
    xlabel=refinement level $L$,
    xtick={0,1,2,3,4},
    legend style={at={(0.95,0.85)}, fill=white, draw=black},
    ]

    \addplot[dotted,thick] table[x index =0,
    y expr={1e-2*2^(-3*\thisrowno{0})}]
    {data/4_1_3_VecLap_house_of_cards/comparison_H1Error_P3_H0.8660254037844386.dat};
    \addplot[draw=black, thick, mark=*] table[x index=0, y index=3]
    {data/4_1_3_VecLap_house_of_cards/comparison_H1Error_P3_H0.8660254037844386.dat};
    \addplot[draw=green!50!black, thick, mark=star] table[x index=0, y index=4]
    {data/4_1_3_VecLap_house_of_cards/comparison_H1Error_P3_H0.8660254037844386.dat};
    \addplot[draw=purple,thick, mark=square] table[x index=0, y index=1]
    {data/4_1_3_VecLap_house_of_cards/comparison_H1Error_P3_H0.8660254037844386.dat};
    \addplot[draw=blue, thick, mark=otimes] table[x index=0, y index=2]
    {data/4_1_3_VecLap_house_of_cards/comparison_H1Error_P3_H0.8660254037844386.dat};
    \legend{
      \footnotesize $\mathcal{O}(h^3)$ \\
    }
\end{semilogyaxis}
\end{tikzpicture}
\hspace*{-0.435cm}
\begin{tikzpicture}
  \begin{semilogyaxis}[
    grid=both,
    axis x line=box,
    axis line style={-},
    width=6cm,
    xmin=0, xmax=4,
    xlabel=refinement level $L$,
    ytick={1e-10,1e-9,1e-8,1e-7,1e-6,1e-5,1e-4,1e-3,1e-2,1e-1},
    yticklabels={,$10^{-9}$,,$10^{-7}$,,$10^{-5}$,,$10^{-3}$,,$10^{-1}$},
    xtick={0,1,2,3,4},
    legend style={at={(0.95,0.85)}, fill=white, draw=black},
    ]
    \addplot[dotted,thick] table[x index =0, y expr={2e-4*2^(-4*\thisrowno{0})}]
    {data/4_1_3_VecLap_house_of_cards/comparison_L2Error_tang_P3_H0.8660254037844386.dat};
    \addplot[draw=black, thick, mark=*] table[x index=0, y index=3]
    {data/4_1_3_VecLap_house_of_cards/comparison_L2Error_tang_P3_H0.8660254037844386.dat};
    \addplot[draw=green!50!black, thick, mark=star] table[x index=0, y index=4]
    {data/4_1_3_VecLap_house_of_cards/comparison_L2Error_tang_P3_H0.8660254037844386.dat};
    \addplot[draw=purple,thick, mark=square] table[x index=0, y index=1]
    {data/4_1_3_VecLap_house_of_cards/comparison_L2Error_tang_P3_H0.8660254037844386.dat};
    \addplot[draw=blue, thick, mark=otimes] table[x index=0, y index=2]
    {data/4_1_3_VecLap_house_of_cards/comparison_L2Error_tang_P3_H0.8660254037844386.dat};
    \legend{
      \footnotesize $\mathcal{O}(h^4)$ \\
    }
\end{semilogyaxis}
\end{tikzpicture}
\hspace*{-0.435cm}
\begin{tikzpicture}
\begin{groupplot}[
    group style={
        group name=my fancy plots,
        group size=1 by 2,
        xticklabels at=edge bottom,
        vertical sep=0pt
    },
    width=6cm,
    xmin=0, xmax=4,
]
\nextgroupplot[ymode=log, ymin=1e-9,ymax=1e0,
      grid=both,
      axis x line=box,
      separate axis lines,
      every outer x axis line/.append style={-,draw=none},
      every outer y axis line/.append style={-},
      ytick={1e-8,1e-7,1e-6,1e-5,1e-4,1e-3,1e-2,1e-1,1e0},
      yticklabels={,$10^{-7}$,,$10^{-5}$,,$10^{-3}$,,$10^{-1}$,},
      xtick={},
      axis y discontinuity=parallel, 
      height=4.7cm,
      legend style={at={(1.5,0.85)}},
      ]
      \addplot[draw=black, thick, mark=*] table[x index=0, y index=3]
      {data/4_1_3_VecLap_house_of_cards/comparison_L2Error_normal_P3_H0.8660254037844386.dat};
      \addplot[draw=green!50!black, thick, mark=star] table[x index=0, y index=4]
      {data/4_1_3_VecLap_house_of_cards/comparison_L2Error_normal_P3_H0.8660254037844386.dat};
      \addplot[draw=purple,thick, mark=square] table[x index=0, y index=1]
      {data/4_1_3_VecLap_house_of_cards/comparison_L2Error_normal_P3_H0.8660254037844386.dat};
      \addplot[draw=blue, thick, mark=otimes] table[x index=0, y index=2]
      {data/4_1_3_VecLap_house_of_cards/comparison_L2Error_normal_P3_H0.8660254037844386.dat};
      \addplot[dotted,thick] table[x index =0,
      y expr={0.025*2^(-5*\thisrowno{0})}]{data/4_1_2_VecLap_sphere/comparison_L2Error_normal_Pk4_Pg5.dat};
      \legend{
        \footnotesize \texttt{H1-L} \\
        \footnotesize \texttt{H1-P} \\
        \footnotesize \texttt{HDG}  \\
        \footnotesize \texttt{Hdiv-HDG} \\
      }
      \draw[thick] (axis description cs:0,1) -- (axis description cs:1,1); 

\nextgroupplot[ymode=log, ymin=7e-17,ymax=1e-15, 
      grid=both,
      axis x line=box,
      separate axis lines,
      every outer x axis line/.append style={-,draw=none},
      every outer y axis line/.append style={-},
      ytick={7e-17,8e-17,9e-17,1e-16,2e-16,3e-16,4e-16,5e-16,6e-16,7e-16,8e-16,9e-16,1e-15},
      yticklabels={,,,$10^{-16}$,,,,,,,,,},
      xlabel=refinement level $L$,
      xtick={0,1,2,3,4},
      height=2.05cm,
      legend style={at={(0.95,5.5)}, fill=white, draw=black},
      ]
      \addplot[dotted,thick] table[x index =0,
      y expr={0.3*2^(-4*\thisrowno{0})}]
      {data/4_1_3_VecLap_house_of_cards/comparison_L2Error_normal_P3_H0.8660254037844386.dat};
      \addplot[draw=purple,thick, mark=square] table[x index=0, y index=1]
      {data/4_1_3_VecLap_house_of_cards/comparison_L2Error_normal_P3_H0.8660254037844386.dat};
      \addplot[draw=blue, thick, mark=otimes] table[x index=0, y index=2]
      {data/4_1_3_VecLap_house_of_cards/comparison_L2Error_normal_P3_H0.8660254037844386.dat};
      \legend{
        \footnotesize $\mathcal{O}(h^4)$ \\
      }
      \draw[thick] (axis description cs:0,0) -- (axis description cs:1,0); 

\end{groupplot}
\node[scale=0.9] at (5.6,2.9) { $k\!\!=\!\!3$,$H\!\!=\!\!\sqrt{3/4}$};

\node[] at (5.55,-0.1) {
  \includegraphics[width=.12\textwidth,trim= 32cm 5cm 3cm 3cm, clip=True]
  {graphics/house_of_cards.png}
};

\end{tikzpicture}\hspace*{-0.5cm}
\end{center}
\vspace*{-0.5cm}
\caption{Error behavior for Vector Laplace problem on the flat (top row) and the house of cards geometry (bottom row) from Section \ref{sec::houseofcards} for four different discretization methods on 5 successively refined meshes (uniform refinements) for fixed discretization order $k=3$ and an exact geometry approximation.}
\label{fig:numex:VecLap2}
\vspace*{-0.25cm}
\end{figure}

\subsection{Stokes} \label{sec:stokesnumex}
Whereas the last section demonstrated the accuracy and advantages of the (H)DG methods introduced in this work for second order elliptic problems on surfaces, we now aim to solve the stationary incompressible Stokes equations, see problem \ref{prob:stokes}. Similarly as in the last section we define a Stokes problem on the flat reference domain $\hat \Gamma = (0,1) \times (0,1)$ and map it isometrically onto a smooth manifold given by half of the surface of an open  cylinder with radius $1/\pi$
\begin{align*}
  \Gamma = \Phi_{1/\pi}(\hat \Gamma) \quad \textrm{with} \quad \Phi_{1/\pi}(\hat x, \hat y) = (\hat x, \sin( (\hat y -1/2) \pi)/ \pi + 1/\pi, \cos( (\hat y -1/2) \pi)/\pi).
\end{align*}
Again, $\Phi_{1/\pi}$ preserves the length, i.e. $\sqrt{\det(F_{1/\pi}^T F_{1/\pi})} = 1$ with $F_{1/\pi} := \nabla\Phi_{1/\pi}$. For a fixed viscosity $\nu$, the reference solutions are given by
\begin{align*}
  \hat u(\hat x, \hat y) = -\frac{\partial \xi}{\partial \hat y} \cdot e_1 +  \frac{\partial \xi}{\partial \hat x} \cdot e_2 \quad \textrm{and} \quad \hat p (\hat x, \hat y) = \hat x^5 + \hat y^5 - 1/3
\end{align*}
with the scalar potential $\xi = \hat x^ 2 (1 - \hat x)^2\hat y^ 2 (1 - \hat y)^2$, and $e_1,e_2$ as before. We set the corresponding right hand side to $\hat f = -2\nu \eps( \hat u) + \nabla \hat p$. Defining the tangential vectors on $\Gamma$ by $\t_i = F_{1/\pi}  e_i, i=1,2$, the exact solutions are given by (see Figure~\ref{fig:stokesex})
\begin{align} \label{stokesexact}
u(x,y,z) = -\frac{\partial \xi}{\partial \hat y} \cdot \t_1 +  \frac{\partial \xi}{\partial \hat x} \cdot \t_2 \quad \textrm{and} \quad p(x,y,z) = \hat x^5 + \hat y^5 - 1/3,
\end{align}
with $(\hat x, \hat y) = \Phi_{1/\pi}^{-1}(x,y,z) = (x, \arcsin(\pi y-1)/\pi+1/2)$, and the right hand side
\begin{align} \label{stokerhssexact}
  f(x,y,z) = f(\hat x, \hat y)_1 \cdot \t_1 +f(\hat x, \hat y)_2 \cdot \t_2.
\end{align}%
\begin{figure}
  \centering
   \includegraphics[width=0.6\textwidth]{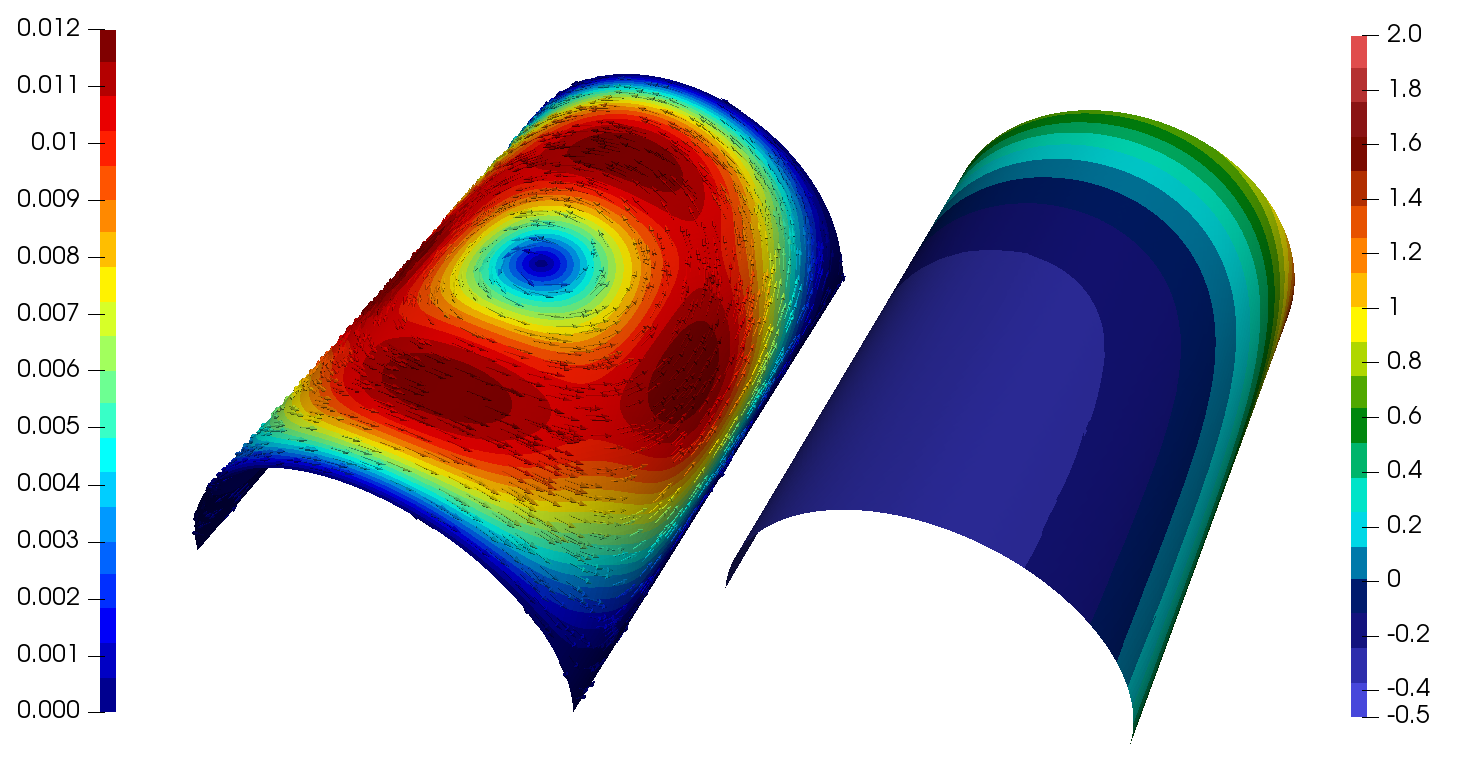}
  \caption{Absolute value $|u|$ (left) of the exact velocity and the exact pressure $p$ (right) given by \eqref{stokesexact}.} \label{fig:stokesex}
\end{figure}%
In the following we compare two different discretizations. The first one is the $H(\divergence_{\Gamma})$-conforming HDG discretization, see Section \ref{sec::HDG} and Section \ref{sec::discstokes}. The second one is based on the HDG formulation for the Vector Laplacian of section \ref{sec::HdivDG}. For the divergence constraint we now further introduce the bilinear form
\begin{align*}
  b_h^{\operatorname{HDG}}((u_h, \lambda_T), q_h) := -\sum_T \int_T \divergence_\Gamma(u_h) q_h \dx + \int_{\partial T} (u - \lambda_T)\cdot \nG~ q_h \ds \quad \forall u_h,\lambda_T \in W_h^{k_u} \times (\Lambda_h^{k_u} \times \Lambda_h^{k_u}), q_h \in Q^{{k_u}-1}_h.
\end{align*}
For a fixed velocity approximation order $k_u$ the HDG Stokes discretization then read as:  Find $u_h,\lambda_T \in W_h^{{k_u}} \times (\Lambda_h^{k_u} \times \Lambda_h^{k_u}), p_h \in Q^{{k_u}-1}_h$ such that
\begin{align} \label{eq::HDGstokes}
  \begin{array}{rll}
  2 \nu a_h^{\operatorname{HDG}}((u_h,\lambda_T) , (v_h,\theta_T)) + b_h^{\operatorname{HDG}}((v_h, \theta_T), p_h)  &=  f_h(v_h) &\quad \forall (v_h,\theta_T) \in W^{{k_u}}_h \times (\Lambda_h^{k_u} \times \Lambda_h^{k_u})  \\
  b_h^{\operatorname{HDG}}((u_h, \lambda_T), q_h) &= 0 &\quad \forall q_h \in Q_h^{{k_u}-1}.
  \end{array}
\end{align}

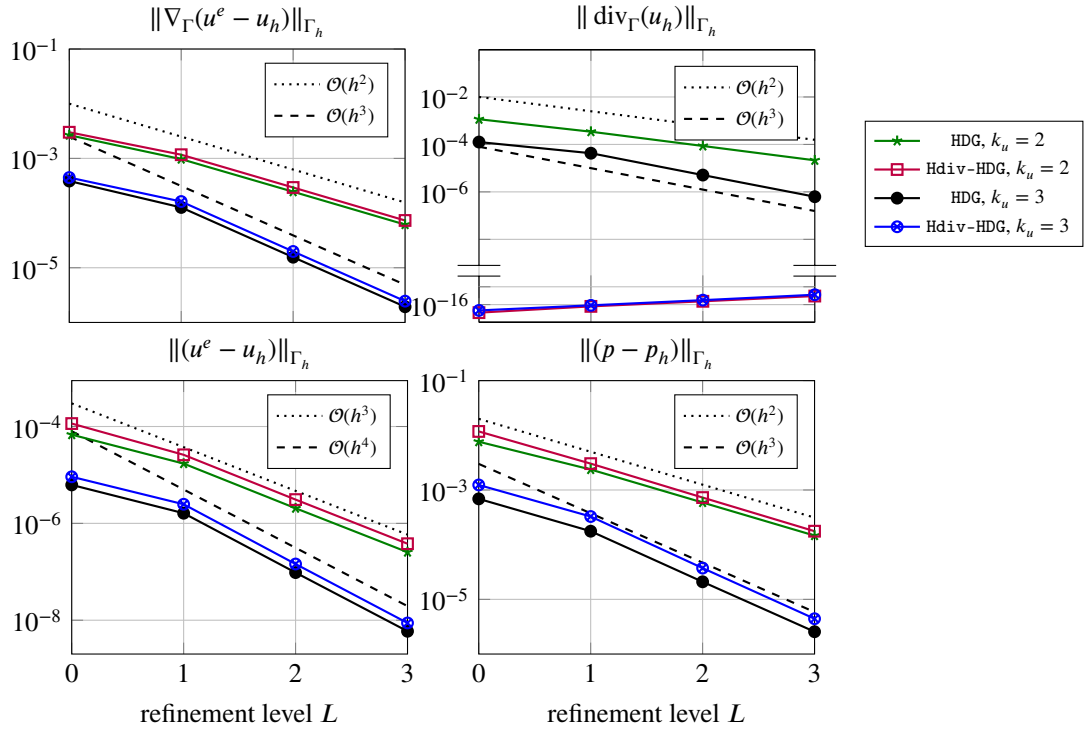
\begin{figure}
\begin{center}
\hspace*{-0.8cm}
\begin{tikzpicture}
  \begin{semilogyaxis}[
    grid=both,
    axis x line=box,
    axis line style={-},
    width=6cm,
    xmin=0, xmax=3,
    xticklabels={},
    ymin=1e-6,
    ymax=1e-1,
    title={$\Vert \nabla_\Gamma (u^e - u_h) \Vert_{\Gamma_h}$},
    title style={at={(axis description cs:0.5,0.95)},anchor=south},
    legend style={at={(0.95,0.95)}, fill=white, draw=black},
    ]

    \addplot[dotted,thick] table[x index =0, y expr={0.01*2^(-2*\thisrowno{0})}]{data/4_2_Stokes/comparison_H1err_u_P2_nu1.dat};
    \addplot[dashed,thick] table[x index =0, y expr={0.0025*2^(-3*\thisrowno{0})}]{data/4_2_Stokes/comparison_H1err_u_P2_nu1.dat}; 
    \addplot[draw=green!50!black, thick, mark=star] table[x index=0, y index=1]{data/4_2_Stokes/comparison_H1err_u_P2_nu1.dat}; 
    \addplot[draw=purple,thick, mark=square] table[x index=0, y index=2]{data/4_2_Stokes/comparison_H1err_u_P2_nu1.dat};

    \addplot[draw=black, thick, mark=*] table[x index=0, y index=1]{data/4_2_Stokes/comparison_H1err_u_P3_nu1.dat}; 
    \addplot[draw=blue,thick, mark=otimes] table[x index=0, y index=2]{data/4_2_Stokes/comparison_H1err_u_P3_nu1.dat};     
    \legend{
      \footnotesize $\mathcal{O}(h^2)$ \\
      \footnotesize $\mathcal{O}(h^3)$ \\
    }
\end{semilogyaxis}
\end{tikzpicture}
\hspace*{-0.35cm}
\begin{tikzpicture}    
\begin{groupplot}[
    group style={
        group name=my fancy plots,
        group size=1 by 2,
        xticklabels at=edge bottom,
        vertical sep=0pt
    },
    width=6cm,
    xmin=0, xmax=3,
]
\nextgroupplot[ymode=log,  ymin = 1e-10, ymax=1, 
      grid=both,
      axis x line=box,
      separate axis lines,
      every outer x axis line/.append style={draw=none},
      every outer y axis line/.append style={-},
      axis y discontinuity=parallel, 
      ytick={1e-8,1e-6,1e-4,1e-2,1e0},
      yticklabels={,$10^{-6}$,$10^{-4}$,$10^{-2}$,},
      xtick={},      
      height=4.7cm,
      legend style={at={(1.8,0.7)}},
      title={$\Vert \divergence_{\Gamma}( u_h) \Vert_{\Gamma_h}$},
      title style={at={(axis description cs:0.5,0.95)},anchor=south},
      ]
     
    \addplot[draw=green!50!black, thick, mark=star] table[x index=0, y index=1]{data/4_2_Stokes/comparison_L2err_divu_P2_nu1.dat}; 
    \addplot[draw=purple,thick, mark=square] table[x index=0, y index=2]{data/4_2_Stokes/comparison_L2err_divu_P2_nu1.dat};   
    \addplot[draw=black, thick, mark=*] table[x index=0, y index=1]{data/4_2_Stokes/comparison_L2err_divu_P3_nu1.dat}; 
    \addplot[draw=blue,thick, mark=otimes] table[x index=0, y index=2]{data/4_2_Stokes/comparison_L2err_divu_P3_nu1.dat};

    \addplot[dotted,thick] table[x index =0, y expr={0.01*2^(-2*\thisrowno{0})}]{data/4_2_Stokes/comparison_H1err_u_P2_nu1.dat};
    \addplot[dashed,thick] table[x index =0, y expr={0.00008*2^(-3*\thisrowno{0})}]{data/4_2_Stokes/comparison_H1err_u_P2_nu1.dat}; 
      \legend{
        \footnotesize \texttt{HDG}, $k_u=2$  \\
        \footnotesize \texttt{Hdiv-HDG}, $k_u=2$ \\
        \footnotesize \texttt{HDG}, $k_u=3$  \\
        \footnotesize \texttt{Hdiv-HDG}, $k_u=3$ \\
      }
      \draw[thick] (axis description cs:0,1) -- (axis description cs:1,1); 

\nextgroupplot[ymode=log, ymin=1e-17,ymax=1e-15, 
      grid=both,
      axis x line=box,
      separate axis lines,
      every outer x axis line/.append style={-,draw=none},
      every outer y axis line/.append style={-},
      ytick={1e-19,1e-17,1e-16,1e-15},
      yticklabels={,,$10^{-16}$,},
      xticklabels={},
      height=2.05cm,
      legend style={at={(0.95,7.2)}, fill=white, draw=black},
      ]
      \addplot[dotted,thick] table[x index =0, y expr={0.01*2^(-2*\thisrowno{0})}]{data/4_2_Stokes/comparison_L2err_divu_P2_nu1.dat};
    \addplot[dashed,thick] table[x index =0, y expr={0.0025*2^(-3*\thisrowno{0})}]{data/4_2_Stokes/comparison_L2err_divu_P2_nu1.dat}; 
    \addplot[draw=green!50!black, thick, mark=star] table[x index=0, y index=1]{data/4_2_Stokes/comparison_L2err_divu_P2_nu1.dat}; 
    \addplot[draw=purple,thick, mark=square] table[x index=0, y index=2]{data/4_2_Stokes/comparison_L2err_divu_P2_nu1.dat};

    \addplot[draw=black, thick, mark=*] table[x index=0, y index=1]{data/4_2_Stokes/comparison_L2err_divu_P3_nu1.dat}; 
    \addplot[draw=blue,thick, mark=otimes] table[x index=0, y index=2]{data/4_2_Stokes/comparison_L2err_divu_P3_nu1.dat};
    
      \legend{
        \footnotesize $\mathcal{O}(h^2)$ \\
         \footnotesize $\mathcal{O}(h^3)$ \\
      }
      \draw[thick] (axis description cs:0,0) -- (axis description cs:1,0); 

\end{groupplot}
\end{tikzpicture}
\end{center}
\vspace*{-0.85cm}
\begin{center}
\hspace*{-4.1cm}
\begin{tikzpicture}
  \begin{semilogyaxis}[
    grid=both,
    axis x line=box,
    axis line style={-},
    width=6cm,
    xmin=0, xmax=3,
    xlabel=refinement level $L$,
    xtick={0,1,2,3,4},
    title={$\Vert (u^e - u_h) \Vert_{\Gamma_h}$},
    title style={at={(axis description cs:0.5,0.95)},anchor=south},
    legend style={at={(0.95,0.95)}, fill=white, draw=black},
    ]

    \addplot[dotted,thick] table[x index =0, y expr={0.0003*2^(-3*\thisrowno{0})}]{data/4_2_Stokes/comparison_L2err_u_P2_nu1.dat};
    \addplot[dashed,thick] table[x index =0, y expr={0.00008*2^(-4*\thisrowno{0})}]{data/4_2_Stokes/comparison_L2err_u_P2_nu1.dat}; 
    \addplot[draw=green!50!black, thick, mark=star] table[x index=0, y index=1]{data/4_2_Stokes/comparison_L2err_u_P2_nu1.dat}; 
    \addplot[draw=purple,thick, mark=square] table[x index=0, y index=2]{data/4_2_Stokes/comparison_L2err_u_P2_nu1.dat};

    \addplot[draw=black, thick, mark=*] table[x index=0, y index=1]{data/4_2_Stokes/comparison_L2err_u_P3_nu1.dat}; 
    \addplot[draw=blue,thick, mark=otimes] table[x index=0, y index=2]{data/4_2_Stokes/comparison_L2err_u_P3_nu1.dat};     
    \legend{
      \footnotesize $\mathcal{O}(h^3)$ \\
      \footnotesize $\mathcal{O}(h^4)$ \\
    }
\end{semilogyaxis}
\end{tikzpicture}
\hspace*{-0.35cm}
\begin{tikzpicture}    
\begin{semilogyaxis}[
    grid=both,
    axis x line=box,
    axis line style={-},
    width=6cm,
    xmin=0, xmax=3,
    ymin=1e-6,
    ymax=1e-1,
    xlabel=refinement level $L$,
    xtick={0,1,2,3,4},   
    title={$\Vert (p - p_h) \Vert_{\Gamma_h}$},
    title style={at={(axis description cs:0.5,0.95)},anchor=south},
    legend style={at={(0.95,0.95)}, fill=white, draw=black},
    ]

    \addplot[dotted,thick] table[x index =0, y expr={0.02*2^(-2*\thisrowno{0})}]{data/4_2_Stokes/comparison_L2err_p_P2_nu1.dat};
    \addplot[dashed,thick] table[x index =0, y expr={0.003*2^(-3*\thisrowno{0})}]{data/4_2_Stokes/comparison_L2err_p_P2_nu1.dat}; 
    \addplot[draw=green!50!black, thick, mark=star] table[x index=0, y index=1]{data/4_2_Stokes/comparison_L2err_p_P2_nu1.dat}; 
    \addplot[draw=purple,thick, mark=square] table[x index=0, y index=2]{data/4_2_Stokes/comparison_L2err_p_P2_nu1.dat};

    \addplot[draw=black, thick, mark=*] table[x index=0, y index=1]{data/4_2_Stokes/comparison_L2err_p_P3_nu1.dat}; 
    \addplot[draw=blue,thick, mark=otimes] table[x index=0, y index=2]{data/4_2_Stokes/comparison_L2err_p_P3_nu1.dat};     
    \legend{
      \footnotesize $\mathcal{O}(h^2)$ \\
      \footnotesize $\mathcal{O}(h^3)$ \\
    }
\end{semilogyaxis}
\end{tikzpicture}
\end{center}
\vspace*{-0.25cm}
\caption{Error behavior for the Stokes problem of section \ref{sec:stokesnumex} for the $H(\divergence_{\Gamma})$-conforming HDG method and the HDG method given by \eqref{eq::HDGstokes} on 4 successively refined meshes (uniform refinements) for a fixed viscosity $\nu=1$, two polynomial orders $k_u = 2,3$ and an geometry approximation order $k_g = k_u + 1$.} \label{fig::errorstokes}
\vspace*{-0.25cm}
\end{figure}
In Figure \ref{fig::errorstokes} the error behavior for the above Stokes problem (with right-hand side \eqref{stokerhssexact}) is given for both discussed discretizations with a fixed viscosity $\nu = 1$ and different polynomial orders $k_u = 2,3$. As before we use the labels \texttt{HDG} and \texttt{Hdiv-HDG}. For $\alpha$ in the SIP stabilization we take $10$. Note, that in contrast to the previous section we used the geometry approximation order $k_g = k_u + 2$. A numerical investigation showed that we have to apply this enhanced geometry approximation (only needed for this particular example) due to the ill-conditioned inverse of the sine function at $1$ and $-1$ (needed for the calculation of the right-hand side). The first, third and fourth plot show the $H^1$-semi-norms error and the $L^2$-norm error of the velocity and the $L^2$-norm error of the pressure for both methods where we used an initial triangulation with 100 elements and 3 refinement levels. As we can see, all errors convergence with optimal order and the accuracy of both methods is approximately the same.
The second plot shows the $L^2$-norm of the surface divergence. As proven by Lemma~\ref{lemma:exdivfree} the solution of the $H(\divergence_{\Gamma})$-conforming method is exactly divergence-free, whereas the standard HDG method is only weakly divergence-free. However the divergence error of the HDG method stills shows the optimal convergence order.

To emphasize the importance of exactly divergence-free velocity solutions we now focus on pressure robustness. Beside the observations discussed in Section \ref{sec::HdivBenefits} we want to discuss pressure robustness with respect to the arising error estimates. A standard a priori error estimate of inf-sup stable Stokes discretizations usually reads as
\begin{align*}
  \| u - u_h \|_{H^1,h} \le c (\inf\limits_{v_h \in V^{k_u}_h} \| u - v_h \|_{H^1,h} + \frac{1}{\nu} \inf\limits_{q_h \in Q^{k_u-1}_h} \| p - q_h \|_{L^2}),
\end{align*}
where $\| \cdot \|_{H^1,h}$ is an appropriate (H)DG-version of an $H^1$-norm, and $c$ is a constant independent on the mesh size $h$ and the viscosity $\nu$. Above estimate shows that the velocity error may depend on the best approximation of the pressure including the factor $1/\nu$, hence the velocity error can blow up in the case of vanishing viscosity $\nu \ll 1$. As discussed in the literature, see for example \cite{LLS_SIAM_2017,GLS_MCS_2019}, methods that provide exactly divergence-free velocities allow to derive an error estimate that reads as
\begin{align*}
  \| u - u_h \|_{H^1,h} \le c \inf\limits_{v_h \in V^{k_u}_h} \| u - v_h \|_{H^1,h} + F(u),
\end{align*}
where $F(u)$ is a function that only depends on the exact solution $u$ (and not $\nu$) and shows optimal convergence properties (for methods yielding exactly divergence-free solutions we have $F(u) = 0$).
Hence, these methods show no bad behavior for small values of the viscosity $\nu$.
Note that above estimate assumes an exact geometry representation $\Gamma_h = \Gamma$, compare the discussion below.

\begin{figure}
  \begin{center}
    \begin{tikzpicture}
  \begin{axis}[
    name=plot1,
    scale=1,
    width = 6cm,
    title = {$|| \nabla_\Gamma( u^e - u_h) ||_{\Gamma_h}$ for $k_g = k_u +2$},
    legend style={text height=0.7em },
    legend style={draw=none},
    style={column sep=0.1cm},
    xlabel=$\nu$,
    x label style={at={(0.95,0.05)},anchor=west},
    xmode=log,
    ymode=log,
    ytick = {1e-6,1e-4,1e-2,1e-0,1e2,1e4},
    y tick label style={
      /pgf/number format/.cd,
      fixed,
      precision=2
    },
    x tick label style={
      /pgf/number format/.cd,
      fixed,
      precision=2
    },
    grid=both,
    legend style={
      cells={align=left},
      at={(1.05,0.8)},
      anchor = north west
    },
    ]

    \addplot[draw=green!50!black, thick, mark=star] table[x index=0, y index=1]{data/4_2_Stokes/nu_comparison_H1err_u_P2_1.dat}; 
    \addplot[draw=purple,thick, mark=square] table[x index=0, y index=2]{data/4_2_Stokes/nu_comparison_H1err_u_P2_1.dat};
    \addplot[draw=black, thick, mark=*] table[x index=0, y index=1]{data/4_2_Stokes/nu_comparison_H1err_u_P3_1.dat};
    \addplot[draw=blue,thick, mark=otimes] table[x index=0, y index=2]{data/4_2_Stokes/nu_comparison_H1err_u_P3_1.dat};
  \end{axis}  
\end{tikzpicture}
\hspace{-0.35cm}
\begin{tikzpicture}
  \begin{axis}[
    name=plot1,
    scale=1,
    width = 6cm,
    title = {$|| \nabla_\Gamma( u^e - u_h) ||_{\Gamma_h}$ for $k_g = 2 k_u +2$},
    legend style={text height=0.7em },
    legend style={draw=none},
    style={column sep=0.1cm},
    xlabel=$\nu$,
    x label style={at={(0.95,0.05)},anchor=west},
    xmode=log,
    ymode=log,
    ytick = {1e-6,1e-4,1e-2,1e-0,1e2,1e4},
    y tick label style={
      /pgf/number format/.cd,
      fixed,
      precision=2
    },
    x tick label style={
      /pgf/number format/.cd,
      fixed,
      precision=2
    },
    grid=both,
    legend style={
      cells={align=left},
      at={(1.05,0.8)},
      anchor = north west
    },
    ]

    \addplot[draw=green!50!black, thick, mark=star] table[x index=0, y index=1]{data/4_2_Stokes/nu_comparison_H1err_u_P2_2.dat}; 
    \addplot[draw=purple,thick, mark=square] table[x index=0, y index=2]{data/4_2_Stokes/nu_comparison_H1err_u_P2_2.dat};
    \addplot[draw=black, thick, mark=*] table[x index=0, y index=1]{data/4_2_Stokes/nu_comparison_H1err_u_P3_2.dat};
    \addplot[draw=blue,thick, mark=otimes] table[x index=0, y index=2]{data/4_2_Stokes/nu_comparison_H1err_u_P3_2.dat};   
  \end{axis}  
\end{tikzpicture}
\hspace{-0.35cm}
    \begin{tikzpicture}
  \begin{axis}[
    name=plot1,
    scale=1,
    width = 6cm,
    title ={ $|| \nabla_\Gamma( u^e - u_h) ||_{\Gamma_h}$ for $k_g = 3 k_u +2$},
    legend style={text height=0.7em },
    legend style={draw=none},
    style={column sep=0.1cm},
    xlabel=$\nu$,
    x label style={at={(0.95,0.05)},anchor=west},
    xmode=log,
    ymode=log,
    ytick = {1e-6,1e-4,1e-2,1e-0,1e2,1e4},
    y tick label style={
      /pgf/number format/.cd,
      fixed,
      precision=2
    },
    x tick label style={
      /pgf/number format/.cd,
      fixed,
      precision=2
    },
    grid=both,
    legend style={
      cells={align=left},
      at={(0.65,0.8)},
      anchor = north west
    },
    ]

    \addplot[draw=green!50!black, thick, mark=star] table[x index=0, y index=1]{data/4_2_Stokes/nu_comparison_H1err_u_P2_3.dat}; 
    \addplot[draw=purple,thick, mark=square] table[x index=0, y index=2]{data/4_2_Stokes/nu_comparison_H1err_u_P2_3.dat};
    \addplot[draw=black, thick, mark=*] table[x index=0, y index=1]{data/4_2_Stokes/nu_comparison_H1err_u_P3_3.dat};
    \addplot[draw=blue,thick, mark=otimes] table[x index=0, y index=2]{data/4_2_Stokes/nu_comparison_H1err_u_P3_3.dat};
  \end{axis}   
\end{tikzpicture}
\begin{tikzpicture}
  \begin{axis}[
    hide axis,    
    legend style={draw=none},
    legend columns=4,
    style={column sep=0.1cm},
    xmin=1,
    xmax=1,
    ymin=0,
    ymax=0.4,
    width = 2cm,
    ]

    \addlegendimage{green!50!black, thick, mark=star};
    \addlegendentry{\footnotesize \texttt{HDG}, $k_u=2$};
    \addlegendimage{purple,thick, mark=square};
    \addlegendentry{\footnotesize \texttt{Hdiv-HDG}, $k_u=2$};    
    \addlegendimage{black, thick, mark=*};
    \addlegendentry{\footnotesize \texttt{HDG}, $k_u=3$};
    \addlegendimage{blue,thick, mark=otimes};
    \addlegendentry{\footnotesize \texttt{Hdiv-HDG}, $k_u=3$};    
  \end{axis}   
\end{tikzpicture}
\end{center}
\vspace*{-0.5cm}
\caption{The $H^1$-seminorm error of the velocity solutions of the Stokes problem of section \ref{sec:stokesnumex} for the $H(\divergence_{\Gamma})$-conforming HDG method and the HDG method given by \eqref{eq::HDGstokes} for varying viscosities $\nu=1,\ldots,10^{-7}$ on a fixed triangulation with geometry approximation order $k_g = i~k_u +2$ for $i=1,2,3$.} \label{fig::probust}
\vspace*{-0.3cm}
\end{figure}
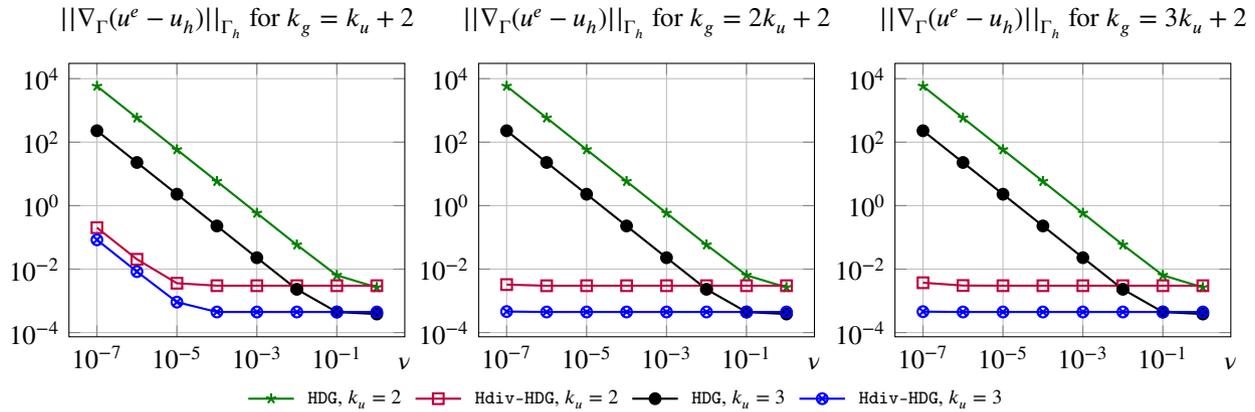
In Figure \ref{fig::probust} the $H^1$-seminorm error of the solution of the same problem as above for varying viscosities $\nu = 10^{-6},\ldots,1$, two different polynomial orders $k_u=2,3$ and different geometry approximation orders $k_g = i~k_u +2$ for $i=1,2,3$ is given.
We can make several observations. First, the errors of the weakly divergence-free HDG method, see equation \eqref{eq::HDGstokes},  show the expected behavior: Reducing the viscosity $\nu$ leads to a blow up of the $H^1$-semi norm error of the velocity independently of the geometry approximation. We can clearly see the scaling $1/\nu$ for all combination of approximation orders $k_u$ and geometry approximations $k_g$. Next, note that the $H(\divergence_{\Gamma})$-conforming method is expected to be pressure robust as the discrete velocity is exactly divergence-free. As we can see, this is indeed true if the geometry approximation is accurate enough. In the case $k_g = 3k_u + 2$ the error is constant and independent of the choice of $\nu$. Reducing the approximation order $k_g$ shows that the error is again affected by a change of the viscosity, however the error is still much better compared to the standard HDG method. This behavior is also known from the flat case, and the problem comes from an inexact evaluation of the right-hand side integral
   \begin{align*}
 \int_{\Gamma_h} f \cdot v_h \dx.
   \end{align*}
   As all integrals are computed using a quadrature rule the right-hand side integral might not vanish in the case where $f = \nabla_{\Gamma} \theta$ and $v_h$ is  divergence-free, i.e. $\divergence_{\Gamma} (v_h) = 0$, compare Section \ref{sec::HdivBenefits}. Beside choosing a high order quadrature rule, we now also have to increase the geometry approximation to ensure that a gradient field $f = \nabla_{\Gamma} \theta$ is also a gradient field on $\Gamma_h$, otherwise again the integral might not vanish even if a high order quadrature rule is used.


\subsection{Generalization of the Sch\"afer Turek benchmark for surfaces} \label{sec:schaeferturek}
We consider the setup of a standard benchmark test case, cf. Ref.\cite[Case 2D-2]{schafer1996benchmark}. First we recall the setup in the plane and afterwards apply different mappings to obtain geometries in 3D, some of them as in Ref. \cite{MR3846120}.
\paragraph{The flat setup (reference configuration)}
The domain is a rectangular channel without an almost vertically centered circular obstacle, cf. Figure \ref{fig:sketchst},
\begin{equation*}
  \hat \Gamma := [0,2.2] \! \times \! [0,0.41] \setminus B_{0.05}((0.2,0.2)).
\end{equation*}
We denote the velocity and pressure solution as $\hat{u}$ and $\hat{p}$. The boundary is decomposed into $\hat \gamma_{\text{in}} := \{\hat x = 0\}$, the inflow boundary, $\hat \gamma_{\text{out}} := \{\hat x = 2.2\}$, the outflow boundary and $\hat \gamma_W := \partial \Gamma \setminus ( \hat \gamma_{\text{in}} \cup \hat \gamma_{\text{out}} )$. On $\hat \gamma_{\text{out}}$ we prescribe natural boundary conditions which read as $ ( - 2 \nu \eps_\Gamma(\hat u) + \hat p P) \cdot \hat \nG = 0 $;
on $\hat \gamma_W$ homogeneous Dirichlet boundary conditions for the velocity and on $\hat \gamma_{\text{in}}$ the inflow Dirichlet boundary conditions
\begin{equation*}
 \hat u(0,y,t) = \hat u_D = \bar{u} \cdot 6 \cdot {y (0.41 - y)}/{0.41^2} \cdot (1,0,0).
\end{equation*}
Here, $\bar{u} = 1$ is the average inflow velocity and the viscosity is fixed to $\nu = 10^{-3}$ which results in a Reynolds number $Re=100$. This setup results in a time-periodic flow with vortex shedding behind the obstacle, cf. Figure \ref{fig:sketchst}.

The (time dependent) quantities of interest in this example are the forces that act on the disc $\hat \gamma_\circ = S_{0.05}((0.2,0.2))$ and the pressure difference before and behind the obstacle:
\begin{equation*}
  F_\circ := \int_{\hat \gamma_\circ} \sigma(u,p) \cdot \hat \nG  \dshat, \quad 
  \Delta p  := p_{\text{front}} - p_{\text{back}},\quad  p_{\text{front}} = p(0.15,0.5), \quad p_{\text{back}} = p(0.25,0.5).
\end{equation*}
This benchmark problem is well studied in the literature and reference values are available in Ref.\cite{schafer1996benchmark}.
Below, we consider similar setups on geometries $\Gamma_i = \Phi_i(\hat \Gamma)$ for mappings $\Phi_i,~i=1,..,4$ specified below. For $b \in \{\text{in},\text{out}, W, \circ \}$ we define the boundary segments accordingly, i.e. $\gamma_b := \Phi_i(\hat \gamma_b)$ and prescribe natural outflow conditions on $\gamma_{\text{out}}$, homogeneuous Dirichlet conditions on $\gamma_W$ and prescribe inflow velocities on $\gamma_{\text{in}}$. For the inflow velocities we take (for sake of comparibility) the choice of Ref.\cite{MR3846120}: $u|_{\gamma_{\text{in}}} = \Phi_i' ~ \hat u \circ \Phi_i^{-1}$. Let us stress that this results in a tangential velocity, but may result in an average inflow velocity that may deviate from $\bar{u}=1$. We note that the degrees of freedoms of our $H(\divergence_{\Gamma})$-conforming HDG formulation naturally fit the tangential boundary conditions. The inflow which is in co-normal direction $\nG$ corresponds to the degrees of freedom of the $V_h^k$ whereas tangential flow conditions (here: zero) correspond to the boundary degrees of freedom of $\Lambda_h^{k_u}$.

\begin{figure}
  \begin{center}
    \hfill
    \begin{tabular}{cc@{\qquad\qquad}c@{}c@{}c@{}c@{}c@{}c@{\qquad\qquad}c@{}c@{}c@{}c@{}c@{}c}
\multicolumn{2}{c@{\qquad\qquad}}{meshes} & \multicolumn{5}{c}{Stokes} & & \multicolumn{5}{c}{Navier--Stokes ($Re=100$)} & \\
\includegraphics[angle=90,height=0.25\textheight]{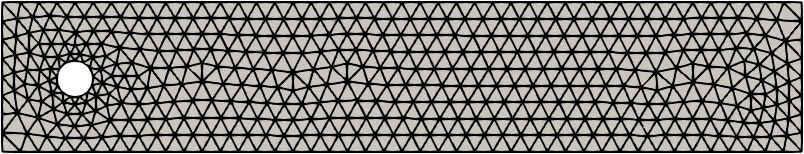} &
\includegraphics[angle=90,height=0.25\textheight]{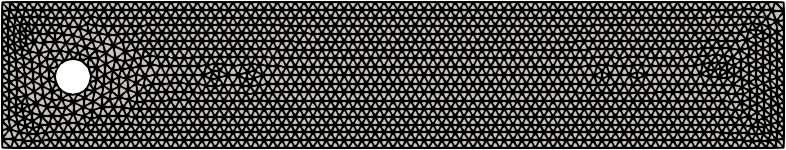} &
\includegraphics[angle=90,height=0.25\textheight]{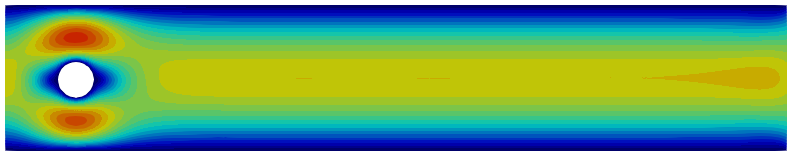} &
\includegraphics[height=0.1\textheight]{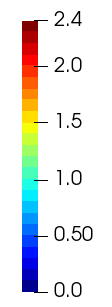} &
\includegraphics[angle=90,height=0.25\textheight]{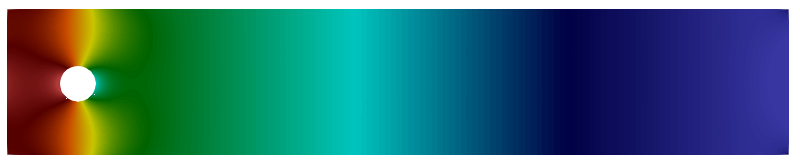} &%
\includegraphics[height=0.1\textheight]{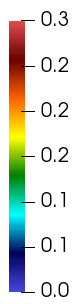} &
\includegraphics[angle=90,height=0.25\textheight]{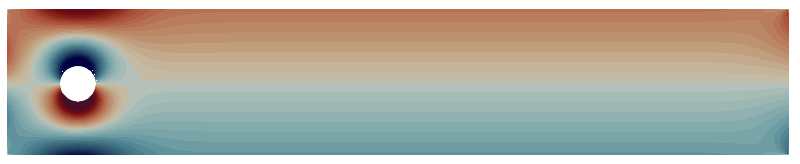} &%
\includegraphics[height=0.1\textheight]{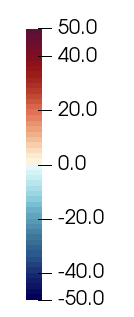}
&%
\includegraphics[angle=90,height=0.25\textheight]{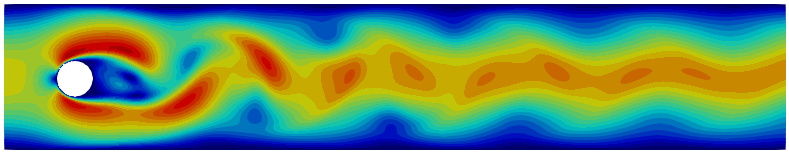} &
\includegraphics[height=0.1\textheight]{graphics/schaefer_turek_roled_velocity_stokes_scale.png}
&
\includegraphics[angle=90,height=0.25\textheight]{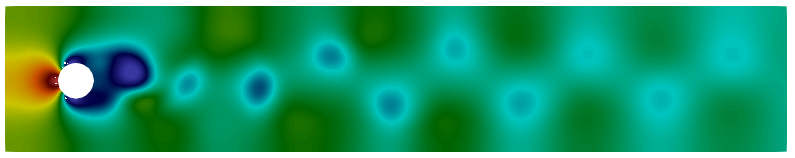} &%
\includegraphics[height=0.1\textheight]{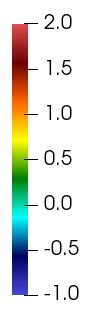}
&
\includegraphics[angle=90,height=0.25\textheight]{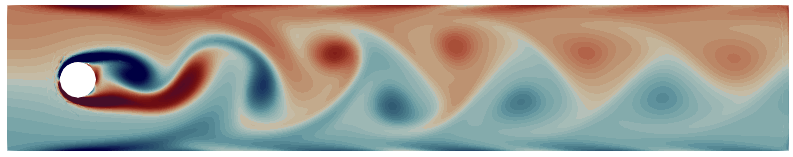} &%
\includegraphics[height=0.1\textheight]{graphics/schaefer_turek_roled_vorticity_stokes_scale.png}%
      \\[-0.1ex]
      \footnotesize{$L=1$}&
      \footnotesize{$L=2$}&
      \footnotesize{$\vert u_h\vert$} & &
      \footnotesize{$p_h$} & &
      \footnotesize{$\omega_h$} & &
      \footnotesize{$\vert u_h\vert$ } & &
      \footnotesize{$p_h$ } & &
      \footnotesize{$\omega_h$} &\\
    \end{tabular}
  \end{center}
  \vspace*{-0.4cm}
  \caption{Flat Sch\"afer-Turek benchmark problem 2D-2. The first two meshes used in the simulations (left), discrete solution of the Stokes problem (center) and discrete solution of the unsteady Navier--Stokes problem at a fixed time (right).}
\label{fig:sketchst}
  \vspace*{-0.4cm}
\end{figure}

\paragraph{Computational setup}
We solve the surface Navier--Stokes problem for four different mappings, but use -- for the most part -- the same computational setup that we want to describe here first.
As initial conditions for the Navier--Stokes problem, we take the solution of a Stokes problem with the same boundary data. To make sure that the simulation reached the time where the periodicity is established we simulate the problem for 30 time units. For the time stepping we use a second order implicit-explicit (IMEX) time stepping method and consider up to three different time step sizes: $\Delta t = 2^{1-L_t} \cdot 10^{-3}, L_t \in \{1,2,3\}$.
We consider up to three mesh levels with a characteristic mesh size (w.r.t. the reference configuration) $\hat h = 2^{1-L_s} \cdot 0.05, ~ L_s \in \{1,2,3\}$ corresponding to mesh levels $L_s=1,2,3$.
For the non-isometric mapping in Subsection~\ref{sec:schaeferturek:nonisometric}
we use the same unstructured meshes as in the flat case, resulting in 857 ($L_s=1$), 3179 ($L_s=2$) and 13081 ($L_s=3$) triangular elements, cf. Fig. \ref{fig:sketchst} for the first two levels. The meshes are not specifically adapted to improve the approximation of potential boundary layers. For the isometric mappings in Subsection~\ref{sec:schaeferturek:isometric}
we consider only one mesh with $865$ elements. In all examples in this section we fix $k_u=4$ and $k_g=5$.

\subsubsection{Isometric mappings} \label{sec:schaeferturek:isometric}
First, we consider two mappings $\Phi_i,~i=1,2$ that are isometric, i.e. there holds $\hat d(\hat{x},\hat{y}) = d_{\Gamma_i}(\Phi_i(\hat x),\Phi_i(\hat y)),~\hat x, \hat y\in \hat{\Gamma}$ where $\hat d$ is the two-dimensional Euclidean distance and $d_\Gamma$ is the geodesic distance on the surface $\Gamma_i=\Phi_i(\hat \Gamma)$ or equivalently $\det ( F_i^T \cdot F_i) = 1$ with $F_i = \Phi_i'$. 
Similar to the setup in Section \ref{sec::houseofcards} we consider as a first case
$$
\Gamma_1 = \Phi_1(\hat \Gamma) \quad \text{with}\quad \Phi_1(\hat x,\hat y) = \left\{
  \begin{array}{ll}
    (\hat x \cdot W,\hat y, H\cdot \hat x) & \text{ if } \hat x \leq 1.1, \\
    (\hat x \cdot W,\hat y, H\cdot (2.2-\hat x)) & \text{ if } \hat x > 1.1,
  \end{array}
\right.
$$
with $H=\frac12$, $W = \sqrt{3/4}$. To consider the kink in the geometry properly we use a slightly different mesh in this setup than in all others of this subsection by making sure that the line $\hat{x}=1.1$ corresponds to a mesh line, c.f.~Fig. \ref{fig:stlength_pres}. Thereby the mapping $\Phi_i$ can be represented exactly in a finite element space so that the mapped mesh introduces no additional geometrical error. We note that already the mesh for $\hat \Gamma$ includes (small) approximation errors due to the approximation of the circle. 

As a second example we consider 
$$ \Phi_2(\hat{x},\hat{y}) = (- \tfrac{0.41}{\pi} \cos(\pi \tfrac{\hat y}{0.41}), \tfrac{0.41}{\pi} (1-\sin(\pi \tfrac{ \hat y }{0.41}))), $$
which corresponds to a bending of the flat geometry around the $x$-axis. 
This time, when projecting into a finite element space, we can not represent $\Phi_2$ exactly, hence the isometry property will only be fulfilled approximately. In Fig. \ref{fig:stlength_pres} we show the geometries, the vorticity at a fixes time ($t=26$) and plots that compare $\Delta p(t)$ in a small range of time ($t \in [26,26.2]$ on the coarsest mesh level $L_s=1$ with time level $L_t=2$.

We observe that there is hardly any difference between the results. The differences between the results for $\hat \Gamma$ and $\Gamma_1$ are only due to round-off errors and not visible even after heavily zooming in. For $\Gamma_2$ we also only observe differences in the $6$th digit.  

\begin{figure}
  \vspace*{-0.25cm}
  \begin{center}
    \includegraphics[width=0.2\textwidth,angle=90]{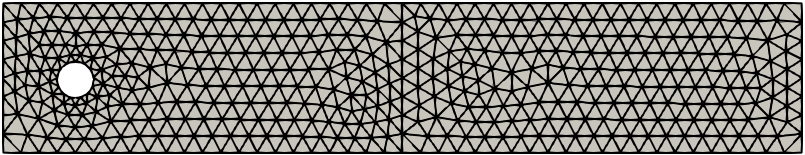}%
    \includegraphics[width=0.2\textwidth,angle=90]{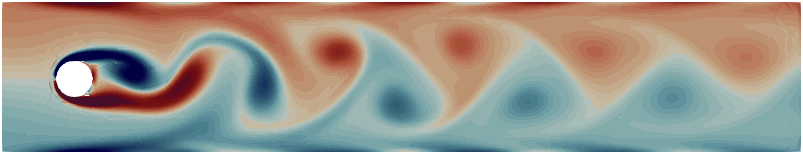}%
    \includegraphics[height=0.2\textwidth]{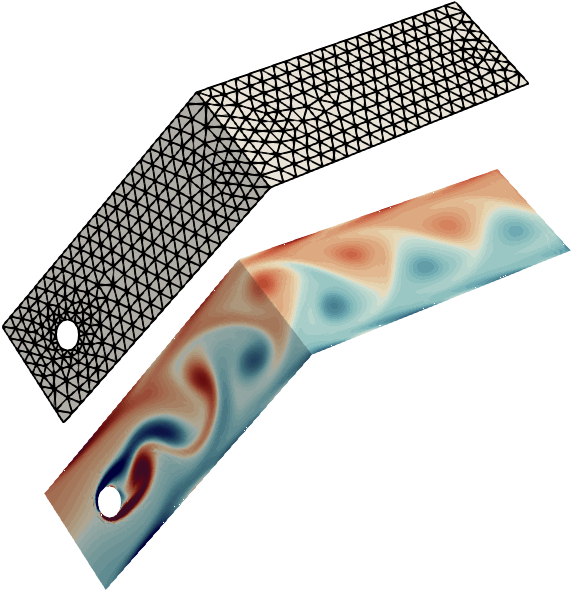}%
    \hspace*{-0.08\textwidth}
    \includegraphics[height=0.08\textwidth]{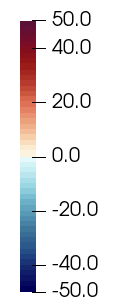}
    \hspace*{0.03\textwidth}
    \includegraphics[height=0.2\textwidth]{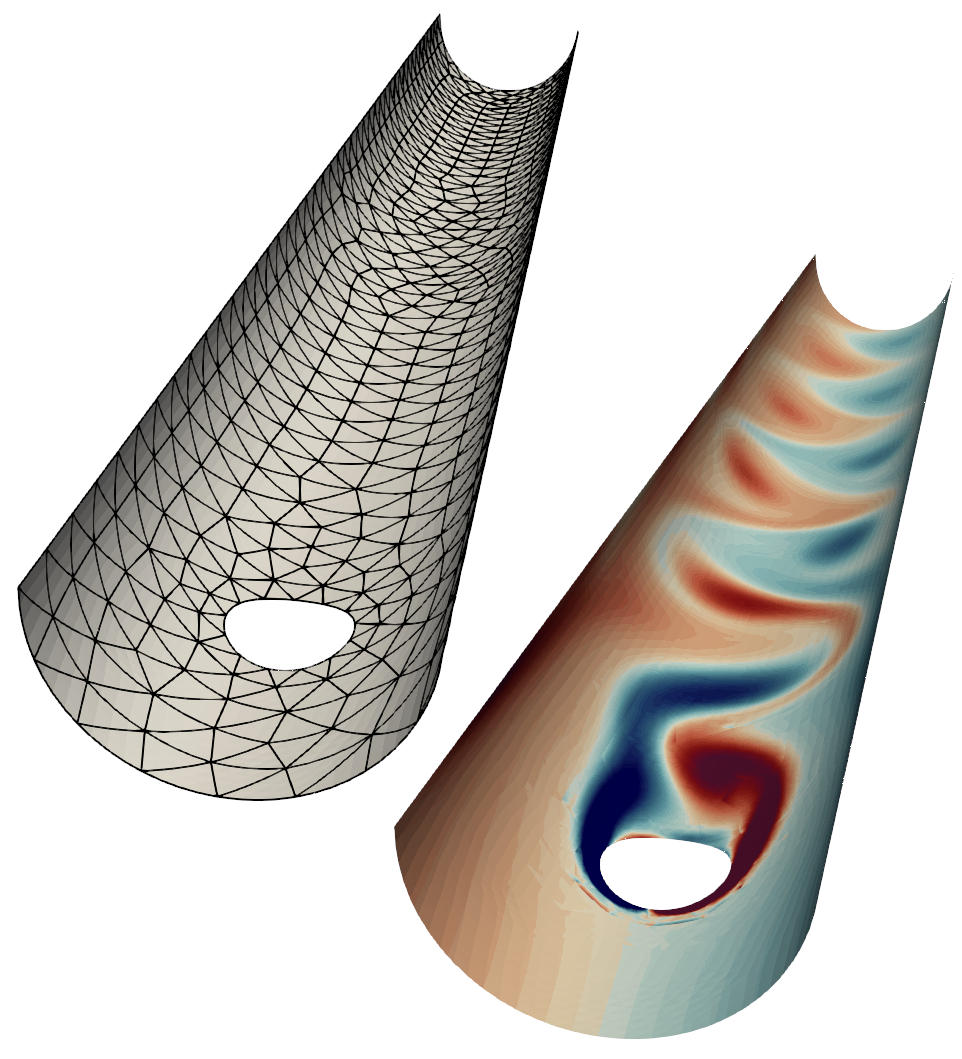}%
    \hspace*{-0.02\textwidth}
    \includegraphics[height=0.08\textwidth]{graphics/schaefer_turek_kink_vorticity_stokes_scale.png}
    \includegraphics[height=0.2\textwidth]{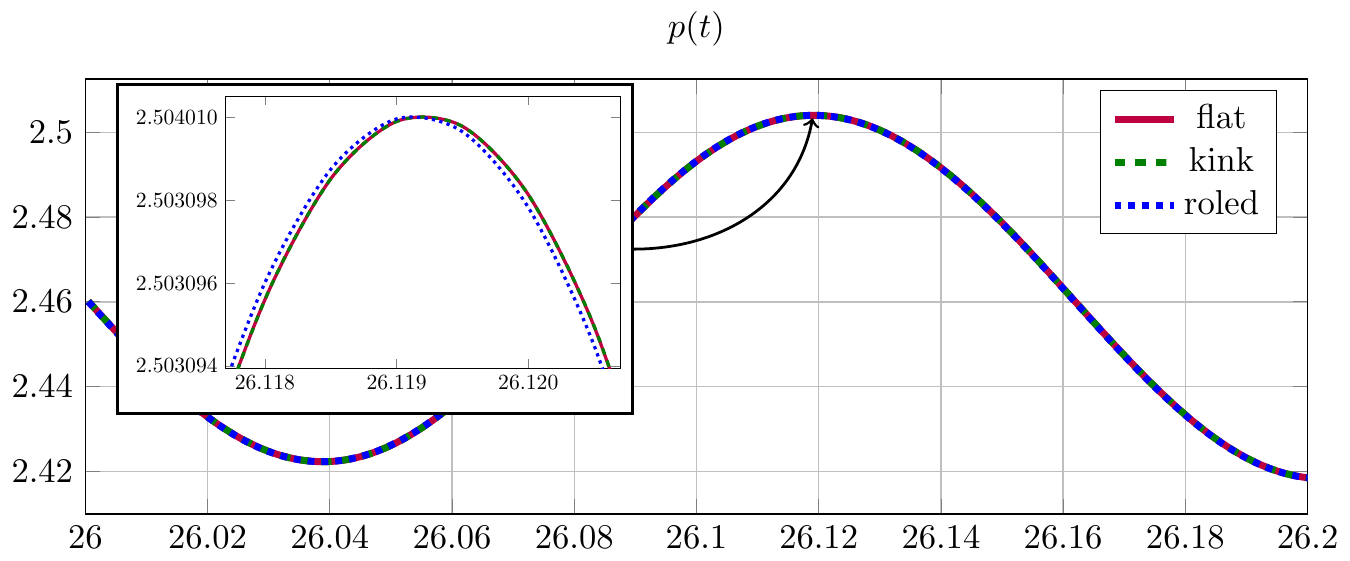}
  \end{center}
  \vspace*{-0.5cm}
  \caption{Meshes and vorticity for $\operatorname{id}$ (flat), $\Phi_1$(kink) and $\Phi_2$(roled) and pressure difference over time.}
  \label{fig:stlength_pres}
\end{figure}

\subsubsection{Non-isometric mapping} \label{sec:schaeferturek:nonisometric}
As another example we consider a non-isometric mapping, i.e. a flow that (not only discretely) deviates from the flat case. The example is one of two similar examples in Ref. \cite{MR3846120}. We note that we also tried the other case, but do not present it here due to the similarity in the results. The mapping is
\begin{align*}
  \Phi_3(\hat x, \hat y) &= \Big( \cos\left( \pi \tfrac{\hat x}{2.2}\right) \cdot f(\hat y), && \sin\left( \pi \tfrac{\hat x}{2.2}\right) \cdot f(\hat y), && 2 + 0.5 f(\hat y) - \sin(3 f(\hat y))\Big), \quad \text{ with } f(\hat y) = \hat y + 0.35. 
\end{align*}
%
%
\begin{figure}
\vspace*{-0.2cm}
  \begin{center}
    \begin{tabular}{cccc}
\includegraphics[height=0.125\textwidth]{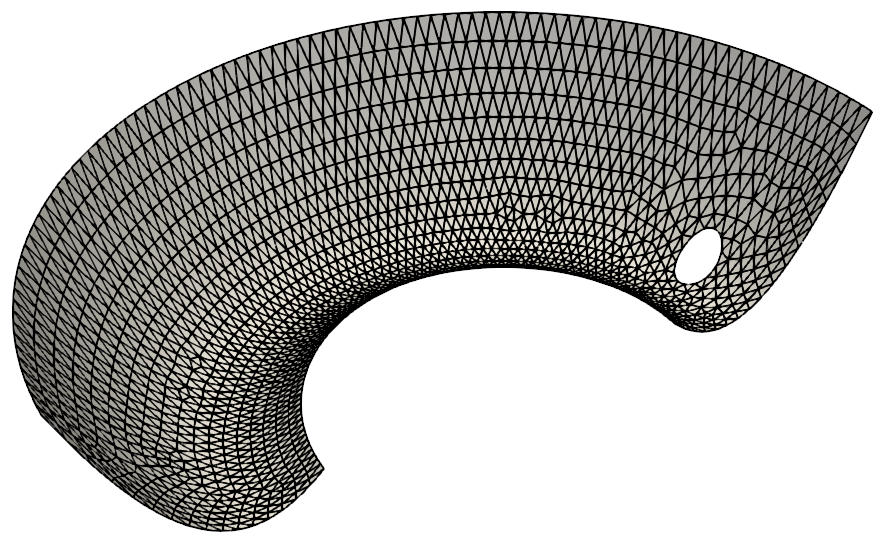} &%
\includegraphics[height=0.125\textwidth]{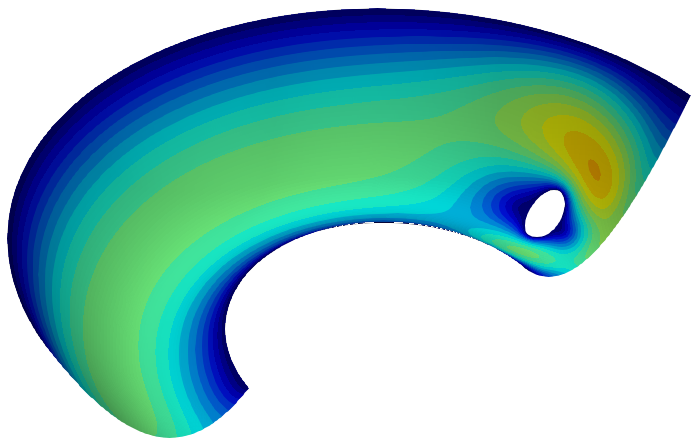}%
\includegraphics[height=0.1\textwidth]{graphics/schaefer_turek_roled_velocity_stokes_scale.png}&%
\includegraphics[height=0.125\textwidth]{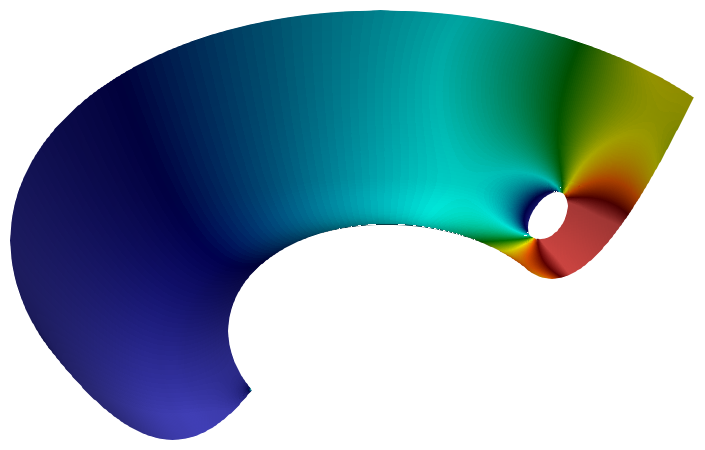}%
\includegraphics[height=0.1\textwidth]{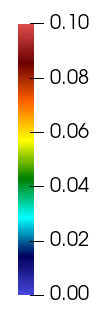}&%
\includegraphics[height=0.125\textwidth]{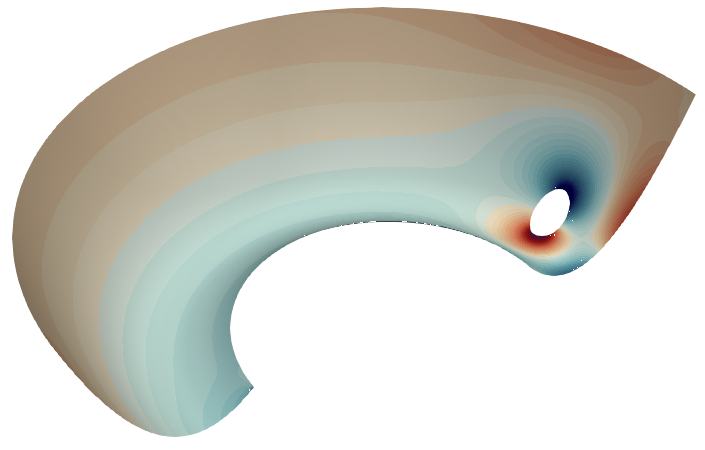}%
\includegraphics[height=0.1\textwidth]{graphics/schaefer_turek_roled_vorticity_stokes_scale.png} \\[-2ex]
&%
\includegraphics[height=0.125\textwidth]{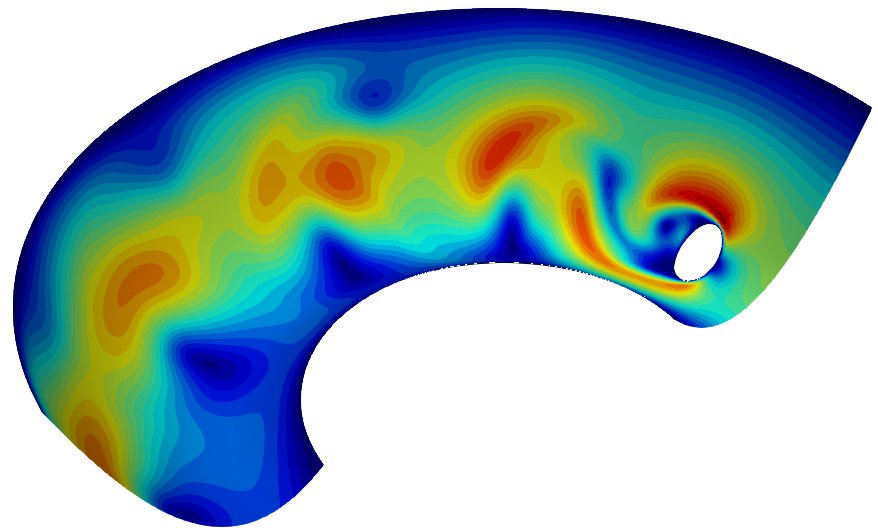}%
\includegraphics[height=0.1\textwidth]{graphics/schaefer_turek_roled_velocity_stokes_scale.png}&%
\includegraphics[height=0.125\textwidth]{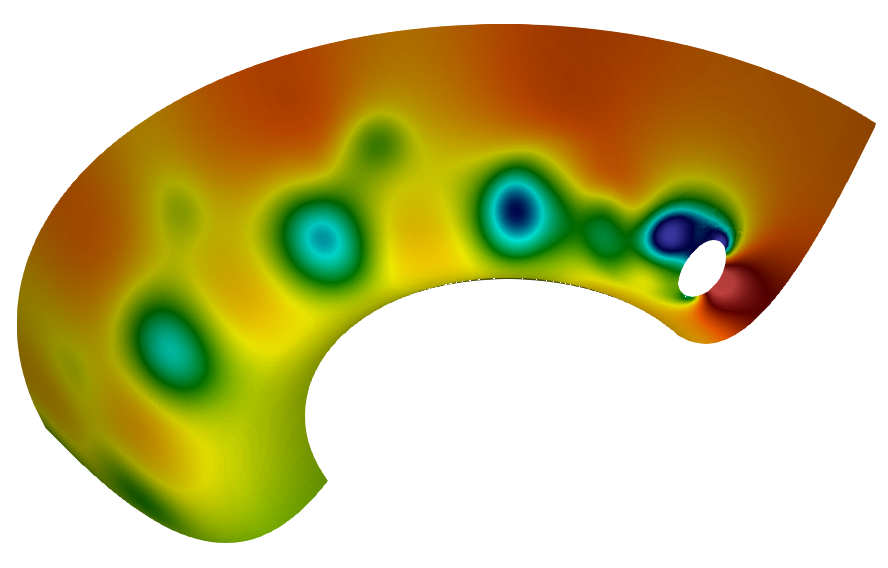}%
\includegraphics[height=0.1\textwidth]{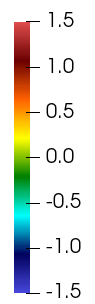}&%
\includegraphics[height=0.125\textwidth]{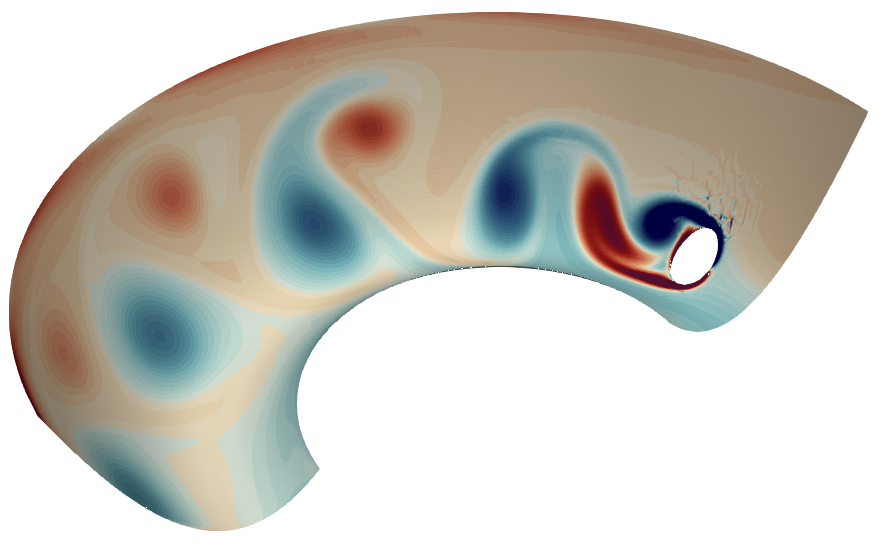}%
                                                                                                \includegraphics[height=0.1\textwidth]{graphics/schaefer_turek_roled_vorticity_stokes_scale.png} \\[-4ex]
      mesh & $\vert u_h \vert$ & $\vert p_h \vert$ & $\vert \omega_h \vert$
    \end{tabular}
  \end{center}
  \vspace*{-0.55cm}
  \caption{Geometry and surface Stokes (top row) and Navier--Stokes solution (bottom row) for mapping $\Phi_3$.}
  \label{fig:phi3}
\end{figure}

\noindent 
In Fig.~\ref{fig:phi3} we display the solution of the surface Stokes and the surface Navier--Stokes problem. Already for the Stokes problem we see a significant deviation from the flat Stokes solution, especially in the pressure. This is partially due to a different inflow profile, but mainly due to a different length, width and shape of the channel. The vortex shedding behind the obstacle shows a qualitatively similar behavior to the flat case. However, the frequency and the forces acting on the obstacle are different. In the simulation, after 10 time units the vortex shedding is close to periodic. In Fig.~\ref{fig:phi3:press} we display the pressure difference for two periods ($t^\ast=0$ corresponds to a maximum of the pressure within each simulation) for different mesh and time levels.
On the finest resolution we obtain a minimum and maximum pressure difference of $2.07982$ and $2.69418$ which is in very good agreement with the findings in Ref.\cite{MR3846120}. Also the period length, which is approximately $0.456343 s$ (corresponds to a frequency of $2.191334$) is in very good agreement with the reference.

\begin{figure}[hb]
\vspace*{-0.4cm}
  \begin{center}
    \includegraphics[width=\textwidth]{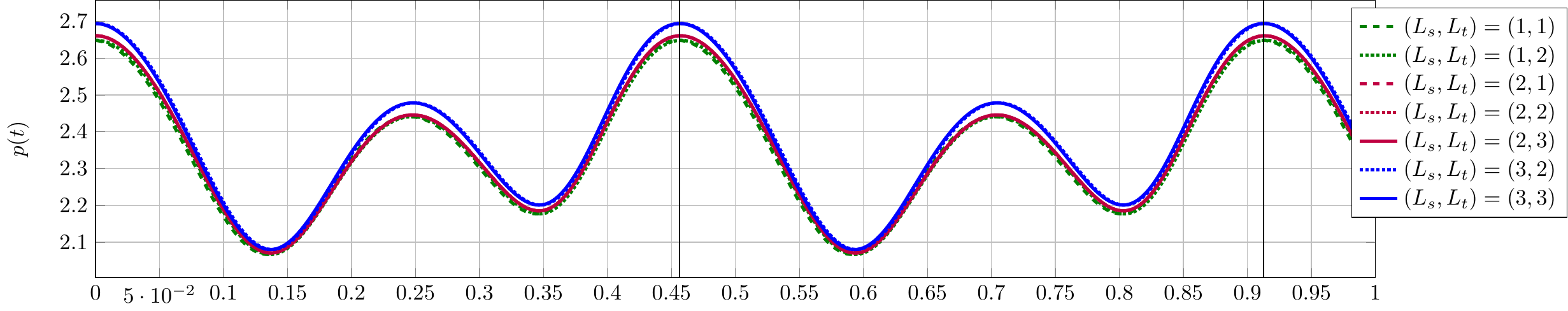}%
  \end{center}
\vspace*{-0.5cm}
  \caption{Evolution of pressure difference for mapping $\Phi_3$ for two periods (starting with a pressure maximum) for different spatial can temporal resolutions.}
\vspace*{-0.5cm}
  \label{fig:phi3:press}
\end{figure}

\subsection{The Kelvin-Helmholtz instability problem on surfaces} \label{sec::numex::KH}
In this section we consider the famous Kelvin-Helmholtz instability problem, cf. Ref. \cite{SJLLLS_CAMWA_2019} and the references therein, which is typically defined on the unit square (periodic in $x$-direction), i.e. in 2D. We generalize it to 2D surfaces in 3D. To this end we solve the unsteady Navier--Stokes equation. In contrast to the previous examples we do not rely on reference domains and mappings, but start directly from surface meshes that are obtained from the mesh generator.

\paragraph{The general setup}
In the following examples we consider geometries that are
rotational invariant around the $z$-axis and use a coordinate system $(\xi,\eta)$ on the surfaces with $\xi \in [-\tfrac12,\tfrac12)$ following the rotational direction and $\eta \in [-\tfrac12,\tfrac12)$ perpendicular. Note that $\xi$ and $\eta$ are not normalized, ie $\Vert \nabla_\Gamma \xi \Vert \neq 1$, $\Vert \nabla_\Gamma \eta \Vert \neq 1$. $e_{\xi} = \nabla_\Gamma \xi / \Vert \nabla_\Gamma \xi \Vert$ and $e_{\eta} = \nabla_\Gamma \eta / \Vert \nabla_\Gamma \eta \Vert$ denote the corresponding unit vectors (in the tangential plane of $\Gamma_h$) and $r = r(\eta)$ denotes the distance (in the ambient space) to the $z$-axis.
The Kelvin-Helmholtz instability is driven solely by its initial condition. These are taken as 
\begin{align}\label{eq:KH:ini_cond}
  u_0(\xi,\eta)
  & = H_s(\eta) \cdot u_\infty \cdot \tfrac{r}{R} \cdot  e_{\xi} + c_n \operatorname{curl}_\Gamma \psi, \nonumber \\
  \text{ with }
H_s(\eta) & := \tanh \left( \tfrac{2\eta}{\delta_0} \right), ~
  \psi (\xi,\eta)
		:= u_\infty\exp \left(- \tfrac{\eta^2}{\delta_0^2} \right)
			\left( a_a \cos (m_a \pi \xi) + a_b \cos (m_b \pi \xi) \right),
\end{align}
the constants $a_a,~m_a,~a_b,~m_b,~c_n, \delta_0 \geq 0$ and $R$ the radius at $z=0$. This means that for $z > \delta_0$ ($H_s \approx 1,~\operatorname{curl}_\Gamma \psi \approx 0$) the velocity field corresponds to a rigid body rotation around the $z$-axis in positive $\xi$-direction whereas for $z < - \delta_0$ the velocity corresponds to a rigid body rotation in the opposite direction. The $\tanh$ term realizes a smooth transition of the velocity in the intermediate layer which is determined by $\delta_0$. The terms related to $\curl_\Gamma \psi$ are a perturbation with the purpose to trigger an instability in a deterministic way.
The setup of the problem leads to a number of vortices forming along the shear layer. These vortices eventually pair up to form fewer but larger vortices. In this setup we consider the following three usually investigated global quantities of interest (gradients and curls are to understood in a broken sense, i.e. element-wise):
\begin{align*}
  \text{Enstrophy:} &\hspace*{-1.5cm}& \mathcal{E}(t) &:= \frac{1}{2} \Vert  \omega_h(t) \Vert_{L^2(\Gamma_h)}^2  \hspace*{-1.5cm}& \text{with vorticity} \quad  \omega_h(t) &:=\operatorname{curl}_{\Gamma} u_h(t),
  \\
  \text{Kinetic energy:} &\hspace*{-1.5cm}& \mathcal{K}(t) &:= \frac{1}{2} \Vert  u_h(t) \Vert_{L^2(\Gamma_h)}^2,
\hspace*{-1.5cm}& \text{Palinstrophy:} \quad \mathcal{P}(t) &:= \frac{1}{2} \Vert  \nabla_{\Gamma} \omega_h(t) \Vert_{L^2(\Gamma_h)}^2.
\end{align*}

\paragraph{Computational setup}
For the simulations discussed below we use -- unless stated otherwise -- the same configuration.
We fix the intermediate layer size $\delta_0 = 1/28$ fix $\nu = \delta_0/ Re$ with $Re=1000$ and set the perturbation constants to $a_a=1,~m_a=8,~a_b=1,~m_b=20$, $c_n=10^{-3}$. With the reference time $t_{\text{ref}} := \tfrac{u_\infty}{\delta_0}$ we introduce the scaled time $\bar{t} = t / t_{\text{ref}}$. For the computations we  consider a triangular unstructured mesh with characteristic mesh size $h=0.05$, $k_u=8$, $k_g=9$, use the second order IMEX time stepping method as before with $\Delta t = 10^{-3} t_{\text{ref}}$ and simulate until $T=200 t_{\text{ref}}$. For $\alpha$ in the SIP stabilization we take $10$ again. To initialize the velocity we use a Helmholtz projection in order to make sure that the discrete initial velocity is already exactly divergence-free, i.e. in \eqref{eq:nse:initial2} we consider $u_h,~v_h$ in the divergence-free subspace of $V_h^{k_u}$, i.e. $\{ v_h \in V_h^{k_u}: b(v_h,q_h) = 0 ~\forall q_h \in Q_h^{k_u-1}\}$.

\begin{figure}[hb]
\vspace*{-0.5cm}
  \includegraphics[height=0.2\textheight,trim=0.2cm 0.07cm 0.2cm 0.1cm, clip=true]{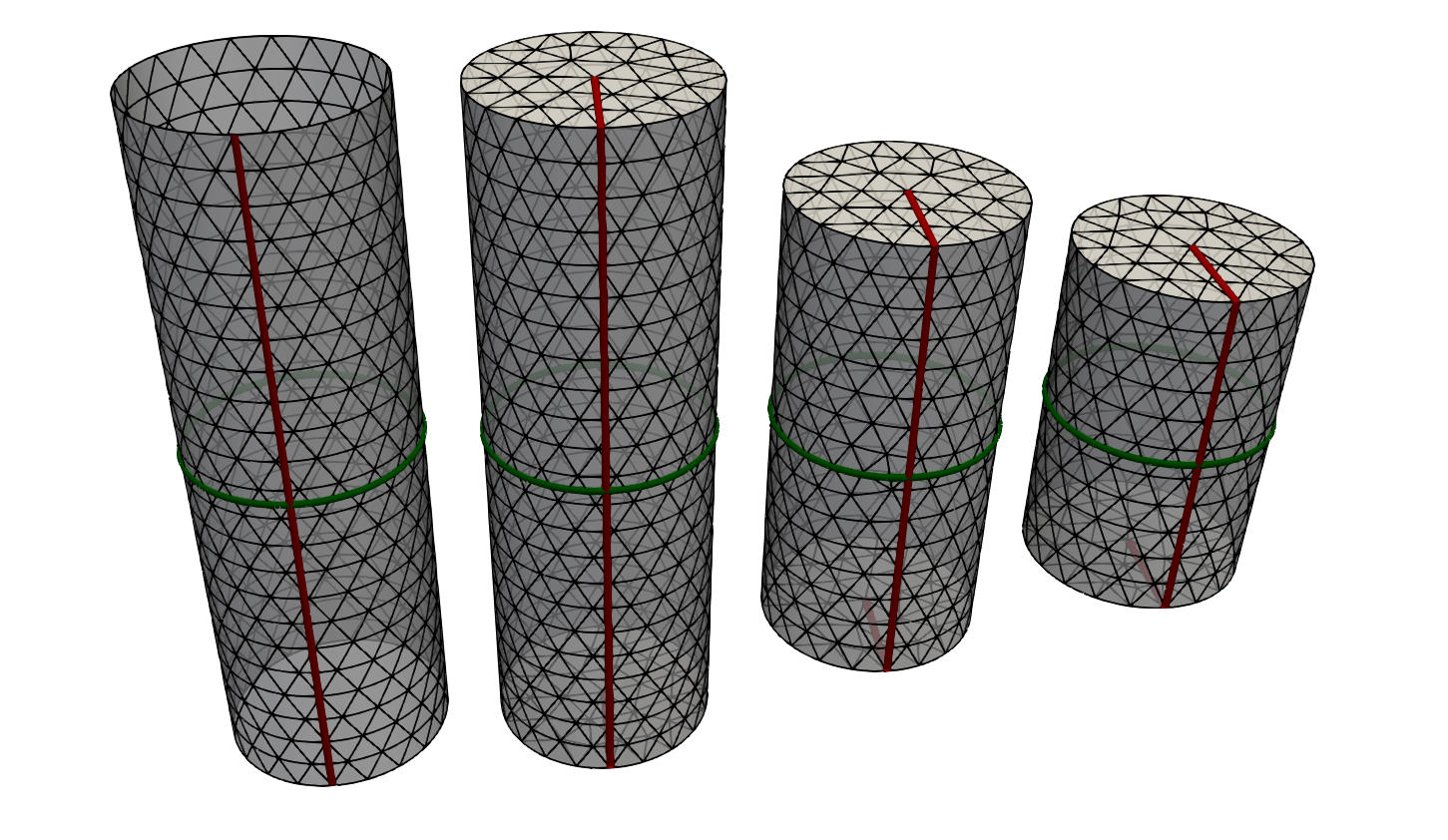}%
  \includegraphics[height=0.225\textheight]{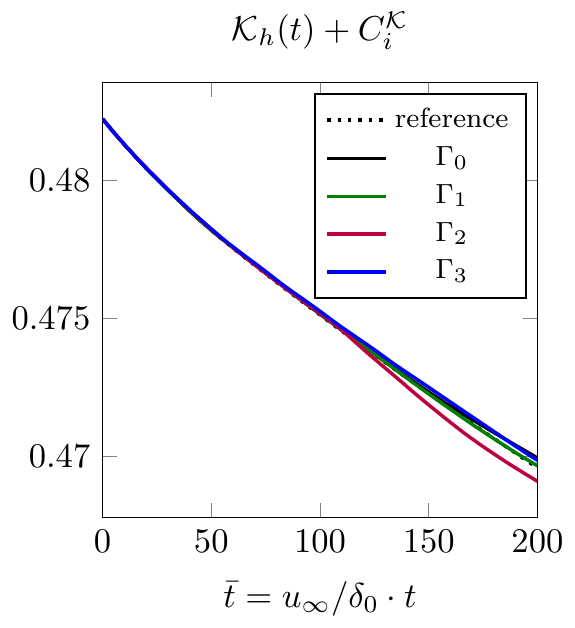}%
  \includegraphics[height=0.225\textheight]{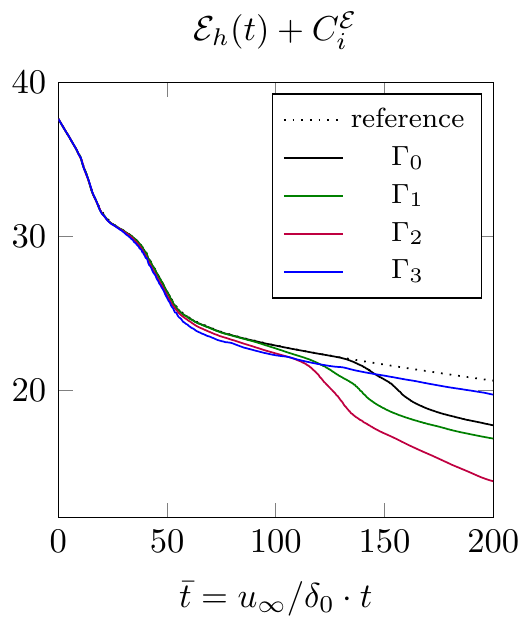}
\vspace*{-0.25cm}
  \caption{Geometries $\Gamma_i,~i=0,1,2,3$ and corresponding meshes (left). Red lines correspond to $\xi=0$, green lines correspond to $\eta=0$. Plots of decay of kinetic energy and enstrophy (right). For comparibility offsets $C_i^\mathcal{K}$, $C_i^\mathcal{E}$, $i=0,1,2,3$
    are added to compensate for differences in initial kinetic energy and initial enstrophy. With $C_0^{\mathcal{K}}=0$, $C_1^{\mathcal{K}}=-\frac{1}{8\pi}$, $C_2^{\mathcal{K}}=\frac{3}{8\pi}$, $C_3^{\mathcal{K}}=-\frac{1}{8\pi}+\frac{1}{4}$ and
$C_0^{\mathcal{E}}=0$, $C_1^{\mathcal{E}}=C_2^{\mathcal{E}}=C_3^{\mathcal{E}}=-4\pi$ the plotted quantities are equal (up to resolution differences) at time $t=0$.}
\vspace*{-0.5cm}
  \label{fig:energyenstrophy}
\end{figure}

\subsubsection{Piecewise smooth manifolds}
In this subsection we consider 4 similar cylindrical setups given by:
\begin{align*}
  \Gamma_0 &:= \{ x \in \rr^3 \mid \Vert (x,y) \Vert = R, |z| \leq 1/2\}, \\
  \Gamma_1 &:= \Gamma_0 \cup \{ x \in \rr^3 \mid \Vert (x,y) \Vert \leq R, \vert z \vert = 1/2 \}, \\
  \Gamma_2 &:= \{ x \in \rr^3 \mid \Vert (x,y) \Vert = R, |z| \leq 1/2-R \} \cup \{ x \in \rr^3 \mid \Vert (x,y) \Vert \leq R, \vert z \vert = 1/2-R \}, \\
  \Gamma_3 &:= \{ x \in \rr^3 \mid \Vert (x,y) \Vert = R, |z| \leq 1/4 \} \cup \{ x \in \rr^3 \mid \Vert (x,y) \Vert \leq R, \vert z \vert = 1/4 \}.
\end{align*}
The first setting, $\Gamma_0$ is an open cylinder of height $1$ with radius $R=(2\pi)^{-1}$, i.e. perimeter $1$ and we can isometrically map the unit square (periodic in $x$-direction) on $\Gamma_0$. On the boundary we prescribe free slip boundary condition. As the surface Navier--Stokes equations are invariant under isometric maps we know that the solution to the corresponding 2D Kelvin--Helmholtz problem is identical. We can hence compare our numerical solution on $\Gamma_0$ to the results in the literature\cite{SJLLLS_CAMWA_2019}.

The second configuration, $\Gamma_1$, is a closed cylinder with bottom and top added, i.e. without boundary. $\Gamma_2$ is similar to $\Gamma_1$ except for the decreased height of $1-2R$. Hence, the geodesics from the center of the top of the cylinder to the center of the bottom of the cylinder have length $1$. The last case, case $3$ considers an even shorter closed cylinder with height $\tfrac12$. In Fig.~\ref{fig:energyenstrophy} the geometries and used meshes are sketched alongside with the decay of energy and enstrophy over time whereas in Fig. \ref{fig:palinstrophy} we plot the palinstrophy alongside a few sketches of the vorticity at selected times.

\begin{figure}
\vspace*{0.25cm}
  \includegraphics[width=\textwidth]{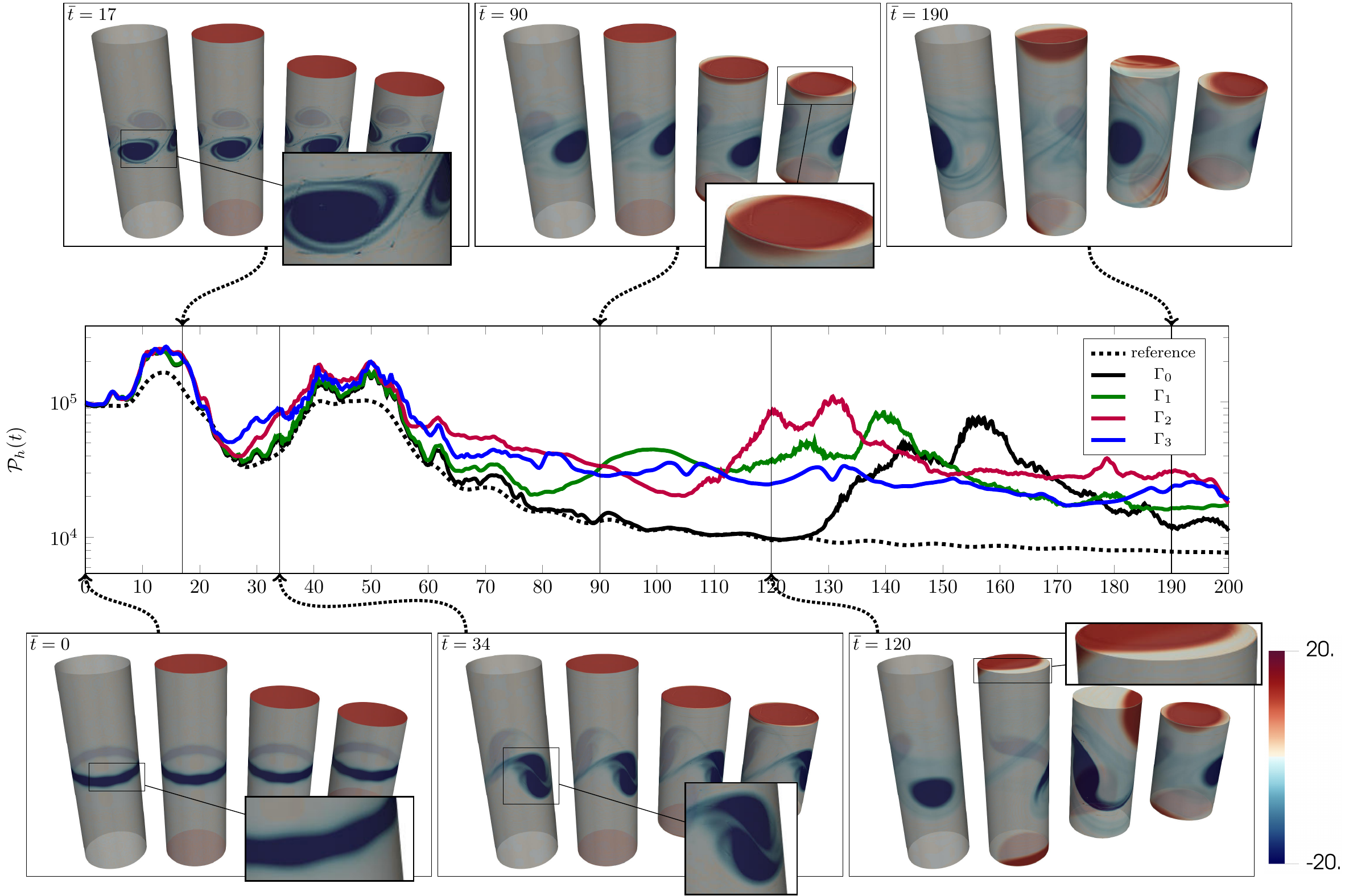}%
  \caption{Evolution of $\mathcal{P}_h$ for $\Gamma_i,~i=0,1,2,3$ and the 2D reference solution with vorticity snapshots.}
\vspace*{0.25cm}
  \label{fig:palinstrophy}
\end{figure}

\vspace*{0.5cm}
\paragraph{Discussion of results}
We observe that the energy dissipation is similar for all cases (up to a constant shift due to different initial kinetic energies). Especially, the kinetic energy is monotonely decreasing. The same holds for the enstrophy for $t \leq 100$. For later times the results deviate significantly. The reason for those deviations can be best explained by the palinstrophy evolution in Fig. \ref{fig:palinstrophy}. Until $\bar t = 100$ the simulations of the four cases agree very much, at least qualitatively. Initially four vortices form which eventually merge to two vortices around $\bar t \in (30,60)$. Until that point of the evolution the rigid body rotations in the most upper and most lower part are essentially not influenced by the interactions in the center. This changes at around $\bar t = 90$ where the rotations of the top and bottom are perturbed for $\Gamma_2$ and $\Gamma_3$ and somewhat later (around $\bar t = 120$) also for $\Gamma_1$. For $\Gamma_1$ and $\Gamma_2$ this results in the merging of the latest two vortices to one vortex which is also reflected by an increase in the palinstrophy. For $\Gamma_3$ the perturbation of the rotations seems to stay confined and the decreased height seems to surpress the interaction of the vortices so that even at $\bar t = 200$ the latest two vortices did not merge yet.

Let us note that we also observe a deviation of the evolution computed for $\Gamma_0$ and the reference solution, especially after $\bar t = 130$ which is in agreement with the extreme sensitivity of the problem to (numerical) perturbations observed in Ref.\cite{SJLLLS_CAMWA_2019}. The final merge is typically observed sooner the higher the perturbations in the simulation are. In contrast to the 2D reference solution we consider a much coarser, not structured and not symmetric mesh with an additional geometry error due to the curved representation.


\subsubsection{Smooth manifold: the sphere}
Now, we consider the Kelvin-Helmholtz problem on a smooth manifold, the unit sphere, $\Gamma_4 = S_{1}(0)$. The upper and lower halfs rotate in different directions with a perturbation similar as before in \eqref{eq:KH:ini_cond}. However, we change the perturbation magnitude to $c_n = 2 \pi \cdot 10^{-3}$ (due to the increase of the equador length from $1$ to $2 \pi$) and choose a different initial perturbation mode with $m_a = 16$, $a_a = 1$, $m_b = 20$ and $a_b = 0.1$.

We use $k_u = 5$, $k_g=6$ on a unstructured mesh consisting of $5442$ triangles. A few snapshots of the flow and the evolution of the palinstrophy are depicted in Fig. \ref{fig:palinstrophy_sphere}. We observe that initially eight vortices form from the initial perturbation at around $\bar t = 50$. Diffusion takes its time until it drives the interaction of two neighboring vortices which pair up to 4 vortices at around $\bar t = 200$. These vortices take even longer to eventually pair up to two vortices at around $\bar t = 500$. The two vortices are positioned on opposite sides rotating in opposite directions. The overall evolution of the palinstrophy and the times where vortices pair up is very similar to those from the cylindrical setup in the previous section. However, due to the larger length scales the interaction between the vortices takes more time.

\begin{figure}[hb!]
 \vspace*{0.25cm}
  \includegraphics[width=\textwidth]{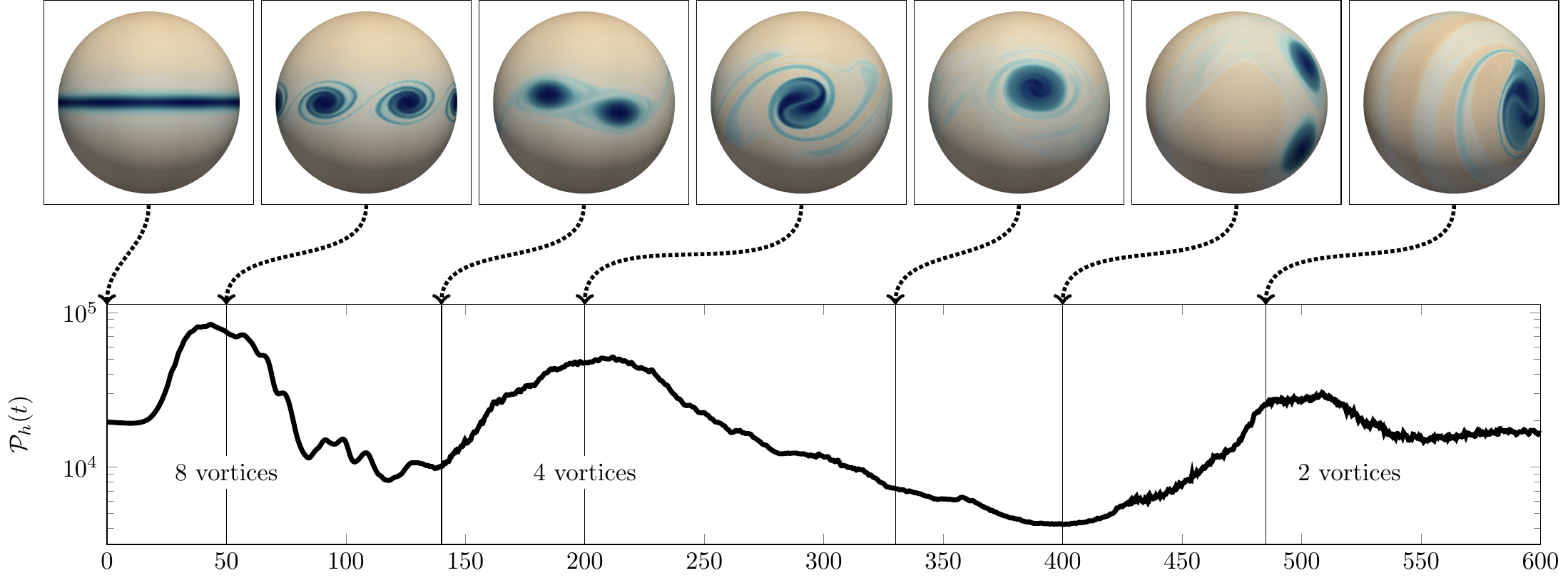}%
  \caption{Evolution of $\mathcal{P}_h$ for the Kelvin-Helmholtz problem on the sphere $\bar t \in \{0,50,140,200,330,400,485\}$.}
  \label{fig:palinstrophy_sphere}
\end{figure}

\subsection{Stanford bunny} \label{sec::bunny}
As a final example we want to demonstrate that the aforementioned methods can be applied on essentially arbitrary geometries. As a prototype of a complex geometry we take the famous Stanford bunny which fits approximately in a bounding box of size $80\times 80 \times 60$ and initially put 60 vortices on the surface.
Let us note that in Reuther and Voigt\cite{RV15} this geometry has also been considered (with different initial data) with the aim to investigate defect configurations which is not the aim here. 
The initial velocity is taken as the surface curl of
$$
\psi(x) = 20 \cdot \sum_{i=1}^{60} (-1)^i \exp(-\tfrac{1}{20} \Vert x - x_i \Vert^2) 
$$
where the $x_i,~i=1,..,60$ are some randomly located but sufficiently separated ($\Vert x_i - x_j \Vert_2 \geq 3.5,~i\neq j$) vertices of the mesh. For the viscosity we choose $\nu = \tfrac{1}{50}$. The flow is again only driven by its initial condition. Pictures of the numerical solution on a mesh with 6054 triangles, $k_u=k_g=4$, time step size $\Delta t = \tfrac{1}{200}$ and $\alpha = 80$ are shown in Fig. \ref{fig:stanford_bunny}. 

We observe that the flow undergoes a process of self-organization as it is well-known from 2D flows. At $t=T=500$ the vortices merged successively into two remaining vortices: One around the ears of the Standford bunny, one on the breast. In contrast to the studies in Reuther and Voigt\cite{RV15}, the flow is not yet in a quasi stationary state at $t=500$, i.e. these two vortices will continue moving and possibly merge while the magnitude of the velocity will decay over time.

\begin{figure}[h]
  \begin{center}
  \includegraphics[width=0.31\textwidth,trim= 3.25cm 0cm 3.25cm 0cm, clip=True]{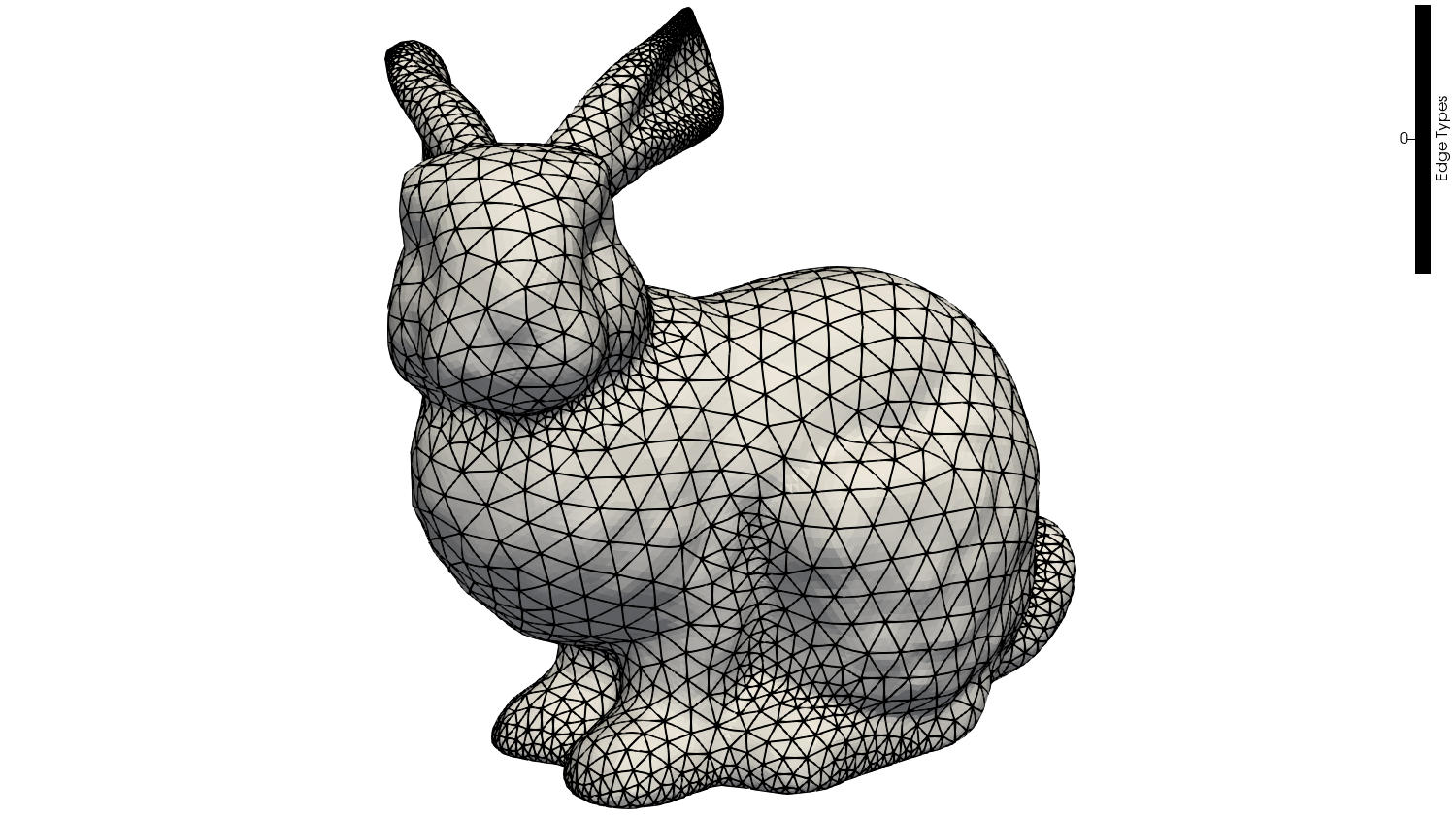}%
  \includegraphics[width=0.31\textwidth,trim= 3.25cm 0cm 3.25cm 0cm, clip=True]{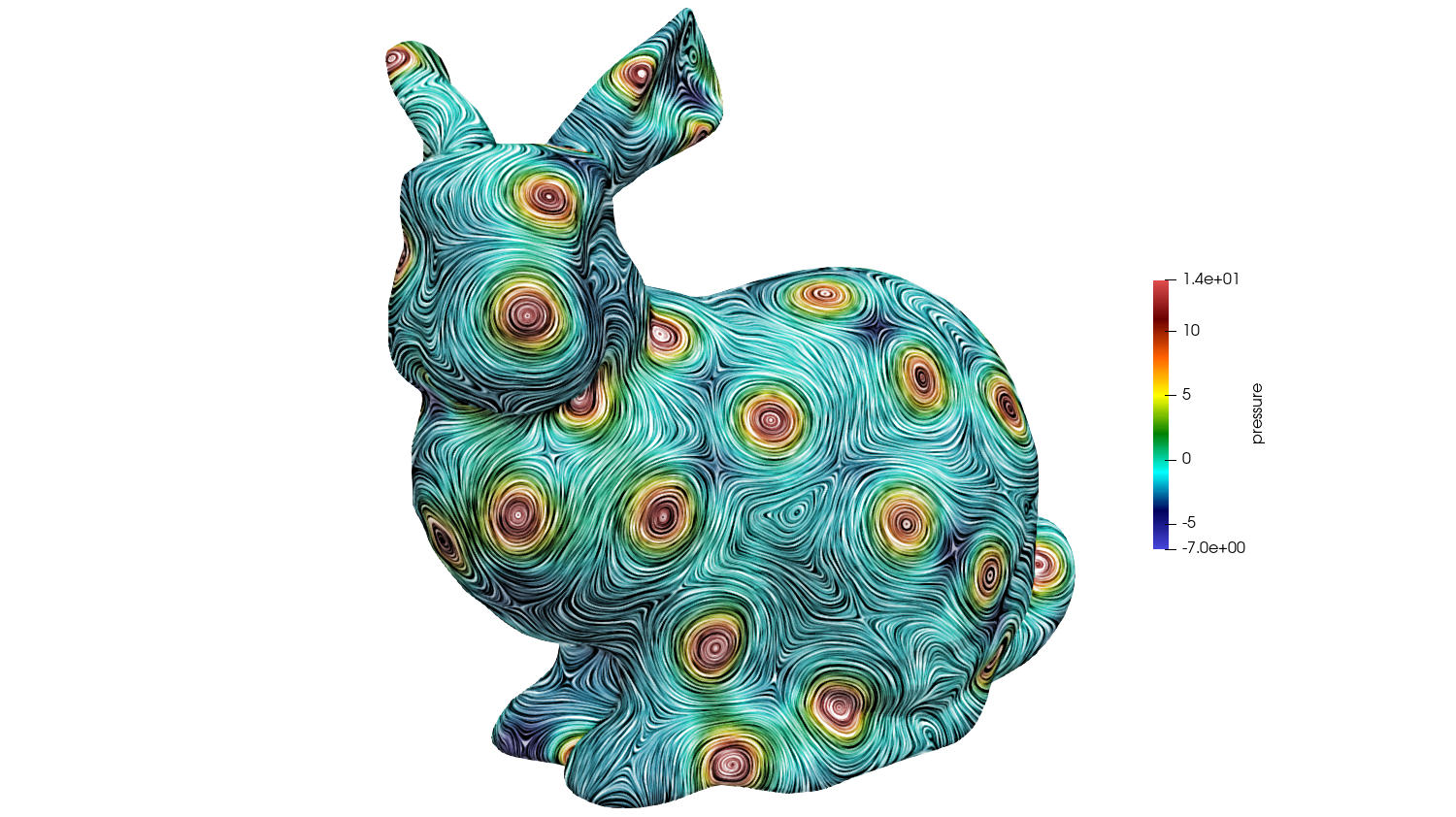}%
  \includegraphics[width=0.31\textwidth,trim= 3.25cm 0cm 3.25cm 0cm, clip=True]{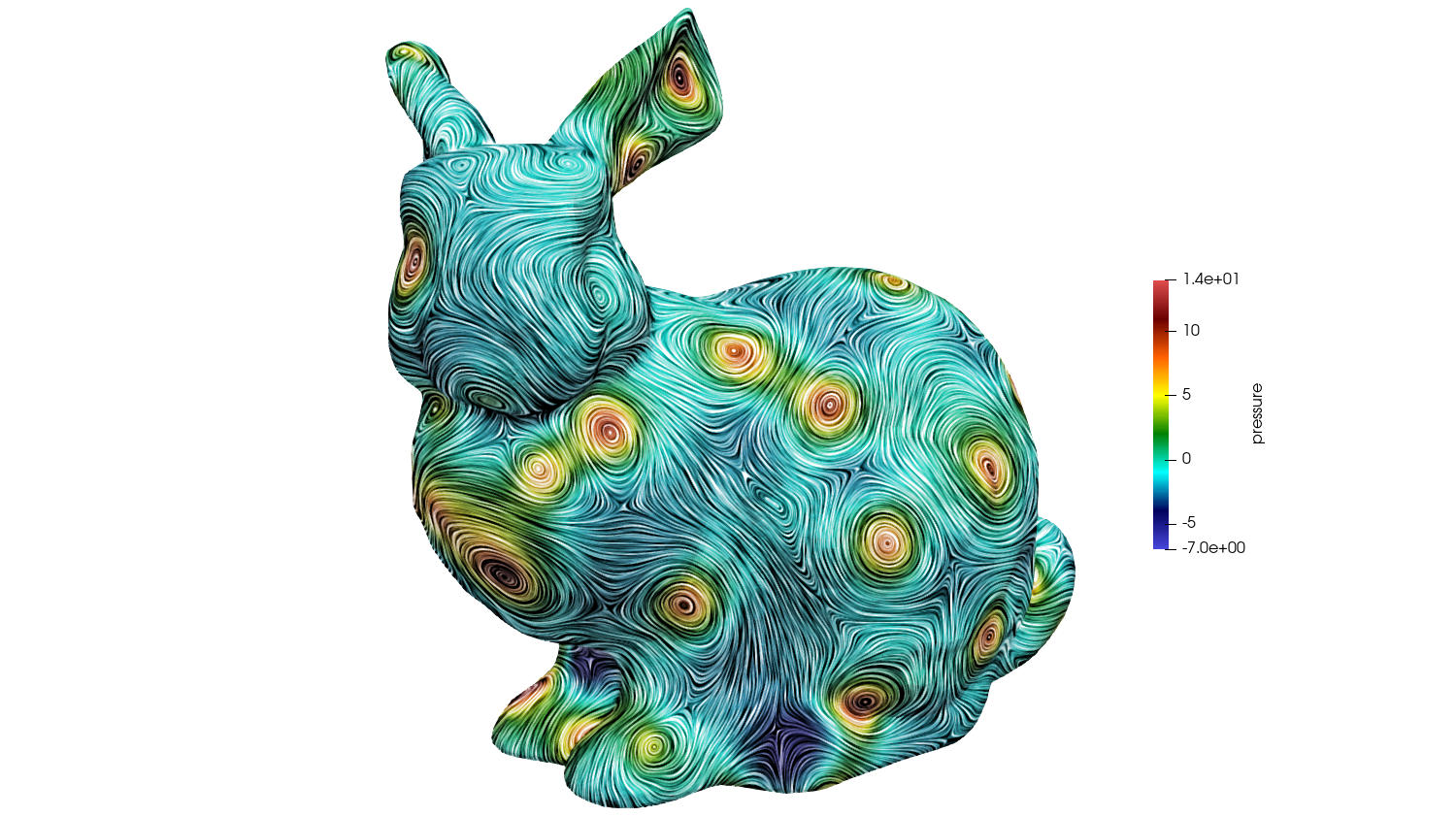}%
  \begin{minipage}{0.04\textwidth}
    \vspace*{-6.5cm}
    \hspace*{0.2cm}
    \includegraphics[width=\textwidth]{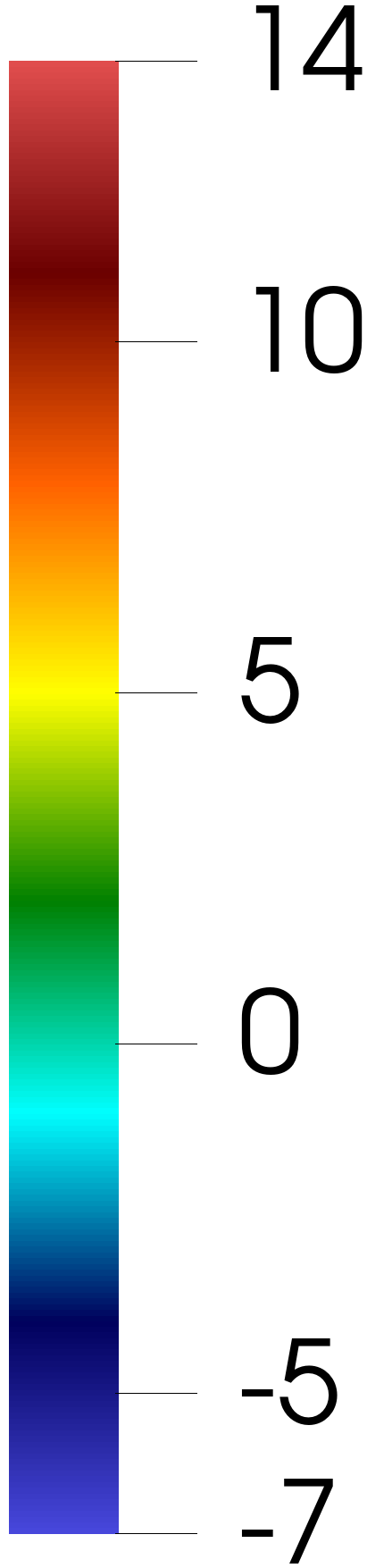}
  \end{minipage}
  \\ 
  \includegraphics[width=0.25\textwidth,trim= 3.25cm 0cm 3.25cm 0cm, clip=True]{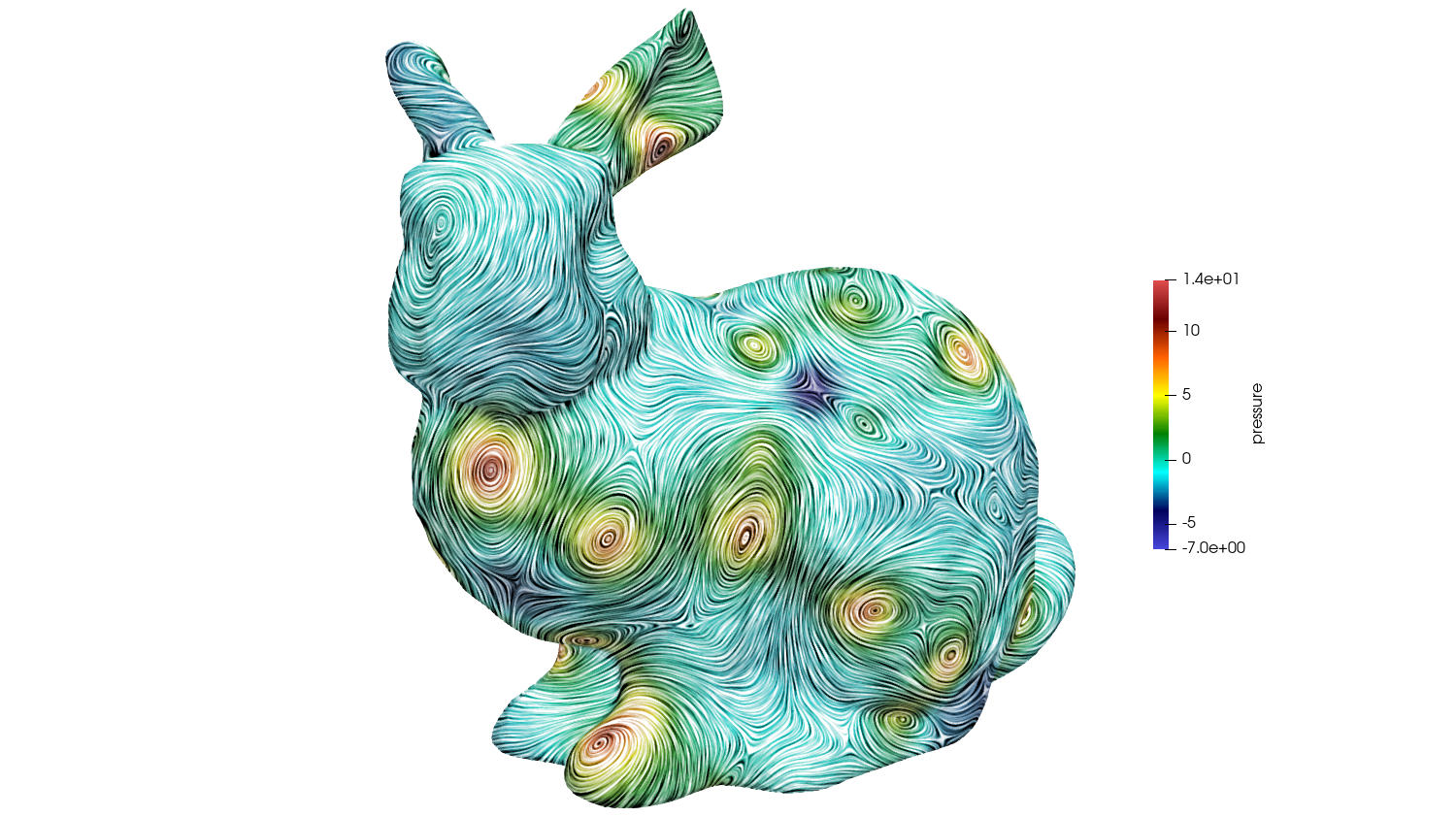}%
  \includegraphics[width=0.25\textwidth,trim= 3.25cm 0cm 3.25cm 0cm, clip=True]{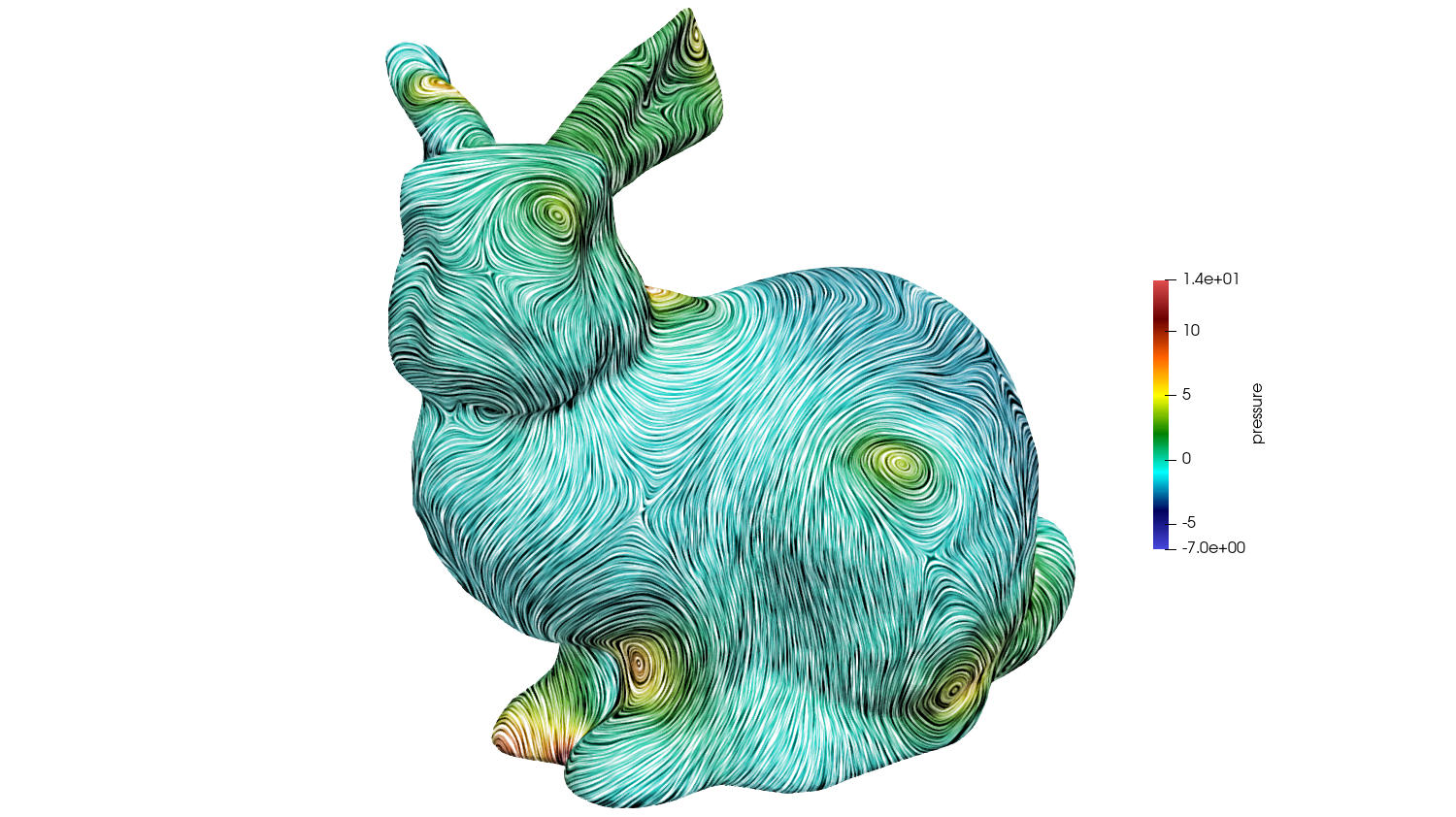}%
  \includegraphics[width=0.25\textwidth,trim= 3.25cm 0cm 3.25cm 0cm, clip=True]{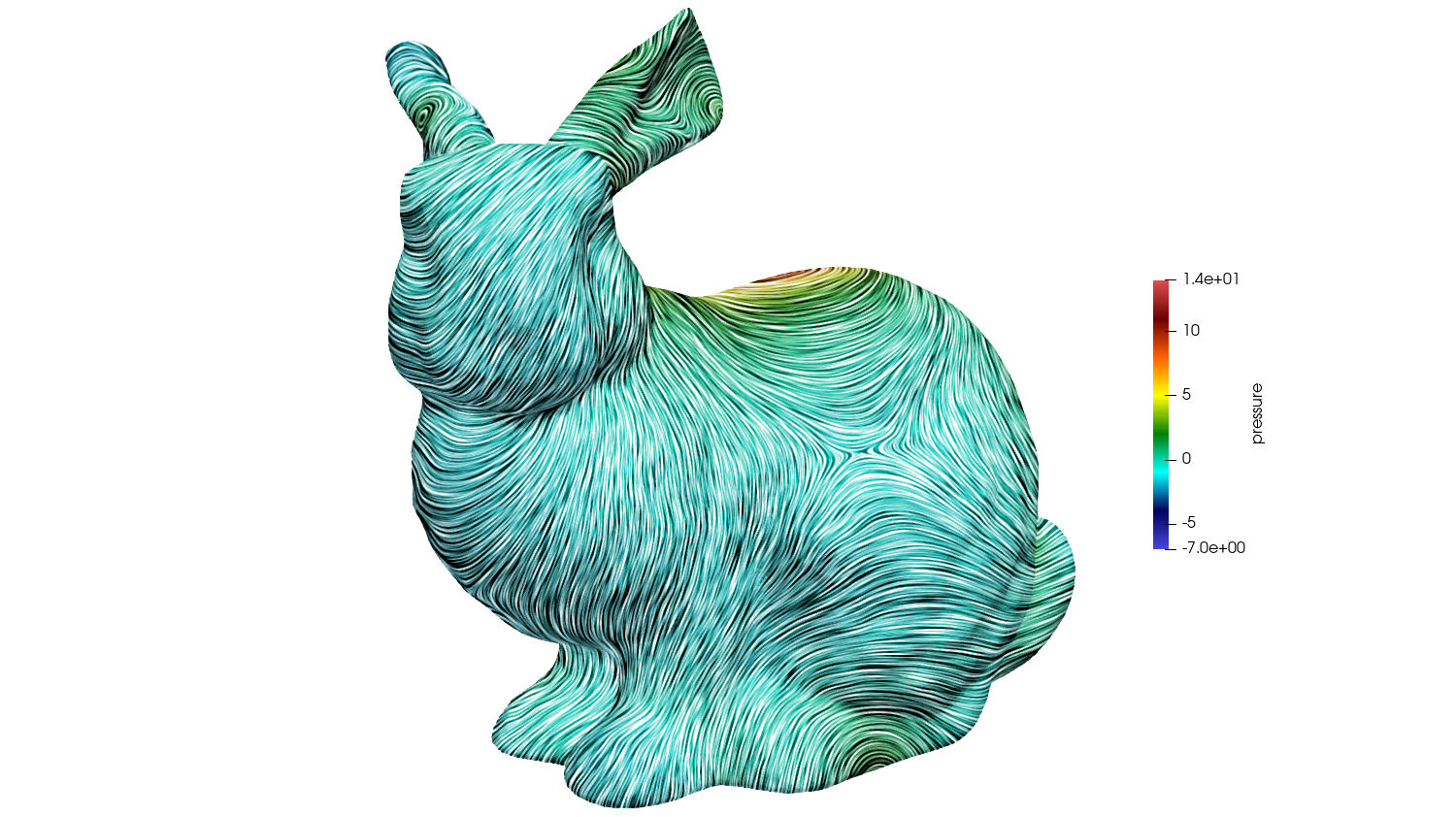}%
  \includegraphics[width=0.25\textwidth,trim= 3.25cm 0cm 3.25cm 0cm, clip=True]{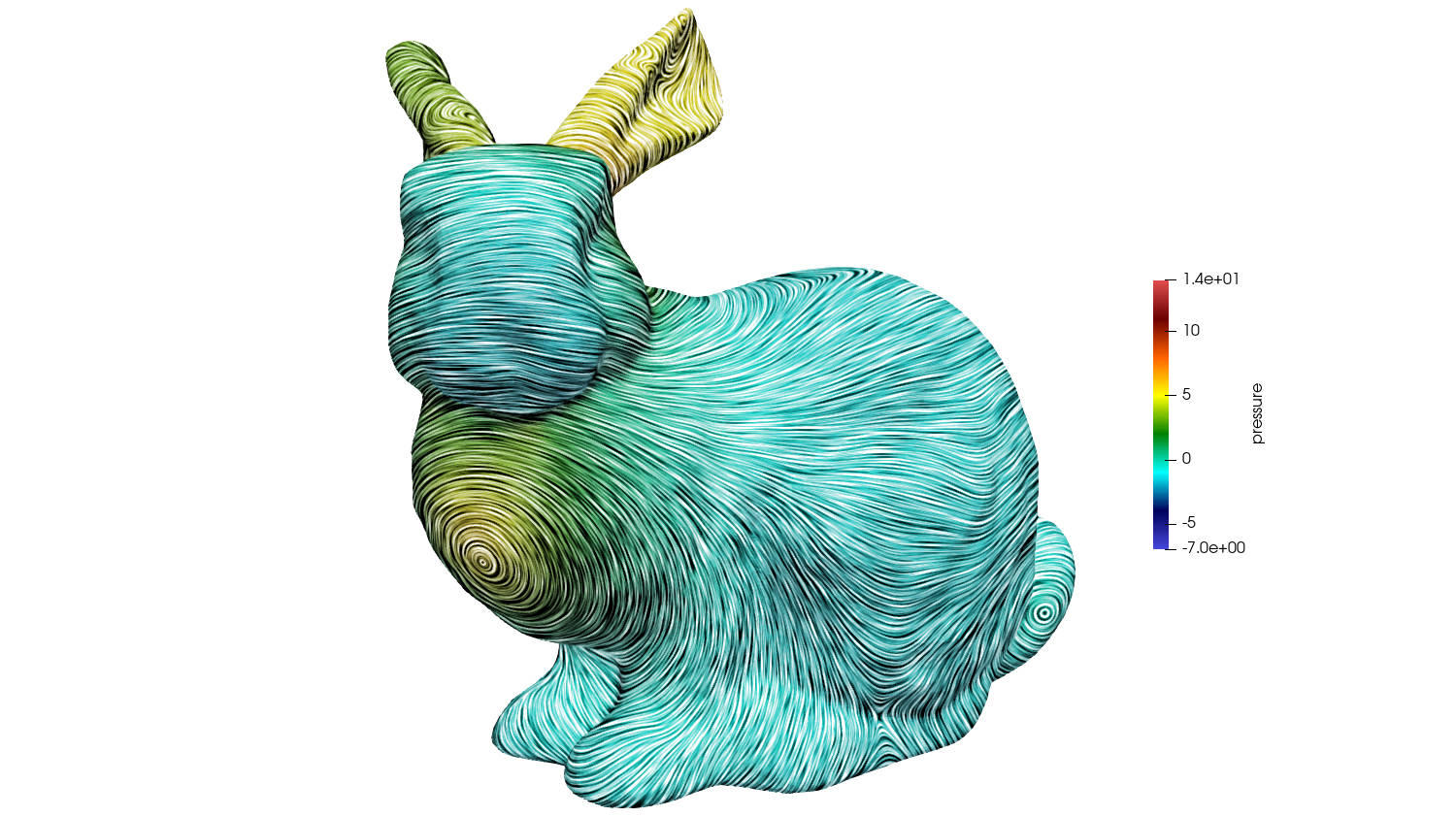}%
\end{center}
\vspace*{-0.5cm}
\caption{Geometry and mesh of the Stanford bunny (upper left picture) and velocity streamlines with pressure coloring for $t\in\{4,32,60,120,250,500\}$.}
\vspace*{-0.5cm}
  \label{fig:stanford_bunny}
\end{figure}


\section{Conclusion and outlook}
In this work we introduced new numerical methods for the discretization of incompressible flows and vector valued elliptic problems on two-dimensional manifolds. Abandoning the $H^1$-conformity (originally demanded by the considered problems) we applied the Piola transformation to construct finite elements which are exactly tangential. Based on these findings we presented non-($H^1(\Gamma)$-)conforming (hybrid) discontinuous Galerkin discretizations which performed extremely well in several considered numerical examples. Among other benefits, it was shown that the resulting methods can outperform $H^1$-conforming discretizations in the aspect of computational costs, and that they can deal with piecewise smooth manifolds which is a unique property so far.

For incompressible flow problems the (hybrid) discontinuous Galerkin methods were modified to be $H(\divergence_{\Gamma})$-conforming. The resulting methods provide important properties such as exactly (surface) divergence-free velocity fields, energy stability and pressure robustness. We studied certain numerical examples showing that the methods are highly accurate and applicable to deal with complex geometries. In particular we also showed a robustness of the new methods with respect to isometric mappings.

The extension of the discretizations to non stationary (evolving) manifolds is left for future research. Further, due to the increasing interest on foldable LC displays and flows on cell membranes, the coupling of the presented discretizations with other PDEs on surfaces is an interesting task for the future.

\subsection*{Acknowledgements}

Philip L. Lederer has been funded by the Austrian Science Fund (FWF) through the research programm ``Taming complexity in partial differential systems'' (F65) - project ``Automated discretization in multiphysics'' (P10).

Furthermore we would like to acknowledge the Stanford University Computer Graphics Laboratory for providing the 3D-model of the Stanford bunny.

\bibliography{literature}

\begin{thebibliography}{10}
\providecommand \doibase [0]{http://dx.doi.org/}%

\bibitem{slattery2007interfacial}
Slattery JC, Sagis L, Oh ES. {\it Interfacial transport phenomena}.
\newblock Springer Science \& Business Media .
\newblock 2007.

\bibitem{brenner2013interfacial}
Brenner H. {\it Interfacial transport processes and rheology}.
\newblock Elsevier .
\newblock 2013.

\bibitem{de1993physics}
De~Gennes PG, Prost J. {\it The physics of liquid crystals}. 83.
\newblock Oxford university press .
\newblock 1993.

\bibitem{napoli2016hydrodynamic}
Napoli G, Vergori L. Hydrodynamic theory for nematic shells: The interplay
  among curvature, flow, and alignment. {\it Physical Review E} 2016\string;
  94(2)\string: 020701.

\bibitem{dziuk1988finite}
Dziuk G. Finite elements for the Beltrami operator on arbitrary surfaces. In:
  Springer.  1988 (pp. 142--155).

\bibitem{dziuk_elliott_2013}
Dziuk G, Elliott CM. Finite element methods for surface PDEs. {\it Acta
  Numerica} 2013\string; 22\string: 289-396.
\newblock \href {\doibase 10.1017/S0962492913000056} {doi:
  10.1017/S0962492913000056}

\bibitem{2019arXiv190602786B}
{Bonito} A, {Demlow} A, {Nochetto} RH. {Finite Element Methods for the
  Laplace-Beltrami Operator}. {\it arXiv e-prints} 2019\string:
  arXiv:1906.02786.

\bibitem{demlow2007adaptive}
Demlow A, Dziuk G. An adaptive finite element method for the Laplace--Beltrami
  operator on implicitly defined surfaces. {\it SIAM Journal on Numerical
  Analysis} 2007\string; 45(1)\string: 421--442.

\bibitem{demlow2009higher}
Demlow A. Higher-order finite element methods and pointwise error estimates for
  elliptic problems on surfaces. {\it SIAM Journal on Numerical Analysis}
  2009\string; 47(2)\string: 805--827.

\bibitem{holst2012geometric}
Holst M, Stern A. Geometric variational crimes: Hilbert complexes, finite
  element exterior calculus, and problems on hypersurfaces. {\it Foundations of
  Computational Mathematics} 2012\string; 12(3)\string: 263--293.

\bibitem{dedner2013analysis}
Dedner A, Madhavan P, Stinner B. Analysis of the discontinuous Galerkin method
  for elliptic problems on surfaces. {\it IMA Journal of Numerical Analysis}
  2013\string; 33(3)\string: 952--973.

\bibitem{antonietti2014high}
Antonietti P, Dedner A, Madhavan P, Stangalino S, Stinner B, Verani M. High
  order discontinuous Galerkin methods on surfaces. {\it arXiv preprint
  arXiv:1402.3428} 2014.

\bibitem{cockburn2016hybridizable}
Cockburn B, Demlow A. Hybridizable discontinuous Galerkin and mixed finite
  element methods for elliptic problems on surfaces. {\it Mathematics of
  Computation} 2016\string; 85(302)\string: 2609--2638.

\bibitem{olshanskii2009finite}
Olshanskii MA, Reusken A, Grande J. A finite element method for elliptic
  equations on surfaces. {\it SIAM journal on numerical analysis} 2009\string;
  47(5)\string: 3339--3358.

\bibitem{burman2015stabilized}
Burman E, Hansbo P, Larson MG. A stabilized cut finite element method for
  partial differential equations on surfaces: The Laplace--Beltrami operator.
  {\it Computer Methods in Applied Mechanics and Engineering} 2015\string;
  285\string: 188--207.

\bibitem{OlRe_2017}
Olshanskii M, Reusken A. Trace Finite Element Methods for PDEs on Surfaces. In:
   Bordas SPA, Burman E, Larson MG, Olshanskii MA. \kern-2pt, eds. {\it
  Geometrically Unfitted Finite Element Methods and Applications}Springer
  International Publishing.  2017 (pp. 211--258)

\bibitem{rognes2013automating}
Rognes ME, Ham DA, Cotter CJ, McRae AT. Automating the solution of PDEs on the
  sphere and other manifolds in FEniCS 1.2. {\it Geoscientific Model
  Development} 2013\string; 6(6)\string: 2099--2119.

\bibitem{burman2017cut}
Burman E, Hansbo P, Larson MG, Massing A. A cut discontinuous Galerkin method
  for the Laplace--Beltrami operator. {\it IMA Journal of Numerical Analysis}
  2017\string; 37(1)\string: 138--169.

\bibitem{nedelec1978computation}
Nedelec J. Computation of eddy currents on a surface in R\^{}3 by finite
  element methods. {\it SIAM Journal on Numerical Analysis} 1978\string;
  15(3)\string: 580--594.

\bibitem{bendali1984numerical}
Bendali A. Numerical analysis of the exterior boundary value problem for the
  time-harmonic Maxwell equations by a boundary finite element method. II. The
  discrete problem. {\it Mathematics of Computation} 1984\string;
  43(167)\string: 47--68.

\bibitem{MR3875687}
Jankuhn T, Olshanskii MA, Reusken A. Incompressible fluid problems on embedded
  surfaces: modeling and variational formulations. {\it Interfaces Free Bound.}
  2018\string; 20(3)\string: 353--377.
\newblock \href {\doibase 10.4171/IFB/405} {doi: 10.4171/IFB/405}

\bibitem{koba2017energetic}
Koba H, Liu C, Giga Y. Energetic variational approaches for incompressible
  fluid systems on an evolving surface. {\it Quarterly of Applied Mathematics}
  2017\string; 75(2).

\bibitem{hansbo2016analysis}
Hansbo P, Larson MG, Larsson K. Analysis of finite element methods for vector
  Laplacians on surfaces. {\it arXiv preprint arXiv:1610.06747} 2016.

\bibitem{MR3840893}
Gross S, Jankuhn T, Olshanskii MA, Reusken A. A trace finite element method for
  vector-{L}aplacians on surfaces. {\it SIAM J. Numer. Anal.} 2018\string;
  56(4)\string: 2406--2429.
\newblock \href {\doibase 10.1137/17M1146038} {doi: 10.1137/17M1146038}

\bibitem{jankuhn2019trace}
Jankuhn T, Reusken A. Trace finite element methods for surface vector-Laplace
  equations. {\it arXiv preprint arXiv:1904.12494} 2019.

\bibitem{hansbo2017stabilized}
Hansbo P, Larson MG. A stabilized finite element method for the Darcy problem
  on surfaces. {\it IMA Journal of Numerical Analysis} 2017\string;
  37(3)\string: 1274--1299.

\bibitem{Olshanskii2018AFE}
Olshanskii MA, Quaini A, Reusken A, Yushutin V. A Finite Element Method for the
  Surface Stokes Problem. {\it SIAM J. Scientific Computing} 2018\string;
  40\string: A2492-A2518.

\bibitem{2019arXiv190902990O}
{Olshanskii} MA, {Reusken} A, {Zhiliakov} A. {Inf-sup stability of the trace
  $\mathbf P_2$$-$$P_1$ Taylor$-$Hood elements for surface PDEs}. {\it arXiv
  e-prints} 2019\string: arXiv:1909.02990.

\bibitem{bonito2019}
Bonito A, Demlow A, Licht M. A divergence-conforming finite element method for
  the surface Stokes equation. {\it arXiv preprint arXiv:1908.11460} 2019.

\bibitem{NNPV18}
Nestler M, Nitschke I, Praetorius S, Voigt A. Orientational Order on Surfaces:
  The Coupling of Topology, Geometry, and Dynamics. {\it J. Comput. Phys.}
  2018\string; 28(1)\string: 147--191.

\bibitem{NNV19}
Nestler M, Nitschke I, Voigt A. A finite element approach for vector-and
  tensor-valued surface PDEs. {\it J. Nonlinear Sci.} 2019\string; 389\string:
  48--61.

\bibitem{RV15}
Reuther S, Voigt A. The Interplay of Curvature and Vortices in Flow on Curved
  Surfaces. {\it Multiscale Model. Sim.} 2015\string; 13(2)\string: 632--643.

\bibitem{NRV17}
Nitschke I, Reuther S, Voigt A. Discrete Exterior Calculus (EC) for the Surface
  Navier-Stokes Equation. In: Springer. ; 2017\string: 177--197.

\bibitem{reusken2018stream}
Reusken A. Stream function formulation of surface Stokes equations. {\it IMA
  Journal of Numerical Analysis} 2018.

\bibitem{Azencot2014}
Azencot O, Wei{\ss}mann S, Ovsjanikov M, Wardetzky M, Ben-Chen M. Functional
  Fluids on Surfaces. {\it Computer Graphics Forum} 2014\string; 33(5)\string:
  237-246.
\newblock \href {\doibase 10.1111/cgf.12449} {doi: 10.1111/cgf.12449}

\bibitem{nitschke2012finite}
Nitschke I, Voigt A, Wensch J. A finite element approach to incompressible
  two-phase flow on manifolds. {\it Journal of Fluid Mechanics} 2012\string;
  708\string: 418--438.

\bibitem{voigtreuther18}
Reuther S, Voigt A. Solving the incompressible surface Navier-Stokes equation
  by surface finite elements. {\it Physics of Fluids} 2018\string;
  30(1)\string: 012107.
\newblock \href {\doibase 10.1063/1.5005142} {doi: 10.1063/1.5005142}

\bibitem{MR3846120}
Fries TP. Higher-order surface {FEM} for incompressible {N}avier-{S}tokes flows
  on manifolds. {\it Internat. J. Numer. Methods Fluids} 2018\string;
  88(2)\string: 55--78.
\newblock \href {\doibase 10.1002/fld.4510} {doi: 10.1002/fld.4510}

\bibitem{Olshanskii2018APF}
Olshanskii MA, Yushutin V. A Penalty Finite Element Method for a Fluid System
  Posed on Embedded Surface. {\it Journal of Mathematical Fluid Mechanics}
  2018\string; 21\string: 1-18.

\bibitem{GA18}
Gross BJ, Atzberger PJ. Hydrodynamic flows on curved surfaces: Spectral
  numerical methods for radial manifold shapes. {\it J. Comput. Phys.}
  2018\string; 371\string: 663--689.

\bibitem{cockburn2007note}
Cockburn B, Kanschat G, Sch{\"o}tzau D. A note on discontinuous {Galerkin}
  divergence-free solutions of the {Navier}--{Stokes} equations. {\it Journal
  of Scientific Computing} 2007\string; 31(1-2)\string: 61--73.

\bibitem{lehrenfeld2010hybrid}
Lehrenfeld C. Hybrid Discontinuous {Galerkin} methods for solving
  incompressible flow problems. {\it Rheinisch-Westfalischen Technischen
  Hochschule Aachen} 2010.

\bibitem{LS_CMAME_2016}
Lehrenfeld C, Sch\"{o}berl J. High order exactly divergence-free Hybrid
  Discontinuous Galerkin Methods for unsteady incompressible flows. {\it
  Computer Methods in Applied Mechanics and Engineering} 2016\string;
  307\string: 339 -- 361.
\newblock \href {\doibase http://dx.doi.org/10.1016/j.cma.2016.04.025} {doi:
  http://dx.doi.org/10.1016/j.cma.2016.04.025}

\bibitem{LedererSchoeberl2017}
Lederer PL, Sch\"oberl J. Polynomial robust stability analysis for
  $H$(div)-conforming finite elements for the Stokes equations. {\it IMA
  Journal of Numerical Analysis} 2017\string: drx051.
\newblock \href {\doibase 10.1093/imanum/drx051} {doi: 10.1093/imanum/drx051}

\bibitem{LLS_SIAM_2017}
Lederer PL, Lehrenfeld C, Sch{\"o}berl J. Hybrid Discontinuous {Galerkin}
  methods with relaxed {H(div)}-conformity for incompressible flows. Part I.
  {\it SIAM J. Numer. Anal.} 2018\string; 56\string: 2070--2094.
\newblock \href {\doibase 10.1137/17M1138078} {doi: 10.1137/17M1138078}

\bibitem{LLS_ESAIM_2019}
Lederer PL, Lehrenfeld C, Sch{\"o}berl J. Hybrid Discontinuous {Galerkin}
  methods with relaxed {H(div)}-conformity for incompressible flows. Part II.
  {\it ESAIM: M2AN} 2019\string; 53\string: 503--522.
\newblock \href {\doibase 10.1051/m2an/2018054} {doi: 10.1051/m2an/2018054}

\bibitem{SLLL_SEMA_2018}
Schroeder PW, Linke A, Lehrenfeld C, Lube G. Towards computable flows and
  robust estimates for inf-sup stable {FEM} applied to the time-dependent
  incompressible {Navier-Stokes} equations. {\it SeMA Journal} 2018\string:
  1--25.
\newblock \href {\doibase 10.1007/s40324-018-0157-1} {doi:
  10.1007/s40324-018-0157-1}

\bibitem{SJLLLS_CAMWA_2019}
Schroeder PW, John V, Lederer PL, Lehrenfeld C, Lube G, Sch\"oberl J. On
  reference solutions and the sensitivity of the 2D Kelvin--Helmholtz
  instability problem. {\it Computers \& Mathematics with Applications}
  2019\string; 77(4)\string: 1010--1028.

\bibitem{Gurtin1975}
Gurtin ME, Ian~Murdoch A. A continuum theory of elastic material surfaces. {\it
  Archive for Rational Mechanics and Analysis} 1975\string; 57(4)\string:
  291--323.

\bibitem{MR3614501}
Koba H, Liu C, Giga Y. Energetic variational approaches for incompressible
  fluid systems on an evolving surface. {\it Quart. Appl. Math.} 2017\string;
  75(2)\string: 359--389.
\newblock \href {\doibase 10.1090/qam/1452} {doi: 10.1090/qam/1452}

\bibitem{CASTRO2016241}
Castro DA, Devloo PR, Farias AM, Gomes SM, Durán O. Hierarchical high order
  finite element bases for H(div) spaces based on curved meshes for
  two-dimensional regions or manifolds. {\it Journal of Computational and
  Applied Mathematics} 2016\string; 301\string: 241 - 258.
\newblock \href {\doibase https://doi.org/10.1016/j.cam.2016.01.053} {doi:
  https://doi.org/10.1016/j.cam.2016.01.053}

\bibitem{girault2012finite}
Girault V, Raviart PA. {\it Finite element methods for Navier-Stokes equations:
  theory and algorithms}. 5.
\newblock Springer Science \& Business Media .
\newblock 2012.

\bibitem{brezzi1985two}
Brezzi F, Douglas~Jr. J, Marini LD. Two families of mixed finite elements for
  second order elliptic problems. {\it Numerische Mathematik} 1985\string;
  47(2)\string: 217--235.

\bibitem{brezzi2012mixed}
Boffi D, Brezzi F, Fortin M. {\it Mixed Finite Element Methods and
  Applications}.
\newblock Springer Science \& Business Media .
\newblock 2013.

\bibitem{arnold2002unified}
Arnold DN, Brezzi F, Cockburn B, Marini LD. Unified analysis of discontinuous
  {Galerkin} methods for elliptic problems. {\it SIAM journal on numerical
  analysis} 2002\string; 39(5)\string: 1749--1779.

\bibitem{braess}
Braess D. {\it Finite Elemente - Theorie, schnelle L{\"o}ser und Anwendungen in
  der Elastizit{\"a}tstheorie.}
\newblock Springer .
\newblock 2013.

\bibitem{cockburn2005locally}
Cockburn B, Kanschat G, Sch{\"o}tzau D. A locally conservative {LDG} method for
  the incompressible {Navier}-{Stokes} equations. {\it Mathematics of
  Computation} 2005\string; 74(251)\string: 1067--1095.

\bibitem{hesthaven2007nodal}
Hesthaven JS, Warburton T. {\it Nodal discontinuous {Galerkin} methods:
  algorithms, analysis, and applications}.
\newblock Springer Science \& Business Media .
\newblock 2007.

\bibitem{ascher1995implicit}
Ascher UM, Ruuth SJ, Wetton BT. Implicit-explicit methods for time-dependent
  partial differential equations. {\it SIAM Journal on Numerical Analysis}
  1995\string; 32(3)\string: 797--823.

\bibitem{ascher1997implicit}
Ascher UM, Ruuth SJ, Spiteri RJ. Implicit-explicit {Runge-Kutta} methods for
  time-dependent partial differential equations. {\it Applied Numerical
  Mathematics} 1997\string; 25(2)\string: 151--167.

\bibitem{fehn2019}
Fehn N, Kronbichler M, Lehrenfeld C, Lube G, Schroeder PW. High-order DG
  solvers for underresolved turbulent incompressible flows: A comparison of L2
  and H(div) methods. {\it International Journal for Numerical Methods in
  Fluids} 2019\string; 0(0).
\newblock \href {\doibase 10.1002/fld.4763} {doi: 10.1002/fld.4763}

\bibitem{linke2014role}
Linke A. On the role of the {Helmholtz} decomposition in mixed methods for
  incompressible flows and a new variational crime. {\it Computer Methods in
  Applied Mechanics and Engineering} 2014\string; 268\string: 782--800.

\bibitem{blms:2015}
Brennecke C, Linke A, Merdon C, Sch{\"o}berl J. Optimal and
  pressure-independent {$L^2$} velocity error estimates for a modified
  {C}rouzeix-{R}aviart {S}tokes element with {BDM} reconstructions. {\it J.
  Comput. Math.} 2015\string; 33(2)\string: 191--208.

\bibitem{lmt:2016}
Linke A, Matthies G, Tobiska L. Robust Arbitrary Order Mixed Finite Element
  Methods for the Incompressible {S}tokes Equations with pressure independent
  velocity errors. {\it ESAIM: M2AN} 2016\string; 50(1)\string: 289-309.

\bibitem{2016arXiv160903701L}
Lederer PL, Linke A, Merdon C, Sch\"oberl J. Divergence-free {R}econstruction
  {O}perators for {P}ressure-{R}obust {S}tokes {D}iscretizations with
  {C}ontinuous {P}ressure {F}inite {E}lements. {\it SIAM J. Numer. Anal.}
  2017\string; 55(3)\string: 1291--1314.
\newblock \href {\doibase 10.1137/16M1089964} {doi: 10.1137/16M1089964}

\bibitem{MR3683678}
John V, Linke A, Merdon C, Neilan M, Rebholz LG. On the divergence constraint
  in mixed finite element methods for incompressible flows. {\it SIAM Rev.}
  2017\string; 59(3)\string: 492--544.
\newblock \href {\doibase 10.1137/15M1047696} {doi: 10.1137/15M1047696}

\bibitem{MR3564690}
Linke A, Merdon C. Pressure-robustness and discrete {H}elmholtz projectors in
  mixed finite element methods for the incompressible {N}avier-{S}tokes
  equations. {\it Comput. Methods Appl. Mech. Engrg.} 2016\string; 311\string:
  304--326.
\newblock \href {\doibase 10.1016/j.cma.2016.08.018} {doi:
  10.1016/j.cma.2016.08.018}

\bibitem{MR3824769}
Ahmed N, Linke A, Merdon C. Towards pressure-robust mixed methods for the
  incompressible {N}avier-{S}tokes equations. {\it Comput. Methods Appl. Math.}
  2018\string; 18(3)\string: 353--372.
\newblock \href {\doibase 10.1515/cmam-2017-0047} {doi: 10.1515/cmam-2017-0047}

\bibitem{MR3875918}
Schroeder PW, Lehrenfeld C, Linke A, Lube G. Towards computable flows and
  robust estimates for inf-sup stable {FEM} applied to the time-dependent
  incompressible {N}avier-{S}tokes equations. {\it SeMA J.} 2018\string;
  75(4)\string: 629--653.
\newblock \href {\doibase 10.1007/s40324-018-0157-1} {doi:
  10.1007/s40324-018-0157-1}

\bibitem{gauger2019}
Gauger NR, Linke A, Schroeder PW. On high-order pressure-robust space
  discretisations, their advantages for incompressible high Reynolds number
  generalised Beltrami flows and beyond. {\it SMAI-Journal of Computational
  Mathematics} 2019\string; 5(1)\string: 89--129.

\bibitem{egger2009hybrid}
Egger H, Sch{\"o}berl J. A hybrid mixed discontinuous Galerkin finite-element
  method for convection--diffusion problems. {\it IMA Journal of Numerical
  Analysis} 2009\string; 30(4)\string: 1206--1234.

\bibitem{cockburn2009unified}
Cockburn B, Gopalakrishnan J, Lazarov R. Unified hybridization of discontinuous
  {Galerkin}, mixed, and continuous {Galerkin} methods for second order
  elliptic problems. {\it SIAM Journal on Numerical Analysis} 2009\string;
  47(2)\string: 1319--1365.

\bibitem{netgen}
Sch{\"o}berl J. {NETGEN An advancing front 2D/3D-mesh generator based on
  abstract rules}. {\it Computing and Visualization in Science} 1997\string;
  1(1)\string: 41--52.

\bibitem{schoeberl2014cpp11}
Sch\"oberl J. C++11 Implementation of Finite Elements in {NGSolve}. Tech. Rep.
  ASC-2014-30, Institute for Analysis and Scientific Computing; Karlsplatz 13,
  1040 Vienna, Austria:   2014.

\bibitem{GLS_MCS_2019}
Gopalakrishnan J, Lederer PL, Sch\"oberl J. {A mass conserving mixed stress
  formulation for the Stokes equations}. {\it IMA Journal of Numerical
  Analysis} 2019.
\newblock \href {\doibase 10.1093/imanum/drz022} {doi: 10.1093/imanum/drz022}

\bibitem{schafer1996benchmark}
Sch{\"a}fer M, Turek S, Durst F, Krause E, Rannacher R. Benchmark computations
  of laminar flow around a cylinder. {\it Flow simulation with high-performance
  computers II} 1996\string: 547--566.

\end{thebibliography}


\end{document}